\newif\ifthesis
\newif\iffull
\newif\ifea
\newif\ifspringer

\fulltrue
\thesisfalse
\eafalse
\springerfalse

\ifthesis
\documentclass{ucalgarythesis}
\else\ifspringer
		
		\RequirePackage{fix-cm}
		\documentclass[twocolumn]{svjour3}          
		\smartqed  
	\else
		\documentclass[11pt,letterpaper]{article}
		\usepackage[left=1in, right=1in, top=1in, bottom=1in]{geometry}
	\fi
\fi

\ifspringer
\else
	\usepackage{amsthm}
\fi

\usepackage{cite, pgf, url, hyperref, color, blindtext, microtype, amsmath, amssymb, mathtools, setspace, tikz, float, mathrsfs,wrapfig,graphicx}

\usepackage[ruled,linesnumbered,noresetcount]{algorithm2e}

\ifthesis
\usepackage{titlesec}
\titleformat{\chapter}
   {\normalfont\Large\bfseries\centering}{\thechapter.}{1em}{}
\titleformat{\section}
   {\normalfont\large\bfseries}{\thesection}{1em}{}
\titlespacing*{\chapter}{0pt}{-28pt}{40pt}
\fi

\let\oldnl\nl
\newcommand{\nonl}{\renewcommand{\nl}{\let\nl\oldnl}}

\makeatletter
\renewcommand{\@algocf@capt@plain}{above}
\makeatother

\usetikzlibrary{arrows,automata}
\usepackage[latin1]{inputenc}

\usepackage{enumitem}
\setlist{noitemsep}

\usepackage{stackengine}
\usepackage{scalerel}
\usepackage{xcolor}

\tikzset{
  linpt/.style={fill,rectangle,yscale=0.085cm,xscale=0.012cm}
}

\ifspringer
\newtheorem{observation}[theorem]{Observation}
\else
\newtheorem{theorem}{Theorem}
\newtheorem{lemma}[theorem]{Lemma}
\newtheorem{observation}[theorem]{Observation}

\newtheorem{definition}[theorem]{Definition}
\fi

\usepackage[textsize=tiny,disable=true]{todonotes}

\begin{document}
\SetEndCharOfAlgoLine{}


\ifthesis
\else\ifspringer
		\title{Strongly Linearizable Implementations of Fundamental Primitives\thanks{We acknowledge the support of the Natural Sciences and Engineering Research Council of Canada (NSERC). This research was undertaken, in part, thanks to funding from the Canada Research Chairs program.}
		}
		
		
		\author{Sean Ovens         \and
		        Philipp Woelfel
		}
		
		
		\institute{S. Ovens \at
					\email{sgovens@ucalgary.ca}
					\and
				   P. Woelfel \at
				   	\email{woelfel@ucalgary.ca}	\\ \\			   	
		              University of Calgary, 2500 University Drive NW, Calgary AB
		}
		
		\date{Received: date / Accepted: date}

		\maketitle
	\else
		\title{Strongly Linearizable Implementations\\ of Snapshots and Other Types}
		\date{}
		\author{Sean Ovens, sgovens@ucalgary.ca, University of Calgary
		\and Philipp Woelfel, woelfel@ucalgary.ca, University of Calgary}
	
		\maketitle
	\fi
\fi

\ifthesis
\begin{thesisabstract}
\else
\begin{abstract}
\fi
Linearizability is the gold standard of correctness conditions for shared memory algorithms, and historically has been considered the practical equivalent of atomicity. However, it has been shown~\cite{stronglin} that replacing atomic objects with linearizable implementations can affect the probability distribution of execution outcomes in randomized algorithms. Thus, linearizable objects are not always suitable replacements for atomic objects. A stricter correctness condition called strong linearizability has been developed and shown to be appropriate for randomized algorithms in a strong adaptive adversary model~\cite{stronglin}.

We devise several new lock-free strongly linearizable implementations from atomic registers. 
In particular, we give the first strongly linearizable lock-free snapshot implementation that uses bounded space.
This improves on the unbounded space solution of Denysyuk and Woelfel~\cite{WaitVsLock}. 
As a building block, our algorithm uses a lock-free strongly linearizable ABA-detecting register. 
We obtain this object by modifying the wait-free linearizable ABA-detecting register of Aghazadeh and Woelfel~\cite{ABAreg}, which, as we show, is not strongly linearizable. 

Aspnes and Herlihy~\cite{generalwaitfree} identified a wide class types that have wait-free linearizable implementations from atomic registers.
These types require that any pair of operations either commute, or one overwrites the other.
Aspnes and Herlihy gave a general wait-free linearizable implementation of such types, employing a wait-free linearizable snapshot object.
Replacing that snapshot object with our lock-free strongly linearizable one, we prove that all types in this class have a lock-free strongly linearizable implementation from atomic registers.
\ifthesis
\end{thesisabstract}
\else
\end{abstract}
\fi
\todo{double-check overfill boxes}
\todo{make sure no occurrences of words ``thesis'' or ``chapter'' or ``??''}

\ifthesis
\else
\section{Introduction}\label{chapter:intro}
\fi



In general, correctness properties for concurrent objects are defined by the sequential behaviours they preserve.
That is, overlapping operations on concurrent objects are expected to respond as they would in some sequential execution on the object.
Linearizability, a particularly popular correctness condition, requires that concurrent executions correspond to sequential histories that preserve the real-time order of operations.
Intuitively, operations on a linearizable implementation appear to take effect (i.e. linearize) at some atomic step between their invocation and response.
Linearizable implementations adequately preserve the sequential behaviour of their atomic counterparts; that is, any execution of a linearizable implementation ``appears'' to be a sequential execution of atomic operations.

However, some subtle guarantees are lost with linearizable object implementations. For instance, it has been shown that the probability distribution over the execution outcomes of a randomized algorithm can change when atomic objects are replaced with linearizable implementations \cite{stronglin}.
More generally, linearizability does not preserve any property that cannot be expressed as a set of allowable sequences of operations \cite{mclean:tracesets}.
To address these shortcomings, Golab, Higham, and Woelfel~\cite{stronglin} defined the notion of strong linearizability, and showed that replacing atomic objects with strongly linearizable implementations does not change the probability distribution of the algorithm's outcome under a \emph{strong adversary}.
A strong adversary has the power to schedule executions with complete knowledge of the system state, including all previous random operations (coin flips).
Strong linearizability requires that, once an operation has linearized, its position in the linearization order does not change in the future.
That is, operations cannot be retroactively inserted into the linearization order.
Strongly linearizable objects can simplify randomized algorithm design, which provides motivation for developing strongly linearizable implementations or proving their non-existence.

\ifthesis
\else
\subsection{Related Work}\label{sec:relwork}
\fi

The notion of strong linearizability was originally introduced by Golab, Higham, and Woelfel~\cite{stronglin}. Strongly linearizable implementations are not only linearizable, but they also exhibit the ``prefix-preservation'' property. 
This extra condition ensures that the linearization order of operations
does not change retroactively.
Such implementations are resilient against strong adaptive adversaries, which have the power to schedule executions with complete knowledge of the system state, including all previous random operations (coin flips).
In fact, it is known that strong linearizability is necessary to curtail the power of the strong adversary to influence the probability distribution of execution outcomes \cite{stronglin}.
Importantly, strong linearizability, like traditional linearizability, is a local property. Roughly speaking, a correctness property is local if the system satisfies the property provided that each object in the system satisfies the property \cite{MultiProcProgramming}. Hence, since strong linearizability is local, if every object in the set $\{O_1, \ldots, O_k\}$ is strongly linearizable, then any execution obtained by performing operations on any subset of $\{O_1, \ldots O_k\}$ is also strongly linearizable \cite{stronglin,AE2019a}. Strong linearizability is also composable, meaning a strongly linearizable implementation $O$ that uses an atomic base object $B$ of type $\mathscr{T}$ remains strongly linearizable when $B$ is replaced by $B'$, where $B'$ is a strongly linearizable implementation of $\mathscr{T}$ \cite{stronglin, AE2019a}. Locality and composability can simplify distributed algorithm design, as we will see in Section~\ref{sectionSLSS} and Section~\ref{sectionGEN} of this paper.

More generally, Attiya and Enea recently showed \cite{AE2019a} that strong linearizability is a specific form of strong observational refinement.
Traditional refinement \cite{Lynch:Simulations:1994} is a relationship between concrete objects and their specifications; that is, an object $O_1$ (i.e. a concrete object) refines an object $O_2$ (i.e. a specification) if the set of traces (i.e. sequences of possible actions) of $O_1$ is a subset of the set of traces of $O_2$.
An object $O_1$ is said to \emph{observationally} refine $O_2$ if every observation that can be made by a program using $O_1$ (i.e. effects of operation calls on $O_1$ that are observable by the program) could also be made by the same program using $O_2$ instead of $O_1$.
Finally, $O_1$ \emph{strongly} observationally refines $O_2$ if, for every schedule of a program that uses $O_1$, there exists a schedule for the same program that uses $O_2$ instead of $O_1$, such that the program makes precisely the same observations in both scenarios.
It has been shown \cite{Bouajjani:Observation:2015} that observational refinement and linearizability are equivalent when the specification (i.e. $O_2$) is atomic.
Similarly, strong observational refinement is equivalent to strong linearizability when the specification is atomic \cite{AE2019a}.
Strong observational refinements are designed to preserve hyperproperties, which are sets of sets of sequences of program observations.
Hyperproperties are a generalization of the notion of ``execution outcomes'', mentioned previously.
It is known that refinement (and observational refinement) preserves trace properties, which are sets of sequences of actions, but it does not preserve hyperproperties in general \cite{AE2019a}.
This provides further insight into the deficiencies of linearizability described by Golab, Higham, and Woelfel \cite{stronglin}.
Since strong observational refinements preserve the hyperproperties satisfied by their specifications, this implies that strongly linearizable implementations preserve the hyperproperties satisfied by their atomic counterparts.

Unless otherwise noted we assume the standard asynchronous shared memory system, where $n$ processes with unique IDs in $\{1,\dots,n\}$ communicate through atomic read and write operations on shared (multi-reader multi-writer) registers.
Almost all prior work on strong linearizability has focused on this model.
However, it is known that standard wait-free universal constructions (e.g. \cite{uniConstSL}) using $n$-process consensus objects are also strongly linearizable \cite{stronglin}. Therefore, there exists a strongly linearizable implementation of any type using atomic compare-and-swap objects, for example.
On the other hand, Attiya, Casta\~{n}eda, and Hendler \cite{ACH2018a} have shown that any wait-free strongly linearizable implementation of a queue or a stack for $n$ processes along with atomic registers can solve $n$-process consensus. Hence, any $n$-process wait-free strongly linearizable implementation of a queue or stack cannot be implemented from base objects with consensus number less than $n$.

Early results on strong linearizability have largely been negative (i.e. impossibility results).
Helmi, Higham, and Woelfel~\cite{PossImposs} have shown that essentially no non-trivial object has a deterministic wait-free strongly linearizable implementation from \emph{single-writer} registers.
Denysyuk and Woelfel~\cite{WaitVsLock} showed that for several fundamental types, including single-writer snapshots (defined below), counters, and unbounded max-registers, there exist no strongly linearizable wait-free implementations, even from multi-writer registers.

While many published results on strong linearizability (especially for lock/wait-free implementations) are discouraging, some fundamental types are known to have strongly linearizable implementations.
For instance, Helmi, Higham, and Woelfel~\cite{PossImposs} describe a strongly linearizable wait-free implementation of a bounded max-register from multi-reader multi-writer registers \cite{maxRegisters}. A simple modification of this algorithm, which we describe in more detail in Section~\ref{sec:unboundedslss}, results in a strongly linearizable lock-free implementation of an unbounded max-register from unbounded multi-reader multi-writer registers.
Helmi, Higham, Woelfel~\cite{PossImposs} also demonstrate that there is a strongly linearizable obstruction-free implementation of a consensus object from multi-reader single-writer registers. This implies that any type has a strongly linearizable obstruction-free implementation from multi-reader single-writer registers.
Denysyuk and Woelfel~\cite{WaitVsLock} have shown that there exists a universal lock-free strongly linearizable construction for versioned objects, which store version numbers that increase with each atomic update operation (a more detailed explanation of this construction, and of versioned objects, is provided in Section~\ref{sec:unboundedslss}).
This construction uses the unbounded modification of the max-register from \cite{PossImposs}, and therefore requires an unbounded number of registers.
The algorithm inherently requires unbounded space, since the version number of the object must increase with each update.

The snapshot type \cite{snapshot} is a fundamental primitive in distributed algorithm design \cite{generalwaitfree,Abrahamson,afek1999instancy,borowsky1993immediate,gawlick1992concurrent,herlihy1993asynchronous}.
In this paper, we consider only single-writer snapshots, which contain an $n$-component vector of values. For any $p \in \{1, \ldots, n\}$, the $p$-th component of a single-writer snapshot is writable only by process $p$. An $update_p(x)$ invocation by process $p$ changes the contents of the $p$-th component of the snapshot object to $x$.
The snapshot type also supports a $scan$ invocation, which returns the entire stored vector.
That is, the $scan$ invocation allows processes to obtain a consistent view of multiple single-writer memory cells; if a $scan$ invocation returns a vector $V$, then the snapshot object must have contained exactly $V$ at some point in its execution interval.
There are many wait-free linearizable implementations of the snapshot type from registers \cite{snapshot,aspnes2013snapshot,attiya1995snapshot,attiya1998snapshot,inoue1994snapshot}, but due to the results of Denysyuk and Woelfel~\cite{WaitVsLock}, it is known that none of these implementations are strongly linearizable.

\ifthesis
\else
\subsection{Results}
\fi

An ABA-detecting register stores a single value from some domain $D$, and supports $DWrite$ and $DRead$ invocations.
A $DWrite(x)$ invocation writes value $x\in D$, and a $DRead$ invocation returns the latest written value together with a Boolean flag.
This flag indicates whether there has been a $DWrite$ operation since the previous $DRead$ by the same process.  
This type was originally defined by Aghazadeh and Woelfel \cite{ABAreg}, who also gave a wait-free linearizable implementation from $O(n)$ bounded registers.
ABA-detecting registers are used to combat the ABA problem, which occurs when two $DRead$ operations with interleaving $DWrite$ operations return the same value; in this scenario the reading process is unable to distinguish between the actual execution and a different execution in which no $DWrite$ operations occur between the two $DRead$ operations.
In Section~\ref{slabasection} we show that Aghazadeh and Woelfel's wait-free linearizable ABA-detecting register is not strongly linearizable, and modify it to achieve strong linearizability.
Our implementation sacrifices wait-freedom for lock-freedom.
\begin{theorem}\label{thm:main-ABA}
  There is a lock-free strongly linearizable implementation of an ABA-detecting register from $O(n)$ registers of size $O(\log n+\log|D|)$.
\end{theorem}
The \emph{amortized} step complexity of our implementation is $O(n)$ (an object has amortized step complexity $k$, if all processes combined execute at most $k\cdot \ell$ steps in any execution that comprises $\ell$ operation invocations).
Moreover, each $DWrite$ needs only $O(1)$ steps, and each $DRead$ has constant step complexity in the absence of contention.

As mentioned previously, Denysyuk and Woelfel \cite{PossImposs} sacrificed wait-freedom for lock-freedom to obtain a strongly linearizable snapshot implementation using an unbounded number of registers.
In Section~\ref{sectionSLSS} we give the first such implementation that needs only bounded space.
\begin{theorem}\label{thm:main-snapshot}
  There is a lock-free strongly linearizable implementation of a snapshot object from $O(n)$ registers of size $O(\log n+\log |D|)$.
\end{theorem}
We provide an analysis of our implementation, showing that the amortized step complexity is $O(n^3)$.
Our algorithm uses as base objects a linearizable snapshot object $S$, so the step complexity heavily depends on the implementation of $S$.
But in the absence of contention, each $update$ and $scan$ operation of the strongly linearizable snapshot needs only a constant number of operations on $S$ in addition to a constant number of register accesses.

Aspnes and Herlihy \cite{generalwaitfree} defined a large class of types with wait-free linearizable implementations.
Any two operations of this type must either commute (meaning the system configuration obtained after both operations have been executed consecutively is independent of the order of the two operations), or one operation overwrites the other (meaning that the system configuration obtained after the overwriting operation has been performed is not affected by whether or not the other operation is executed immediately before it).
In this paper we refer to such types as \emph{simple} types.
Aspnes and Herlihy \cite{generalwaitfree} describe a general wait-free construction of any simple type.
Their algorithm uses an atomic snapshot object, which may be replaced by a linearizable implementation to obtain a wait-free linearizable implementation of any simple type from registers.
In Section~\ref{sectionGEN} we prove that Aspnes and Herlihy's construction is also strongly linearizable.
Combining this with Theorem~\ref{thm:main-snapshot}, and using the composability of strong linearizability, we obtain the following:
\begin{theorem}\label{thm:main-general-construction}
  Any simple type has a lock-free strongly linearizable implementation from registers.
\end{theorem}
Aspnes and Herlihy introduce the notion of \emph{linearization graphs}, which are directed acyclic graphs whose nodes are operations (we will define these structures more formally in
\ifthesis
\else
Section~\ref{sectionGEN}).
\fi
Aspnes and Herlihy define a linearization function based on topological orderings of these linearization graphs.
However, since operations may be written to the ``middle'' of a linearization graph (i.e. an operation might have outgoing edges immediately as it is written to the graph), this linearization function is not prefix-preserving.
Hence, even though we do not modify the algorithm beyond the snapshot object replacement, the proof of strong linearizability is involved.

\ifthesis
\else
\section{Preliminaries}\label{defsection}
\fi

We consider an asynchronous shared memory system with $n$ processes, each of which has a unique identifier in $\{1, \dots, n\}$.
\iffull
More precisely, processes communicate by performing operations on shared \emph{objects}, which are each an instance of a \emph{type}.
A type may be defined as a state machine. That is, if $\mathscr{T}$ is a type, then $\mathscr{T} = (\mathcal{S}, s_0, \mathcal{O}, \mathcal{R}, \delta)$, where $\mathcal{S}$ is a set of \emph{states}, $s_0 \in \mathcal{S}$ is an \emph{initial state}, $\mathcal{O}$ is a set of \emph{invocation descriptions}, $\mathcal{R}$ is a set of \emph{responses}, and $\delta : \mathcal{S} \times \mathcal{O} \rightarrow \mathcal{S} \times \mathcal{R}$ is a \emph{transition function}. An invocation description $invoke \in \mathcal{O}$ applied to an object of type $\mathscr{T}$ in state $s \in \mathcal{S}$ returns a value $resp \in \mathcal{R}$ and causes the object to enter state $s'$, where $\delta(invoke, s) = (resp, s')$. Throughout this paper, we only consider types for which $\delta(invoke, s)$ is defined for every invocation description $invoke$ and for every state $s$. An invocation consists of a name, a set of arguments, and a process identifier (indicating the process that performs the invocation).
The \emph{sequential specification} of $\mathscr{T}$ is the set of allowable sequences of invocation/response pairs. We may define the sequential specification of $\mathscr{T}$ inductively;\todo{refine this definition} that is, a sequence $(invoke_1, resp_1), \ldots, (invoke_k, resp_k)$ is in the sequential specification of $\mathscr{T}$ if either (1) it is empty, or (2) $(invoke_1, resp_1), \ldots, (invoke_{k-1}, resp_{k-1})$ is in the sequential specification of $\mathscr{T}$, and $\delta(invoke_k, s_k) = (resp_k, s)$, for some state $s$, where $s_k$ is the state reached after applying the invocations $invoke_1, \ldots, invoke_{k-1}$ in order.
We define types using a descriptive approach; that is, instead of explicitly providing an automaton, we describe the values that an object stores, the invocation descriptions it supports, and how these invocation descriptions change the stored values in sequential executions.
We present invocation descriptions using pseudocode; for instance, $op_p(x, y)$ denotes an invocation description named $op$, with $x$ and $y$ as arguments, and with $p$ as the associated process identifier.
Objects that are provided by the system are called \emph{base objects}.
\fi
An operation $op$ consists of an \emph{invocation event}, denoted $inv(op)$, and possibly a \emph{response event}, denoted $rsp(op)$. Processes perform operations sequentially.
\iffull
An invocation event is a tuple $(O, M, id)$, where $O$ is an object instance, $M$ is a invocation description, and $id$ is a unique integer that identifies the invocation event. A response event is a pair $(r, id)$, where $r$ is a return value, and $id$ is an integer. An invocation event $(O, M, id_i)$ \emph{matches} a response event $(r, id_j)$ (and vice versa) if and only if $id_i = id_j$. 
\fi
A \emph{transcript} is a sequence of \emph{steps}, each of which is either an invocation event or a response event of some operation. If $T$ and $U$ are transcripts, we use $T \circ U$ to denote the concatenation of $T$ and $U$.

A \emph{projection} of a transcript $T$ onto an object $O$, denoted $T|O$, is the sequence of steps in $T$ that are performed on $O$. Similarly, a projection of a transcript $T$ onto a process $p$, denoted $T|p$, is the sequence of invocation and response events by process $p$.
\iffull
We say $e \in T$, for some invocation or response event $e$ and some transcript $T$, if $e$ is a member of the sequence defined by $T$. As a shorthand, we say $op \in T$, for an operation $op$ and a transcript $T$, if $inv(op) \in T$. 
\fi
An operation $op$ is \emph{pending} in some transcript $T$ if $T$ contains its invocation but no matching response. If an operation $op \in T$ is not pending in $T$, then it is \emph{complete}. A transcript $T$ is \emph{complete} if, for every operation $op \in T$, $op$ is complete. Otherwise, $T$ is \emph{incomplete}. An operation $op$ is \emph{atomic} in a transcript $T$ if $inv(op)$ is immediately followed by $rsp(op)$ in $T$. A transcript $T$ is \emph{well-formed} if $T$ is empty, or for every $p \in \{1, \ldots, n\}$, $T|p = i \circ T_1 \circ r \circ T_2$, where $i$ is some invocation event, $T_1$ and $T_2$ are well-formed transcripts, $T_1$ is complete, and $r$ is either a response event that matches $i$, or $r$ is empty. Throughout this paper, we assume all transcripts are well-formed.

An object $O$ is \emph{atomic} if every operation in any transcript on $O$ is atomic. 
\iffull
We say an object $O$ of type $\mathscr{T}$ is \emph{implemented} if each invocation description provided by $\mathscr{T}$ is associated with a \emph{method}, which is a sequence of invocations.
A process $p$ that executes an operation invocation $inv(op)$ on an implemented object $O$, such that the invocation description of $inv(op)$ is associated with a method $M$, sequentially executes each invocation step described by $M$.
\fi
During the execution of this method, other processes may also take steps, which interleave with the steps taken by $p$.

The order of operations in a transcript is a partial order, since some operations may overlap. Operation $op_1$ \emph{happens before} operation $op_2$ in transcript $T$, or $op_1 \xrightarrow{T} op_2$, if and only if $rsp(op_1)$ occurs before $inv(op_2)$ in $T$.
Operations $op_1, op_2 \in T$, for some transcript $T$, are \emph{concurrent} if $op_1$ does not happen before $op_2$, and $op_2$ does not happen before $op_1$.

A \emph{history} is a transcript such that, for every process $p$, every operation in $H|p$ is atomic. A history may be considered a sequence of ``high-level'' invocation and response events.  A \emph{completion} of a history $H$ is a complete history $H'$ that is constructed from $H$ as follows: for each pending operation $op$ in $H$, either a response for $op$ is appended to $H'$, or $op$ is removed from $H'$. A \textit{sequential history} is a history that contains no concurrent operations. Suppose $S = inv(op_1) \circ rsp(op_1) \circ \ldots \circ inv(op_k) \circ rsp(op_k)$ is a sequential history on an object $O$ of type $\mathscr{T}$, where $inv(op_i) = (O, invoke_i, id_i)$ and $rsp(op_i) = (O, resp_i, id_i)$ for every $i \in \{1, \ldots, k\}$. Then $S$ is \emph{valid} if and only if $(invoke_1, resp_1), \ldots, (invoke_k, resp_k)$ is in the sequential specification of $\mathscr{T}$. If $T$ is a transcript, the \emph{interpreted history} $\Gamma(T)$ consists of all ``high-level'' steps that exist in $T$. That is, $\Gamma(T)$ can be constructed by removing, for every process $p$, every step that appears after $inv(op)$ and not after $rsp(op)$, for any operation $op \in T$ with process identifier $p$. If $\mathcal{T}$ is a set of transcripts, then $\Gamma(\mathcal{T}) = \{\Gamma(T)\;:\;T \in \mathcal{T}\}$.

Linearizability was originally defined by Herlihy and Wing~\cite{linearizability}. The following definition is taken from the textbook by Herlihy and Shavit \cite{MultiProcProgramming}: a history $H$ is \textit{linearizable} if it has a completion $H'$ such that there is a sequential history $S$ with the following properties:

\begin{itemize}
	\item All operations in $H'$ are present in $S$, with identical invocations and responses;
	\item the sequential history $S$ is valid; and
	\item the happens-before order of operations in $S$ extends the happens-before order of operations in $H'$.
\end{itemize}

We call a sequential history $S$ that satisfies the above properties a \emph{linearization} of $H$. An implementation of an object is linearizable if every history in the set of possible histories on the object is linearizable. That is, if $O$ is some object and $\mathcal{H}$ is the set of possible histories on $O$, then $O$ is linearizable if and only if for all $H \in \mathcal{H}$, $H$ is linearizable. If $\mathcal{H}$ is a set of histories, and $f$ is a function such that, for every $H \in \mathcal{H}$, $f(H)$ is a linearization of $H$, then $f$ is called a \emph{linearization function} for the set $\mathcal{H}$.

We often refer to the \emph{time} at which particular steps are executed; this simply refers to the step's position in a transcript. That is, if $e$ is a step in a transcript $T$, then $time_T(e) = t$ if the $t$-th element of $T$ is $e$. Where $T$ is clear from context, we simply write $time(e) = t$. If $e$ is a step that is not present in a transcript $T$, then let $time_T(e) = \infty$. For the sake of brevity, if $am$ is an atomic operation in a transcript $T$, then we say $am$ \emph{happens} at $time(rsp(am))$. If $op$ is a complete operation whose implementation contains an atomic operation on line~$x$, we use $op^x$ to denote $rsp(am)$, where $am$ is the operation invoked by the final call to line~$x$ performed by $op$.

Another characterization of linearizability relies on the notion of \emph{linearization points}.
Let $O$ be a linearizable object.
Then, for any transcript $T$, a \emph{linearization point function} $pt$ for $O$ maps operations in $\Gamma(T|O)$ to points in time in $T$, such that
\begin{enumerate}[label=(\roman*)]
	\item for every operation $op \in \Gamma(T|O)$, \ifspringer \\ \else\fi $pt(op) \in \bigl[time_T(inv(op)), time_T(rsp(op))\bigr]$, and
	\item there exists a linearization $S$ of $\Gamma(T|O)$ such that for every operation $op \in \Gamma(T|O)$ such that $pt(op) \neq \infty$, $op \in S$, and for every pair of operations $op_1, op_2 \in S$, if $op_1 \xrightarrow{S} op_2$ then $pt(op_1) \leq pt(op_2)$.
\end{enumerate}
Intuitively, a linearization point is a point in time between the invocation and response of an operation $op$ at which $op$ ``appears'' to take effect. In any transcript containing operations on a linearizable object $O$, each operation on $O$ can be assigned a linearization point between its invocation and response, such that the sequential history that results from ordering each operation on $O$ by these points is valid.

The \textit{prefix closure} of a set of transcripts $\mathcal{T}$, denoted $close(\mathcal{T})$, is the set of all transcripts $S$ such that there exists a transcript $T$ such that $S \circ T \in \mathcal{T}$. A \textit{strong linearization function} $f$ for a set of transcripts $\mathcal{T}$ has the following properties~\cite{stronglin}:
\begin{itemize}
\item The function $f$ is a linearization function for the set of histories $\Gamma(close(\mathcal{T}))$.
\item For any two transcripts $S, T \in \mathcal{T}$ such that $S$ is a prefix of $T$, $f(S)$ is a prefix of $f(T)$. That is, $f$ is \textit{prefix-preserving}.
\end{itemize}
An implementation $O$ of a type $\mathscr{T}$ is called \textit{strongly linearizable} if and only if the set of all transcripts on instances of $O$ has a strong linearization function.

\ifea
\else
Let $\mathcal{T}$ be the set of transcripts of an implementation $S$ of some type $\mathscr{T}$.
A \emph{continuation} of a transcript $T$ is a transcript $U$ such that $T \circ U$ is well-formed and $T \circ U \in \mathcal{T}$.
Then $S$ is \emph{wait-free} if, for every transcript $T \in \mathcal{T}$ and every pending operation $op \in \Gamma (T)$ by process $p$, every continuation $U$ of $T$ that contains an infinite number of steps by $p$ has a finite prefix $U'$ such that $rsp(op) \in U'$.
The implementation $S$ is \emph{lock-free} if, for every transcript $T \in \mathcal{T}$ such that $\Gamma (T)$ contains at least one pending operation, for every infinite continuation $U$ of $T$, there exists an $op \in \Gamma (T)$ that is pending in $T$ such that, for some finite prefix $U'$ of $U$, $rsp(op) \in U'$.
Intuitively, an implementation is wait-free if each pending operation by process $p$ responds within a finite number of steps by $p$, and an implementation is lock-free if some pending operation responds provided that some process takes sufficiently many steps.
\fi

\ifthesis
\else
\section{A Strongly Linearizable ABA-Detecting Register}\label{slabasection}
\fi
    An ABA-detecting register \cite{ABAreg} is a type that stores a single value $R$ from some domain $D$, and supports the invocation descriptions $DWrite_q(x)$ for $x\in D$, and $DRead_q()$ with the following sequential specification:
    Initially, $R=\bot\in D$, and a $DWrite_q(x)$ invocation changes the value of $R$ to $x$. 
    Invocation $DRead_q()$ returns a pair $(x,a)\in D\times\{true,false\}$, where $x$ is the value of $R$, and $a$ is $true$ if and only if $q$ performed an earlier $DRead_q()$ operation, and a $DWrite_p$ was performed by some process $p$ since $q$'s last $DRead_q()$.

\ifthesis
\else
\subsection{A Linearizable ABA-Detecting Register}
\fi
Aghazadeh and Woelfel~\cite{ABAreg} presented a wait-free linearizable ABA-detecting register, which is included here as Algorithm~\ref{linaba} for reference. For a detailed description of the algorithm and a proof of its linearizability, see \cite{zahrathesis, ABAreg}. This algorithm works by associating each write with a process identifier and a sequence number. Processes are also responsible for ``announcing'' the sequence numbers they read into a global array of single-writer registers. That is, if process $p$ reads process identifier $q$ and sequence number $s$, then it writes the pair $(q, s)$ to the $p$-th entry of the announcement array $A$; a writer does not use a sequence number if it is paired with their process identifier in this announcement array. A writer also does not use any sequence number that is present in a local queue, called $usedQ$ in Algorithm~\ref{linaba}, which is a queue of $n + 1$ values. This queue stores the previous $n + 1$ sequence numbers chosen by the writing process, and it initially contains $n + 1$ elements valued $\bot$. Finally, if a write occurs during a read operation, the reader sets a local flag. This flag is used to delegate the task of acknowledging the modification to the reading process' next read operation.

\begin{algorithm}
	\small
    \caption{A linearizable ABA-detecting register \cite{ABAreg}}
    \label{linaba}
    
    \SetKwProg{Fn}{Function}{:}{}
    
    \nonl\textbf{shared:}\;
    	\nonl\quad register $X = (\bot, \bot, \bot)$\;
    	\nonl\quad register $A[0 \ldots n-1] = ((\bot, \bot), \ldots , (\bot, \bot))$\;
    \nonl\;  
    \nonl\textbf{local} (to each process):\; 
    	\nonl\quad Boolean $b = False$\;
    	\nonl\quad Queue $usedQ[n+1] = (\bot, \ldots , \bot)$\;
    	\nonl\quad Set $na = \{\}$\;
    	\nonl\quad Integer $c = 0$\;
    
    \nonl\;
	\nonl\Fn{DWrite$_p(x)$} {
      $s \leftarrow GetSeq()$\; \label{getseq}
      $X.Write(x, p, s)$\;		\label{linwritelin}
    }
    
    \nonl\; 
    \nonl\Fn{GetSeq$_p()$} {
	  $(r, s_r) \leftarrow A[c].Read()$\;
      \If{$r = p$} {
        $na \leftarrow (na \setminus \{(c, i) \; | \; i \in \mathbb{N}\}) \cup (c, s_r)$\;
      }
      \Else {
        $na \leftarrow na \setminus \{(c, i) \; | \; i \in \mathbb{N}\}$\;
      }
      $c \leftarrow (c + 1)$ mod $n$\;
      choose arbitrary $s \in \bigl(\{0, \ldots , 2n+1\}$
      		$\setminus (\{i \; | \; (j, i) \in na\} \cup usedQ)\bigr)$\; \label{chooseseq}
      $usedQ.enq(s)$\;						\label{enqseq}
      $usedQ.deq()$\;						\label{deqseq}
      \Return $s$
    }
	
	\nonl\;
    \nonl\Fn{DRead$_q()$} {
      $(x, p, s) \leftarrow X.Read()$\;		\label{linfirstread}
      $(r, s_r) \leftarrow A[q].Read()$\;	\label{linannread}
      $A[q].Write(p, s)$\;					\label{linannounce}
      $(x', p', s') \leftarrow X.Read()$\;	\label{linsecondread}
      \If{$(p, s) = (r, s_r)$} {			\label{linretifstatement}
        $ret \leftarrow (x, b)$\;			\label{linretsetb}
      }
      \Else {
        $ret \leftarrow (x, True)$\;		\label{linretsettrue}
      }
      \If{$(x, p, s) = (x', p', s')$} {		\label{linflagifstatement}
        $b \leftarrow False$\;				\label{linfalseflag}
      }
      \Else {
        $b \leftarrow True$\;				\label{lintrueflag}
      }
      \Return $ret$
    }
\end{algorithm}

\begin{observation}\label{obslinabanotsl}
	The ABA-detecting register in Algorithm~\ref{linaba} is not strongly linearizable.
\end{observation}

Algorithm~\ref{linaba} is not strongly linearizable because the point at which a $DRead$ operation takes effect depends on whether or not a $DWrite$ operation executes line~\ref{linwritelin} between lines \ref{linfirstread} and \ref{linsecondread} of the $DRead$; that is, if a process sets its $b$ flag during a $DRead$ operation, then it must linearize on line~\ref{linfirstread}, since the following read will detect any $DWrite$ operations that occur between lines \ref{linfirstread} and \ref{linsecondread}. Conversely, if a process does not set its $b$ flag during a $DRead$ operation, then this operation is responsible for detecting any $DWrite$ operations that occur between lines \ref{linfirstread} and \ref{linsecondread}, and it must therefore linearize on line~\ref{linsecondread}. This behaviour allows a scheduler to insert a $DRead$ operation in front of $DWrite$ operations that have already taken effect.

\iffull
\begin{proof}[Proof of Observation~\ref{obslinabanotsl}]
Consider an execution of Algorithm~\ref{linaba} where process $p$ executes two $DRead_p$ operations $dr_1$ and $dr_2$, and process $q$ executes an infinite sequence of $DWrite_q(x)$ operations $dw_1, dw_2, \ldots$ for some value $x$. Let $dw_i, dw_{i+1}$ be two consecutive $DWrite_q(x)$ operations by $q$, and let $s$ be the sequence number chosen by $dw_i$ on line~\ref{chooseseq}. Since $q$ performs a $usedQ.enq(s)$ operation on line~\ref{enqseq} of $dw_i$, and $usedQ$ contains $n+2$ elements after this operation, $usedQ$ contains $s$ after the $usedQ.deq()$ operation performed by $q$ on line~\ref{deqseq} of $dw_i$. Following the $usedQ.deq()$ operation in $dw_i$, $usedQ$ is not modified again until line~\ref{enqseq} is executed by $q$ during $dw_{i+1}$. Hence, $usedQ$ contains $s$ while $q$ selects a sequence number on line~\ref{chooseseq} of $dw_{i+1}$. Thus, the sequence number chosen by $dw_{i+1}$ is different from $s$. We have shown that
\ifspringer
\begin{equation}
	\begin{split}
	&\text{no two consecutive $DWrite$ operations by $q$} \\ &\text{choose the same sequence number.}\label{distinct_seq}
	\end{split}
\end{equation}
\else
\begin{equation}
	\text{no two consecutive $DWrite$ operations by $q$ choose the same sequence number.}\label{distinct_seq}
\end{equation}
\fi

Additionally, since sequence numbers are chosen from a finite set of integers, in the infinite sequence of $DWrite$ operations by $q$,
\ifspringer
\begin{equation}
	\begin{split}
	&\text{there are distinct operations $dw_i$ and $dw_j$ that} \\ &\text{choose the same sequence number.}\label{repeat_seq}
	\end{split}
\end{equation}
\else
\begin{equation}
	\text{there are distict operations $dw_i$ and $dw_j$ that choose the same sequence number.}\label{repeat_seq}
\end{equation}
\fi

Let $dw_i, dw_j$ be two distinct $DWrite_q(x)$ operations in the infinite sequence performed by $q$ (assume $i < j$), both of which choose the same sequence number $s$. Note that (\ref{distinct_seq}) guarantees that $dw_{i+1}$ chooses a sequence number $s' \neq s$, and hence $dw_{i+1} \neq dw_j$. The following transcripts $S$, $T_1$, and $T_2$ are possible transcripts produced by the programs described for $p$ and $q$:
\ifspringer
\begin{align*}
	&S = dw_1 \circ \ldots \circ dw_i \circ \\ &\qquad (dr_1 \text{ to the end of line \ref{linannread}}) \circ dw_{i+1} \\
	&T_1 = S \circ dw_{i+2} \circ \ldots \circ dw_j \circ \\ &\qquad (dr_1 \text{ from line \ref{linannounce} to completion}) \circ dr_2 \\
	&T_2 = S \circ (dr_1 \text{ from line \ref{linannounce} to completion}) \circ dr_2
\end{align*}
\else
\begin{align*}
	&S = dw_1 \circ \ldots \circ dw_i \circ (dr_1 \text{ to the end of line \ref{linannread}}) \circ dw_{i+1} \\
	&T_1 = S \circ dw_{i+2} \circ \ldots \circ dw_j \circ (dr_1 \text{ from line \ref{linannounce} to completion}) \circ dr_2 \\
	&T_2 = S \circ (dr_1 \text{ from line \ref{linannounce} to completion}) \circ dr_2
\end{align*}
\fi
To derive a contradiction, assume that Algorithm~\ref{linaba} is strongly linearizable. Then there must be a strong linearization function for $\{S, T_1, T_2\}$; let $f$ be such a strong linearization function. Since $dw_i$ executes an $X.Write(x, q, s)$ operation on line~\ref{linwritelin}, and no later $X.Write$ operations occur before $dr_1$ executes the $X.Read()$ operation on line~\ref{linfirstread} during $S$,
\ifspringer
\begin{equation}
	\begin{split}
	&\text{the $X.Read()$ operation performed by $dr_1$ on} \\ &\text{line~\ref{linfirstread} in $S$ returns $(x, q, s)$.}\label{dr1_firstread}
	\end{split}
\end{equation}
\else
\begin{equation}
	\text{the $X.Read()$ operation performed by $dr_1$ on line~\ref{linfirstread} in $S$ returns $(x, q, s)$.}\label{dr1_firstread}
\end{equation}
\fi

Since $dw_1, \ldots, dw_i$ are performed in sequence, and because each of these operations respond before any other operation is invoked in all of the above transcripts, each of $f(S)$, $f(T_1)$, and $f(T_2)$ must begin with $dw_1 \circ \ldots \circ dw_i$.

Suppose $dr_1$ linearizes prior to $dw_{i+1}$ in $S$. That is, suppose 
\begin{equation}\tag{A-1}
	f(S) = dw_1 \circ \ldots \circ dw_i \circ dr_1 \circ dw_{i+1}.\label{suppose_linfirst}
\end{equation}

The following table summarizes the steps that affect the return value of $dr_2$ in $T_1$:

\begin{table}[H]
\centering
\begin{tabular}{ccll}
	\textbf{Line \#} & \textbf{Operation} & \textbf{Code Statement} & \textbf{Response} \\
	\ref{linwritelin}	& $dw_j$	& $X.Write(x, q, s)$	& --- \\
	\ref{linannounce}	& $dr_1$	& $A[p].Write(q, s)$	& --- \\
	\ref{linsecondread}	& $dr_1$	& $X.Read()$			& $(x, q, s)$ \\
	\ref{linfalseflag}	& $dr_1$	& $b \gets False$		& --- \\
	\ref{linfirstread}	& $dr_2$	& $X.Read()$			& $(x, q, s)$ \\
	\ref{linannread}	& $dr_2$	& $A[p].Read()$			& $(q, s)$ \\
	\ref{linsecondread}	& $dr_2$	& $X.Read()$			& $(x, q, s)$ \\
	\ref{linretsetb}	& $dr_2$	& $ret \gets (x, False)$& ---
\end{tabular}
\end{table}

Since $dr_1$ reads $(x, q, s)$ on line~\ref{linfirstread} and line~\ref{linsecondread}, $dr_1$ sets its $b$ flag to $False$ on line~\ref{linfalseflag} in $T_1$.
Since $dr_2$ reads $(x, q, s)$ on line~\ref{linfirstread} and line~\ref{linsecondread}, by the condition on line~\ref{linretifstatement} $dr_2$ executes $ret \gets (x, b)$ on line~\ref{linretsetb} in $T_1$. Hence,
\begin{equation}\tag{A-2}
	\text{$dr_2$ returns $(x, False)$ in $T_1$.}\label{dr2_retfalse}
\end{equation}

By (\ref{suppose_linfirst}), the fact that $dw_{i+2}, \ldots, dw_j$, and $dr_2$ are performed sequentially in $T_1$, and our supposition that $f$ is prefix-preserving,
\begin{equation}\tag{A-3}
	f(T_1) = dw_1 \circ \ldots \circ dw_i \circ dr_1 \circ dw_{i+1} \circ \ldots \circ dw_j \circ dr_2.\label{t1_linearization}
\end{equation}

However, the history in (\ref{t1_linearization}) is not valid, since there is at least one $DWrite$ operation between $dr_1$ and $dr_2$, but $dr_2$ returns $(x, False)$ by (\ref{dr2_retfalse}).

Now suppose that $dr_1$ does not linearize prior to $dw_{i+1}$ in $S$. That is, suppose
\ifspringer
\begin{equation}\tag{B-1}
	\begin{split}
	&\text{either $f(S) = dw_1 \circ \ldots \circ dw_{i+1}$ or} \\ &\text{$f(S) = dw_1 \circ \ldots \circ dw_{i+1} \circ dr_1$.}\label{suppose_linsecond}
	\end{split}
\end{equation}
\else
\begin{equation}\tag{B-1}
	\text{either $f(S) = dw_1 \circ \ldots \circ dw_{i+1}$ or $f(S) = dw_1 \circ \ldots \circ dw_{i+1} \circ dr_1$.}\label{suppose_linsecond}
\end{equation}
\fi

The following table summarizes the steps that affect the return value of $dr_2$ in $T_2$:

\begin{table}[H]
\centering
\begin{tabular}{ccll}
	\textbf{Line \#} & \textbf{Operation} & \textbf{Code Statement} & \textbf{Response} \\
	\ref{linwritelin}	& $dw_{i+1}$	& $X.Write(x, q, s')$	& --- \\
	\ref{linannounce}	& $dr_1$		& $A[p].Write(q, s)$	& --- \\
	\ref{linfirstread}	& $dr_2$		& $X.Read()$			& $(x, q, s')$ \\
	\ref{linannread}	& $dr_2$		& $A[p].Read()$			& $(q, s)$ \\
	\ref{linretsettrue} & $dr_2$		& $ret \gets (x, True)$	& ---
\end{tabular}
\end{table}

Due to (\ref{dr1_firstread}), $dr_1$ executes an $A[p].Write(q, s)$ operation on line~\ref{linannounce} in $T_2$.
Hence, the $A[p].Read()$ operation performed by $dr_2$ on line~\ref{linannread} in $T_2$ must return $(q, s)$.
Since the $X.Read()$ operation performed by $dr_2$ on line~\ref{linfirstread} returns $(x, q, s')$, and $s \neq s'$, $dr_2$ executes the $ret \gets (x, True)$ statement on line~\ref{linretsettrue} by the condition on line~\ref{linretifstatement} in $T_2$.
Therefore,
\begin{equation}\tag{B-2}
	\text{$dr_2$ returns $(x, True)$ in $T_2$.}\label{dr2_rettrue}
\end{equation}

By (\ref{suppose_linsecond}), the fact that $dr_1$ and $dr_2$ are performed sequentially, and our assumption that $f$ is prefix-preserving,
\begin{equation}\tag{B-3}
	f(T_2) = dw_1 \circ \ldots \circ dw_{i+1} \circ dr_1 \circ dr_2.\label{t2_linearization}
\end{equation}

However, the history in (\ref{t2_linearization}) is not valid, since there are no $DWrite$ operations between $dr_1$ and $dr_2$, but $dr_2$ returns $(x, True)$ by (\ref{dr2_rettrue}).

Thus, no strong linearization function can be defined over the set $\{S, T_1, T_2\}$, which proves the observation.
\end{proof}
\fi

\ifthesis
\else
\subsection{Making the Algorithm Strongly Linearizable}
\fi
Algorithm~\ref{linaba} can be modified in order to make the implementation strongly linearizable. Our modification to the $DRead$ method of the linearizable ABA-detecting register is provided in Algorithm~\ref{slaba}.
The $GetSeq$ and $DWrite_p$ methods are the same as in \cite{zahrathesis}.
Our strategy is to ``stretch'' $DRead$ operations until a period of quiescence is observed by the reading process. 
Our new $DRead$ method performs the same sequence of reads as in Algorithm~\ref{linaba}; however, each $DRead$ is now responsible for acknowledging all concurrent $DWrite$ operations, rather than delegating this task to the next $DRead$ by the same process.
Processes no longer maintain a local $b$ flag; instead, each $DRead$ operation begins by initializing a flag called $changed$ to $False$ on line~\ref{changeinit}, before starting a repeat-until loop.
During an iteration of the repeat-until loop by a $DRead$ operation, a process $p$ that notices a difference in the values read from $X$ on line~\ref{regread1} and line~\ref{regread2}, or that $A[p]$ does not contain the same value as $X$, sets $changed$ to $True$ on line~\ref{changed} before repeating its sequence of reads. A process $p$ performing a $DRead$ operation also announces the process identifier/sequence number pair read from $X$ on line~\ref{regread1} to $A[p]$ on line~\ref{announce1}. As before, the purpose of this announcement is to prevent $DWrite$ operations from choosing sequence numbers that have been observed recently.
When $p$ sees that its sequence of reads all return the same value, it can safely return this value along with the $changed$ flag. 
In that case, $p$'s return value is consistent with the state of the ABA-detecting register at the point of $p$'s last shared memory operation, i.e., when it reads $X$ for the last time.
Hence, each $DRead$ method may now always be linearized at the time of its final read operation.
Similarly, $DWrite$ operations linearize at their final shared memory operation, which is when they write to $X$.
It is easy to see that if all operations can linearize with their final shared memory operation, the corresponding linearization function is prefix-preserving, and thus the object is strongly linearizable.

\begin{algorithm}
    \caption{$DRead$ of a strongly linearizable ABA-Detecting register}
    \label{slaba}
    
    \SetKwProg{Fn}{Function}{:}{}
    
    \nonl\Fn{DRead$_q$} {
      $changed = False$\;		\label{changeinit}
      \Repeat{$(p, s) = (r, s_r)\;\mathbf{and}\; (x, p, s) = (x', p', s')$} { \label{slaba:repeatloop}
        $(x, p, s) \gets X.Read()$\;	\label{regread1}
        $(r, s_r) \gets A[q].Read()$\;	\label{annread}
        $A[q].Write(p, s)$\;		\label{announce1}
        $(x', p', s') \gets X.Read()$\;	\label{regread2}
        \If{$(p, s) \neq (r, s_r)\;\mathbf{or}\; (x, p, s) \neq (x', p', s')$} {		\label{EQcondABA}
          $changed \gets True$\;			\label{changed}
        }
      } \label{repeatcond}
      
      $\mathbf{return}\;(x', changed)$\;
    }
  \end{algorithm}

We now provide a formal argument that Algorithm~\ref{slaba} is strongly linearizable.
For any operation $op$ in a transcript $T$ on an instance of the implementation in Algorithm~\ref{slaba}, let $pt(op)$ be defined as follows:

\begin{enumerate}[labelindent=0pt,labelwidth=\widthof{\ref{def:interp2}},itemindent=1em,leftmargin=!,label=\textbf{Q-\arabic*},ref={Q-\arabic*}]
\item If $op$ is a $DRead$ operation, then let $pt(op) = time(op^{\ref{regread2}})$.\label{slaba:ptr1}
\item If $op$ is a $DWrite$ operation, then let $pt(op) = time(op^{\ref{linwritelin}})$.\label{slaba:ptr2}
\end{enumerate}

\iffull
Let $\mathcal{T}$ represent the set of all possible transcripts of Algorithm~\ref{slaba}.
For every transcript $T \in \mathcal{T}$, define a sequential history $f(T)$ that orders operations according to $pt$, and excludes all operations $op$ for which $pt(op) = \infty$.
That is, for any two operations $op_1, op_2 \in \Gamma (T)$ such that $pt(op_1) \neq \infty$ and $pt(op_2) \neq \infty$, $op_1 \xrightarrow{f(T)} op_2$ if and only if $pt(op_1) < pt(op_2)$.
Note that there is no pair of operations $op_i, op_j \in \Gamma(T)$ such that $pt(op_i) = pt(op_j) \neq \infty$, because, for every operation $op \in \Gamma(T)$, if $pt(op) \neq \infty$ then the step of $T$ at $pt(op)$ is performed by $op$ (see \ref{slaba:ptr1} and \ref{slaba:ptr2}).

For the remainder of Section~\ref{slabasection}, let $T \in \mathcal{T}$ be some finite transcript of some ABA-detecting register implemented by Algorithm~\ref{slaba}.

The following observation is immediate from the implementation of Algorithm~\ref{slaba} and \ref{slaba:ptr2}:

\begin{observation}\label{obs:dwritelinatwrite}
	If an $X.Write(x, p, s)$ operation happens at time $t$, then there exists a $DWrite_p(x)$ operation $dw$ by $p$ such that $pt(dw) = t$.
\end{observation}

The following observations are immediately obtained from Claims~5.9 and~5.10 in \cite{zahrathesis}:

\begin{observation}\label{obs:from_thesis}
  \begin{enumerate}[label=(\alph*)]
    \item 
    Consider two $GetSeq$ calls $g_1$ and $g_2$ by some process $p$, where $g_1$ is invoked before $g_2$.
    If $g_1$ and $g_2$ return the same sequence number $s$, then $p$ completes at least $n$ $GetSeq$ calls between $g_1$ and $g_2$.\label{obs:from_thesis_sameseq}
    \item
    Suppose $X=(x,p,s)\neq(\bot,\bot,\bot)$ at some point $t$, and $A[p]=s$ throughout $[t,t']$, where $t'\geq t$.
    Then process $p$ does not write $(x',p,s)$ to $X$ during $(t,t']$ for any $x'\in D$.\label{obs:from_thesis_announce}
  \end{enumerate}
\end{observation}
\fi

\begin{lemma}\label{lemma:nointerfereABA}
  Let $dr$ be a complete $DRead_q()$ operation performed by process $q$, and suppose at least one $DWrite$ $dw$ linearizes after $time(dr^{\ref{regread1}})$ and before $q$ invokes any other $DRead()$ operation following $dr$.
  Let $p$ be the process executing $dw$, and $s$ the sequence number associated with $dw$.
  Then
  \begin{enumerate}[label=(\alph*)]
    \item\label{lemma:nointA} $A[q]\neq (p,s)$ at $time(dr^{\ref{annread}})$; and
    \item\label{lemma:nointB} $pt(dw)\not\in\bigl[time(dr^{\ref{regread1}}),time(dr^{\ref{regread2}})\bigr]$.
  \end{enumerate}
\end{lemma}

\iffull
\begin{proof}
  We first prove \ref{lemma:nointA}. 
  For the purpose of a contradiction, assume $A[q]=(p,s)$ at $time(dr^{\ref{annread}})$. 
  Since $q$ executes $dr^{\ref{annread}}$ in its final iteration of the repeat-until loop, it follows from the if-condition in line~\ref{repeatcond} that
  \ifspringer
  \begin{equation}\label{eq:noInterfereABA:100}
  	\begin{split}
    &\text{at $time(dr^{\ref{regread1}})$ process $q$ reads $(x,p,s)$ from $X$} \\ &\text{for some value $x\in D$}.
    \end{split}
  \end{equation}
  \else
  \begin{equation}\label{eq:noInterfereABA:100}
    \text{at $time(dr^{\ref{regread1}})$ process $q$ reads $(x,p,s)$ from $X$ for some value $x\in D$}.
  \end{equation}
  \fi
  Therefore, $q$ writes $(p,s)$ to $A[q]$ in $dr^{\ref{announce1}}$.
  Since $A[q]=(p,s)$ prior to that write, and only process $q$ can write to $A[q]$ (and only in line~\ref{announce1}), it follows that $A[q]=(p,s)$ throughout the final iteration of the repeat-until loop of $dr$.
  Moreover, the value of $A[q]$ remains unchanged until $q$ invokes another $DRead_q()$ operation.
  By the lemma assumption, $pt(dw)$ occurs before $q$'s next $DRead_q()$ invocation, so
  \begin{equation}\label{eq:noInterfereABA:200}
    \text{$A[q]=(p,s)$ throughout $\bigl[time(dr^{\ref{regread1}}),pt(dw)]$}.
  \end{equation}
  By the assumption of the lemma and the fact that $dw$ linearizes when $q$ performs line~\ref{linwritelin} by \ref{slaba:ptr2},
  \ifspringer
  \begin{equation}\label{eq:noInterfereABA:300}
  	\begin{split}
    &\text{at $pt(dw)$ process $p$ writes $(y,p,s)$ to $X$,} \\ &\text{for some $y\in D$.}
    \end{split}
  \end{equation}
  \else
  \begin{equation}\label{eq:noInterfereABA:300}
    \text{at $pt(dw)$ process $p$ writes $(y,p,s)$ to $X$, for some $y\in D$.}
  \end{equation}
  \fi
  Statements (\ref{eq:noInterfereABA:100}), (\ref{eq:noInterfereABA:200}), and (\ref{eq:noInterfereABA:300}) contradict Observation~\ref{obs:from_thesis} \ref{obs:from_thesis_announce}.
  This completes the proof of part \ref{lemma:nointA} of this lemma.
  
  \sloppypar
  We now prove part \ref{lemma:nointB}.
  Suppose the statement is not true.
  Then let $dw$ be the $DWrite$ with the latest linearization point $pt(dw)\in\bigl[time(dr^{\ref{regread1}}),time(dr^{\ref{regread2}})\bigr]$.
  Recall that process $p$ executes $dw$, and $s$ is the sequence number used.
  That is, at $pt(dw)$ process $p$ writes a triple $(x,p,s)$ to $X$, for some $x\in D$.
  
  Since each write to $X$ occurs at the linearization point of some $DWrite$ operation, and no other $DWrite$ linearizes in $\bigl(pt(dw),time(dr^{\ref{regread2}})\bigr]$, we have
  that $X=(x,p,s)$ at point $time(dr^{\ref{regread2}})$.
  Thus, process $q$ reads $(x,p,s)$ from $X$ in line~\ref{regread2} during its final iteration of the repeat-until loop of $dr$.
  Then by the loop-guard in line~\ref{repeatcond}, $q$ reads $(p,s)$ from $A[q]$ in line~\ref{annread}, i.e., when it executes $dr^{\ref{annread}}$.
  This contradicts part \ref{lemma:nointA} of this lemma.
\end{proof}
\fi

\begin{lemma}
\label{validityABA}
Let $dr_1, dr_2$ be two complete $DRead$ operations in $T$ by process $q$, with $dr_1 \xrightarrow{T} dr_2$, where $dr_1$ is the latest $DRead$ operation performed by $q$ that precedes $dr_2$. Suppose $dr_2$ returns $(val, a)$ for some value $val$ and $a \in \{True, False\}$. Then $a = True$ if and only if some $DWrite$ operation linearizes in the interval $\bigl(pt(dr_1), pt(dr_2)\bigr)$.
\end{lemma}

\begin{proof}
	First, suppose $a = True$.
	Then $q$ executes line~\ref{changed} during $dr_2$.
	Since the if condition on line~\ref{EQcondABA} is satisfied if and only if the loop-guard on line~\ref{repeatcond} is false, $q$ must execute line~\ref{changed} on the first iteration of the repeat-until loop on line~\ref{slaba:repeatloop} during $dr_2$.
	Let $(x, p, s)$ and $(x', p', s')$ be the tuples returned by the $X.Read()$ operations $xr$ and $xr'$ performed during the first iteration of the loop on line~\ref{regread1} and line~\ref{regread2}, respectively.
	Also let $(r, s_r)$ be the pair returned by the $A[q].Read()$ operation $ar$ performed during the first iteration of the loop on line~\ref{annread}.
	Since $dr_1$ and $dr_2$ are performed in sequence, $xr$, $xr'$, and $ar$ all happen after $pt(dr_1)$.
	Additionally, since $pt(dr_2) = time(dr_2^{\ref{regread2}})$ (i.e. $dr_2$ linearizes at its \emph{final} execution of line~\ref{regread2}), and $dr_2$ performs the repeat-until loop on line~\ref{slaba:repeatloop} more than once,
	\begin{equation}
		\text{$xr$, $xr'$, and $ar$ all happen in $\bigl(pt(dr_1), pt(dr_2)\bigr)$.}\label{slaba:ininterval}
	\end{equation}
	By the condition on line~\ref{EQcondABA}, there are two cases:
	\begin{enumerate}[label=(\roman*)]
		\item Let $(x, p, s) \neq (x', p', s')$.
			Since $xr'$ returns $(x', p', s')$, there is an $X.Write(x', p', s')$ operation that happens in $\bigl(time(xr), time(xr')\bigr)$.
			Hence, by (\ref{slaba:ininterval}) and Observation~\ref{obs:dwritelinatwrite} there exists a $DWrite_{p'}(x')$ operation that linearizes in $\bigl(pt(dr_1), pt(dr_2)\bigr)$.
		\item Let $(p, s) \neq (r, s_r)$.
			Assume that $dr_1$ writes $(p_1, s_1)$ to $A[q]$ in its final call to line~\ref{announce1}.
			By the loop-guard on line~\ref{repeatcond},
			\begin{equation}
				\text{$dr_1^{\ref{regread2}}$ returns $(x_1, p_1, s_1)$.}\label{slaba:lastread}
			\end{equation}
			Since $q$ performs $dr_2$ immediately after $dr_1$, and $q$ is the only process that can write to $A[q]$, $(r, s_r) = (p_1, s_1)$.
			But since $(p, s) \neq (r, s_r)$ by assumption, $(x, p, s) \neq (x_1, p_1, s_1)$.
			Due to~(\ref{slaba:lastread}), there must be an $X.Write(x, p, s)$ operation that happens in the interval $\bigl(time(dr_1^{\ref{regread2}}), time(xr)\bigr)$.
			By Observation~\ref{obs:dwritelinatwrite}, there exists a $DWrite$ operation that linearizes in $\bigl(time(dr_1^{\ref{regread2}}), time(xr)\bigr)$, and by (\ref{slaba:ininterval}), this $DWrite$ linearizes in $\bigl(pt(dr_1), pt(dr_2)\bigr)$.
	\end{enumerate}
	
	Now suppose $a = False$. Assume, for the sake of a contradiction, that some $DWrite$ linearizes in $\bigl(pt(dr_1), pt(dr_2)\bigr)$. Let $dw$ be the $DWrite$ operation that linearizes at the latest time in this interval. Suppose $dw$ is a $DWrite_{p_1}(x_1)$ by process $p_1$ with associated sequence number $s_1$. Due to Lemma~\ref{lemma:nointerfereABA}~\ref{lemma:nointB}, we know that no $DWrite$ can linearize in the interval $\bigl[time(dr_2^{\ref{regread1}}), time(dr_2^{\ref{regread2}})\bigr]$. Thus, 
\begin{equation}\label{dwisadwrite}
	\text{$dw$ linearizes in $\bigl(time(dr_1^{\ref{regread2}}), time(dr_2^{\ref{regread1}})\bigr)$.}
\end{equation}

Since $dw$ is the final $DWrite$ operation that linearizes prior to $dr_2^{\ref{regread2}}$,
\begin{equation}\label{Xstaysthesame}
	X = (x_1, p_1, s_1)\text{ throughout }\bigl[pt(dw), pt(dr_2)\bigr].
\end{equation}

If $dr_2$ performs the repeat-until loop on line~\ref{slaba:repeatloop} more than once, then since the loop-guard on line~\ref{CHcond} is true if and only if the condition line~\ref{EQcondABA} is true, $q$ sets $changed$ to $True$ during $dr_2$, which is a contradiction.
Hence, $dr_2$ must perform only one iteration of the repeat-until loop.
By (\ref{dwisadwrite}) and (\ref{Xstaysthesame}), the $X.Read()$ operation $dr_2^{\ref{regread1}}$ must return $(x_1, p_1, s_1)$.
Also, $q$ reads $(p_1, s_1)$ from $A[q]$ during $dr_2^{\ref{annread}}$, as otherwise the loop would repeat by the loop-guard on line~\ref{CHcond}.
Hence, the last write to $A[q]$ prior to $dr_2^{\ref{annread}}$ must have been an $A[q].Write(p_1, s_1)$ operation.
Since $q$ only performs one iteration of the repeat-until loop during $dr_2$, and thus does not write anything to $A[q]$ prior to $dr_2^{\ref{annread}}$ while executing $dr_2$, $q$ must have written $(p_1, s_1)$ to $A[q]$ in the final iteration of the repeat-until loop of $dr_1$ (i.e. during $dr_1^{\ref{announce1}}$).
Since the repeat-until loop of $dr_1$ also terminated after that write, $q$ must have read $(p_1, s_1)$ from $A[q]$ during $dr_1^{\ref{annread}}$.
This, along with (\ref{dwisadwrite}), contradicts Lemma~\ref{lemma:nointerfereABA}~\ref{lemma:nointA}.
\end{proof}

\begin{lemma} \label{lemma:valvalidityABA}
	Let $dr$ be a complete $DRead_p()$ operation that returns $(val, a)$ for some $val \in D$, $val \neq \bot$, and some $a \in \{True, False\}$. Then 
	\begin{enumerate}[label=(\arabic*)]		
		\item there is some $DWrite$ operation that linearizes prior to $pt(dr)$, and\label{lemma:valvalidityABA_exist}
		\item if $dw$ is a $DWrite_q(x)$ operation, and no $DWrite$ operation linearizes in $\bigl(pt(dw), pt(dr)\bigr]$, then $x = val$.\label{lemma:valvalidityABA_latest}
	\end{enumerate}
\end{lemma}

\iffull
\begin{proof}
We first prove \ref{lemma:valvalidityABA_exist}.
Since $dr$ returns $(val, a)$, the first component of $X$ contains $val$ at $time(dr^{\ref{regread2}}) = pt(dr)$.
Hence, some $X.Write(val, p, s)$ operation occurs before $pt(dr)$, for some process identifier $p$ and some sequence number $s$.
Then by Observation~\ref{obs:dwritelinatwrite}, there is some $DWrite_p(val)$ operation that linearizes prior to $pt(dr)$.

We now prove \ref{lemma:valvalidityABA_latest}.
Let $dw$ be defined as in the statement of \ref{lemma:valvalidityABA_latest}.
Since $dw$ linearizes when it writes $x$ to $X$, and no $DWrite$ linearizes in the interval $\bigl(pt(dw), pt(dr)\bigr]$, the first component of $X$ contains $x$ throughout the interval $\bigl(pt(dw), pt(dr)\bigr]$.
\end{proof}

\begin{theorem}
\label{linABAthm}
	The sequential history $f(T)$ is a linearization of the interpreted history $\Gamma (T)$. 
\end{theorem}

\begin{proof}
First note that for any operation $op \in \Gamma(T)$, $pt(op) \in \bigl[inv(op), rsp(op)\bigr]$, since $pt(op)$ is assigned directly to a line of code that is executed by $op$ for both $DWrite$ and $DRead$ operations. Thus, $f(T)$ preserves the happens-before order of the interpreted history $\Gamma (T)$.

Lemma~\ref{validityABA} and Lemma~\ref{lemma:valvalidityABA} ensure that the history $f(T)$ is valid with respect to the sequential specification of ABA-detecting registers. Thus, $f(T)$ is a linearization of $\Gamma (T)$.
\end{proof}
\fi

\begin{lemma}
\label{prefABA}
	The function $f$ is prefix-preserving.
\end{lemma}

\iffull
\begin{proof}
Consider each step $t$ of $T$, and some operation $op \in \Gamma (T)$. Then $pt(op) = t$ if

\begin{enumerate}[label=(\roman*)]
	\item operation $op$ is some $DWrite$ operation, and $t$ is the step at which $op$ executes line~\ref{linwritelin} (this case follows from \ref{slaba:ptr2}), or
	\item operation $op$ is some $DRead$ operation, and $t$ is the final step at which $op$ executes line~\ref{regread2}. Note that whether $t$ is the final execution of line~\ref{regread2} is entirely determined at step $t$, since all of the values compared on line~\ref{repeatcond} are stored in local memory at $t$ (this case follows from \ref{slaba:ptr2}).
\end{enumerate}

Thus, at step $t$ it is entirely determined which operations $op$ satisfy $pt(op) = t$.
That is, whether $t$ satisfies $t = pt(op)$ depends entirely on steps that occur at or before $t$, and not on steps that occur after $t$.
Hence, if $T$ is a prefix of $T' \in \mathcal{T}$, then $f(T)$ is a prefix of $f(T')$.
\end{proof}
\fi

\begin{theorem}\label{thm:ABA-safety}
	The implementation represented by Algorithm~\ref{slaba} is strongly linearizable.
\end{theorem}

\iffull
\begin{proof}
Theorem~\ref{linABAthm} shows that the sequential history $f(T)$ is a linearization of the interpreted history $\Gamma (T)$.
Furthermore, Theorem~\ref{prefABA} shows that $f$ is prefix-preserving.
Thus, $f$ is a strong linearization function for $\mathcal{T}$.
\end{proof}
\fi

\iffull
\ifthesis
\section{Lock-Freedom and Complexity Analysis}
\else
\subsection{Lock-Freedom and Complexity Analysis}
\fi
\fi
It is easy to see that each $DWrite_q(x)$ operation performs only two shared memory steps (a read of $A[q]$ and a write to $X$).
Hence, the implementation of the $DWrite$ method is wait-free.
However, a $DRead_q()$ operation by a process $q$ may not terminate if it is ``interrupted'' by infinitely many $DWrite$ operations.
  But note that in each iteration of the repeat-until loop process $q$ writes the same pair $(p,s)$ to $A[q]$ in line~\ref{announce1} that it previously read in line~\ref{regread1} from the second and third component of $X$.
  Hence, if $q$ executes sufficiently many steps while $X$ does not change, then $q$ will eventually read the same pair from $X$ in line~\ref{regread1}, from $A[q]$ in line~\ref{annread}, and from $X$ again in line~\ref{regread2}.
  After that, $q$'s $DRead$ terminates.
  Thus, in any transcript in which $q$ takes sufficiently many steps, either its $DRead$ terminates, or a $DWrite$ terminates.

\iffull
\begin{lemma}\label{lemma:dreadrestart}
	Let $dr$ be a $DRead_q()$ operation by process $q$. Let $xr_1, xr_2,$ and $xr_3$ be three consecutive $X.Read()$ operations on line~\ref{regread1} by $dr$. Then there exists a $DWrite$ operation that linearizes in $\bigl(time(xr_1), time(xr_3)\bigr)$.
\end{lemma}

\begin{proof}
	If the values returned by $xr_1$ and $xr_2$ are not equal, then some $X.Write$ operation occurs in the interval $\bigl(time(xr_1), time(xr_2)\bigr)$.
	Then by Observation~\ref{obs:dwritelinatwrite}, a $DWrite$ operation linearizes in $\bigl(time(xr_1), time(xr_2)\bigr)$.
	
	Now suppose $xr_1$ and $xr_2$ return the same tuple $(x, p, s)$.
	Then $dr$ performs an $A[q].Write(p, s)$ operation on line~\ref{announce1} after $xr_1$, and $A[q]$ is not modified again before line~\ref{annread} is performed by $dr$ following $xr_2$.
	Hence, the $A[q].Read()$ operation on line~\ref{annread} performed by $dr$ following $xr_2$ returns $(p, s)$.
	If the $X.Read()$ operation on line~\ref{regread2} performed by $dr$ following $xr_2$ returns $(x, p, s)$, then by the loop-guard on line~\ref{repeatcond}, $dr$ terminates after this iteration of the main loop.
	This is a contradiction, since $dr$ must restart the repeat-until loop after $xr_2$ in order to perform $xr_3$.
	Therefore, the $X.Read()$ operation on line~\ref{regread2} performed by $dr$ following $xr_2$ returns $(x', p', s') \neq (x, p, s)$.
	Then some $X.Write(x', p', s')$ must occur after $time(xr_2)$ and before the following execution of line~\ref{regread2} by $dr$, and hence in the interval $\bigl(time(xr_1), time(xr_3)\bigr)$.
	By Observation~\ref{obs:dwritelinatwrite}, there is a $DWrite_{p'}(x')$ operation that linearizes in this interval.
\end{proof}
\fi

\begin{theorem}\label{thm:ABA-complexity}
  \begin{enumerate}[label=(\alph{*})]
    \item Each $DWrite()$ performs at most two shared memory operations; and\label{thm:ABA-complexity_write}
    \item for any transcript that contains $r$ $DRead$ and $w$ $DWrite$ invocations, the total number of steps devoted to $DRead$ operations is $O(\min (r,n)\cdot w+r)$.\label{thm:ABA-complexity_read}
  \end{enumerate} 
  In particular, the implementation is lock-free and has amortized step complexity $O(n)$.
\end{theorem}
\iffull
\begin{proof}
  Part~\ref{thm:ABA-complexity_write} follows immediately from the pseudocode (Algorithm~\ref{slaba}).
  
  By Lemma~\ref{lemma:dreadrestart}, each process reads $X$ on line~\ref{regread1} at most $3w + 1$ times during a single $DRead$ operation. 
  This immediately shows that the total number of steps devoted to $DRead$ operations is $O(r \cdot (w + 1))$.
  This proves part~\ref{thm:ABA-complexity_read} for the case $r\leq n$.
  
  We now consider the case $r>n$.
  For any process let $r_p$ denote the number of $DRead$ invocations by process $p$.
  Further, for $i\in\{1,\dots,r_p\}$ let $k_{p,i}$ denote the total number of times process $p$ reads $X$ in line~\ref{regread1} during its $i$-th $DRead$ operation.
  From Lemma~\ref{lemma:dreadrestart} we obtain $\sum_i k_{p,i}=O(w + r_p)$ for each process $p$.
  Using $r=\sum_p r_p$ we obtain that the total number of times all processes read $X$ during all $DRead$ operations is
  \begin{displaymath}
  \sum_{p}\sum_{i=1}^{r_p}k_{p,i}
  =
  \sum_p O(w + r_p)
  =
  O(n\cdot w + r)
  \end{displaymath}
  This proves part~\ref{thm:ABA-complexity_read} for the case $r>n$.
\end{proof}
\fi

Theorems~\ref{thm:ABA-safety} and~\ref{thm:ABA-complexity} yield Theorem~\ref{thm:main-ABA}.

\ifthesis
\else
\subsection{Remarks}
\fi

In this section, we presented the first strongly linearizable implementation of an ABA-detecting register by modifying a previous linearizable implementation by Aghazadeh and Woelfel \cite{ABAreg,zahrathesis}.
An obvious extension of this work would examine the possibility of a wait-free strongly linearizable implementation of an ABA-detecting register.
However, we strongly suspect that such an implementation is impossible.
Denysyuk and Woelfel showed that there is no wait-free strongly linearizable implementation of a counter (defined in Section~\ref{chapter:intro}) from registers, and by reduction this implies that no such implementations exist for snapshots or max-registers, either \cite{WaitVsLock}.
The authors first assume there exists a wait-free implementation of a counter, then define a history in which a $Read$ operation takes an infinite number of steps, but the return value of the $Read$ is never determined; this contradicts the wait-freedom of the implementation.
We believe a similar argument could be applied to the ABA-detecting register.
For this paper, our efforts were concentrated on the snapshot implementation; the strongly-linearizable ABA-detecting register was incidental.
The analysis we performed on our snapshot implementation is not affected by the fact that the ABA-detecting register implementation is only lock-free (as opposed to wait-free).
Hence, we did not dedicate much time or effort to designing a wait-free strongly linearizable ABA-detecting register (or proving the impossibility of such an implementation).
However, we would like to study this issue further.
In fact, it would be interesting to know if \emph{any} nontrivial type has a wait-free strongly linearizable implementation from registers.

\ifthesis
\else
\section{A Strongly Linearizable Snapshot} \label{sectionSLSS}
\fi

\ifea
\else
  A \emph{(single-writer) snapshot} \cite{snapshot} is a type that provides the invocation descriptions $update_q(x)$ and $scan()$. A snapshot object has, for each process $p \in \{1, \ldots, n\}$, an \emph{entry} that stores a value from some finite domain $D$. That is, the snapshot object contains a vector $X \in D^n$, which is initially $(\bot, \ldots, \bot)$. The $update_p(x)$ invocation, for any value $x \in D \setminus \{\bot\}$, changes the $p$-th entry of $X$ to $x$, and the $scan()$ invocation returns the vector $X$. We emphasize that the value $\bot$ strictly signifies the initial state of each entry; that is, once an entry of the snapshot contains a value $x \neq \bot$, no process may change the value of this entry back to $\bot$.
\fi
  
We use brackets to denote individual entries in vectors and snapshot objects. More precisely, for any vector $X = (x_1, \ldots, x_k)$ and any $p \in \{1, \ldots, k\}$, $X[p] = x_p$. Additionally, if $O$ is a snapshot object, then $O[p]$ denotes the $p$-th entry of the object (i.e. the entry that is writeable by process $p$).

\ifthesis
\else
\subsection{Unbounded Implementation}\label{sec:unboundedslss}
\fi
Denysyuk and Woelfel \cite{WaitVsLock} define a general lock-free construction for \emph{versioned} objects, each storing a version number. A versioned object has an atomic update operation, which increases the object's version number every time it is invoked. A versioned object also supports a read operation which returns the state of the object along with its version number. The simple lock-free linearizable algorithm based on clean double collects from \cite{snapshot} can be transformed into a versioned snapshot object easily, by adding a sequence number field to each component that is incremented with each update of the component. The version number of the entire object may be obtained by calculating the sum of the sequence numbers of every component.

The strongly linearizable construction of a versioned object also uses a strongly linearizable bounded max-register described by Helmi, Higham, and Woelfel \cite{PossImposs}. Denysyuk and Woelfel augment the max-register implementation such that it stores a pair $(x, y)$. The augmented max-register supports a $maxRead()$ invocation, which returns the pair $(x, y)$, and a $maxWrite(x', y')$ invocation, which replaces the stored pair $(x, y)$ with $(x', y')$ provided that $x' > x$.

Denysyuk and Woelfel's construction uses a single instance of a versioned object $S$ of type $\mathscr{T}$, along with a single instance of an augmented max-register $R$. A strongly linearizable object $S'$ of type $\mathscr{T}$ is obtained as follows: to perform an $S'.update(x)$ operation, a process executes an $S.update(x)$ operation, reads $S$ to obtain the pair $(y, v)$, and finally performs an $R.maxWrite(v, y)$ operation (note that $v$ represents the version number of the object). To execute an $S'.read()$ operation, a process performs an $R.maxRead()$ operation to obtain the pair $(v, y)$, and returns the value $y$.

The fact that this algorithm is strongly linearizable follows from a simple argument. Let $up$ be some $S'.update(x)$ operation. Suppose that $S'$ has version number $v$ immediately after $up$ performs its $S.update(x)$ operation. Then $up$ may be linearized as soon as some $R.maxWrite(v', y)$ operation, with $v' \geq v$, linearizes. If multiple $S'.update$ operations linearize at the same step, then these operations may be ordered by the times at which their atomic $S.update$ operations responded. An $S'.read()$ operation may be linearized as soon as its $R.maxRead()$ operation linearizes. Since each operation on $S'$ can be linearized at the same step as an operation on $R$, then $S'$ is strongly linearizable because $R$ is strongly linearizable.

\iffull

\ifthesis
\section{Interpreted Value}\label{sec:interpvalue}
\else
\subsection{Interpreted Value}\label{sec:interpvalue}
\fi
  Suppose $O$ is an instance of a snapshot object, and let $T$ be a transcript that contains operations on $O$.
  If $O$ is atomic, then it is easy to determine the value of $O$ at any step $t$ of $T$, since $update$ operations on $O$ take effect instantaneously.
  However, if $O$ is a linearizable implementation of a snapshot object (in particular, if $update$ operations on $O$ are non-atomic), then the value of $O$ at any step $t$ of $T$ is not well-defined.
  To address this issue, we begin by introducing the notion of \emph{interpreted value}, which allows us to reason about the contents of a linearizable object $O$ at every step of $T$.
  

Let $T$ be a fixed transcript, and let $O$ be some linearizable snapshot object.
Let $pt$ be a linearization point function for $O$ (recall that for any transcript $T$, $pt$ maps operations in $\Gamma(T|O)$ to points in time in $T$).
The \emph{interpreted value} of $O[p]$ induced by $pt$ at time $t$ of $T$ is $x$ if and only if one of the following statements hold:
\begin{enumerate}[labelindent=0pt,labelwidth=\widthof{\ref{def:interp2}},itemindent=1em,leftmargin=!,label=\textbf{T-\arabic*},ref={T-\arabic*}]
	\item There is an $O.update_p(x)$ operation $up \in \Gamma(T|O)$ by $p$ such that $pt(up) < t$, and there does not exist any $O.update_p(x')$ operation $up' \in \Gamma(T|O)$ by $p$ such that $pt(up) < pt(up') \leq t$, for any $x' \neq x$.\label{def:interp1}
	\item There is no such $O.update_p(x)$ operation by $p$ in $\Gamma(T|O)$, and $x = \bot$.\label{def:interp2}
\end{enumerate}
When $pt$ and $T$ are clear from context, we say that the interpreted value of $O[p]$ at time $t$ is $x$.

Intuitively, if the interpreted value of $O[p]$ at time $t$ is $x$, then any $scan$ operation $sc \in \Gamma(T)$ such that $pt(sc) = t$ must return a vector with $x$ in its $p$-th entry (this simply follows from the sequential specification of the snapshot type, along with the assumption that $S$ is a linearization of $\Gamma(T)$).
Note that by the definition above, if $up_1$ and $up_2$ are two consecutive $O.update$ operations by process $p$ such that $up_1$ and $up_2$ write distinct values and $pt(up_2) \neq \infty$, then the interpreted value of $O[p]$ at $pt(up_2)$ is $\bot$.
However, if $up_1$ and $up_2$ both write the same value $v$, then the interpreted value of $O[p]$ at $pt(up_2)$ is $v$.

\begin{observation}\label{obs:interpvalue}
	Let $T$ be a transcript, let $O$ be a linearizable snapshot object, and let $pt$ be a linearization point function for $O$. Suppose the interpreted value of $O[p]$ induced by $pt$ at time $t \neq \infty$ of $T$ is $x \neq \bot$. Then for every $scan$ operation $sc$ such that $pt(sc) = t$, $sc$ returns a vector $(x_1, \ldots, x_n)$ such that $x_p = x$.
\end{observation}

\begin{proof}
	Since the interpreted value of $O[p]$ at time $t$ is $x \neq \bot$, by \ref{def:interp1} there exists an $O.update_p(x)$ operation $up \in \Gamma(T|O)$ by $p$ such that
	\ifspringer
	\begin{gather}
		\text{$pt(up) < t$, and}\label{obs:interpless} \\
		\begin{split}
		&\text{no $O.update_p(x')$ operation $up' \in \Gamma(T|O)$ by $p$} \\ &\text{satisfies $pt(up) < pt(up') \leq t$, for any $x' \neq x$.}\label{obs:interpnoop}
		\end{split}
	\end{gather}
	\else
	\begin{gather}
		\text{$pt(up) < t$, and}\label{obs:interpless} \\
		\text{no $O.update_p(x')$ operation $up' \in \Gamma(T|O)$ by $p$ satisfies $pt(up) < pt(up') \leq t$, for any $x' \neq x$.}\label{obs:interpnoop}
	\end{gather}
	\fi
	By the definition of linearization point functions, there exists a linearization $S$ of $\Gamma(T|O)$ such that, 
	\ifspringer
	\begin{gather}
	\begin{split}
		&\text{for any $op \in \Gamma(T|O)$ such that $pt(op) \neq \infty$,} \\ &\text{$op \in S$, and}\label{obs:interpopins}
	\end{split} \\
	\begin{split}
		&\text{for any $op_1, op_2 \in S$, if $op_1 \xrightarrow{S} op_2$ then} \\ &\text{$pt(op_1) \leq pt(op_2)$.}\label{obs:interpextend}
	\end{split}
	\end{gather}
	\else
	\begin{gather}
		\text{for any $op \in \Gamma(T|O)$ such that $pt(op) \neq \infty$, $op \in S$, and}\label{obs:interpopins} \\
		\text{for any $op_1, op_2 \in S$, if $op_1 \xrightarrow{S} op_2$ then $pt(op_1) \leq pt(op_2)$.}\label{obs:interpextend}
	\end{gather}
	\fi
	From the lemma statement, $pt(sc) = t$ and $t \neq \infty$.
	By this, (\ref{obs:interpless}), and (\ref{obs:interpopins}), $up, sc \in S$.
	Since $pt(up) < pt(sc)$, $up \xrightarrow{S} sc$ by contrapositive of (\ref{obs:interpextend}).
	
	Suppose there exists an $O.update_p(x')$ operation $up' \in S$ by $p$ such that
	\begin{gather}
		\text{$x \neq x'$, and}\label{obs:interpxneqxprime} \\
		\text{$up \xrightarrow{S} up' \xrightarrow{S} sc$.}\label{obs:interpSorder}
	\end{gather}
	By (\ref{obs:interpextend}) and (\ref{obs:interpSorder}),
	\begin{gather}
		\text{$pt(up) \leq pt(up')$, and}\label{obs:interpupbeforeupprime} \\
		\text{$pt(up') \leq pt(sc)$.}\label{obs:interpupprimebeforesc}
	\end{gather}
	Note that by definition of linearization point functions,
	\begin{gather}
		\text{$pt(up) \in \bigl[time_T(inv(up)), time_T(rsp(up))\bigr]$, and}\label{obs:interpupinitsoperation} \\
		\text{$pt(up') \in \bigl[time_T(inv(up')), time_T(rsp(up'))\bigr]$.}\label{obs:interpuppriminitsoperation}
	\end{gather}
	By (\ref{obs:interpupinitsoperation}), (\ref{obs:interpuppriminitsoperation}), the fact that processes perform operations sequentially, and (\ref{obs:interpupbeforeupprime}),
	\begin{equation}
		\text{$pt(up) < pt(up')$.}\label{obs:interpupstrictbeforeupprime}
	\end{equation}
	Together, (\ref{obs:interpxneqxprime}), (\ref{obs:interpupprimebeforesc}), and (\ref{obs:interpupstrictbeforeupprime}) contradict (\ref{obs:interpnoop}).
	Hence, there is no $O.update_p(x')$ operation $up' \in S$ by $p$ such that $up \xrightarrow{S} up' \xrightarrow{S} sc$, for any $x' \neq x$.
	By this, the fact that $up \xrightarrow{S} sc$, and the fact that $S$ is a linearization of $\Gamma(T|O)$, the sequential specification of the snapshot type requires that $sc$ returns a vector $(x_1, \ldots, x_n)$ with $x_p = x$.
\end{proof}

\begin{observation}\label{obs:interpintervalafter}
	Let $T$ be a transcript, let $O$ be a linearizable snapshot object, and let $pt$ be a linearization point function for $O$. Suppose $up$ is some $O.update_p(x)$ operation by $p$ such that $pt(up) \neq \infty$. If there is no $O.update_p(x')$ operation $up'$ by $p$ with $x' \neq x$ such that $pt(up') \in \bigl(pt(up), t\bigr]$ for some time $t > pt(up)$, then the interpreted value of $O[p]$ is $x$ throughout $\bigl(pt(up), t\bigr]$.
\end{observation}

\begin{proof}
	This follows trivially from \ref{def:interp1}. 
\end{proof}
\fi

\ifthesis
\else
\subsection{Bounded Implementation}
\fi

Golab, Higham, and Woelfel \cite{stronglin} have shown that the wait-free snapshot implementation designed by Afek, Attiya, Dolev, Gafni, Merritt, and Shavit~\cite{snapshot} is not strongly linearizable. Previous strongly linearizable implementations of snapshot objects are either not lock-free \cite{PossImposs}, or use an unbounded number of registers \cite{WaitVsLock}.
We design a strongly linearizable implementation of a single-writer snapshot object using a linearizable instance of a single-writer snapshot object, along with an atomic ABA-detecting register.

The linearizable snapshot object used by our implementation can be any lock-free or wait-free linearizable implementation of a snapshot object.
To achieve a strongly linearizable snapshot object that uses bounded space, we must ensure that the underlying linearizable snapshot implementation also uses only bounded space.
For the sake of concreteness, we use an implementation by Attiya and Rachman~\cite{attiya1998snapshot}, which is wait-free and linearizable.
The bounded implementation presented in \cite{attiya1998snapshot} uses $O(n^3)$ registers of size $O(\log n + \log |D|)$ to represent views of the snapshot object (where $D$ is the set of values that may be stored by a component), plus $O(n^5)$ registers of size $O(\log n)$ to manage sequence numbers.
This implementation therefore has space complexity $O(n^3(\log n + \log |D|) + n^5\log n)$.
The algorithm performs $O(n\log n)$ operations on MRSW registers during any $scan$ or $update$ operation.
Let $S$ be an instance of this bounded wait-free linearizable snapshot object implementation.
Our algorithm also uses a shared atomic ABA-detecting register $R$. The snapshot object $S$ is used to hold the contents of the strongly linearizable snapshot object, while the ABA-detecting register $R$ contains a vector of size $n$ that represents the state of $S$ at some previous point in time.
Since strong linearizability is a composable property \cite{stronglin,AE2019a},
we can replace the atomic ABA-detecting register $R$ with our strongly linearizable one from Section~\ref{slabasection}.

In order to clearly distinguish between operations on the linearizable snapshot object $S$, and the implemented strongly linearizable snapshot, we call the operations on the latter one $SLupdate$ and $SLscan$.
Pseudocode for our implementation is presented in Algorithm~\ref{slss}.

\begin{algorithm}
  \caption{A strongly linearizable snapshot object}
  \label{slss}
  
  \SetKwProg{Fn}{Function}{:}{}
  
  \nonl\textbf{shared:}\;
    \nonl\quad linearizable snapshot object $S = (\bot, \ldots, \bot)$\;
    \nonl\quad atomic ABA-detecting register $R = (\bot, \ldots, \bot)$\;
  
  \nonl\;
  \nonl\Fn{SLupdate$_p(x)$} {
    $S.update_p(x)$\;				\label{Supdate}
    $s \gets S.scan()$\;				\label{SLUscan}
    $R.DWrite_p(s)$\;					\label{SLUDWrite}
  }
    
  \nonl\;  
  \nonl\Fn{SLscan$_p()$} {
    \Repeat{$(s_1 = \ell = s_2) \; \mathbf{and} \; !c_2$} { \label{slssrepeat}
      $(s_1, c_1) \gets R.DRead_p()$\;		\label{DRead1}
      $\ell \gets S.scan()$\;		\label{SLSscan}
      $(s_2, c_2) \gets R.DRead_p()$\;		\label{DRead2}
      \If{$!(s_1 = \ell = s_2)$} {	\label{EQcond}
        $R.DWrite_p(\ell)$\;	\label{SLSDWrite}
      }
    } \label{CHcond}
    
    \Return $s_2$
  }    
\end{algorithm}

We employ a similar strategy as in the unbounded implementation, but the role of the max-register is now filled by the ABA-detecting register. The $SLupdate$ operation is nearly identical to the update operation of the unbounded implementation by Denysyuk and Woelfel: to perform an $SLupdate_p(x)$ operation, for some $x \in D$, a process $p$ first performs an $S.update_p(x)$ operation on line~\ref{Supdate}, changing the contents of the $p$-th entry of $S$ to $x$. Process $p$ then performs an $S.scan()$ operation on line~\ref{SLUscan}, and finally writes the result of this call to $R$ on line~\ref{SLUDWrite}. Since components of the snapshot object are single-writer, the vector returned by this $S.scan()$ operation must contain $x$ in its $p$-th entry.

An $SLscan_p()$ operation, for some process $p$, is ``stretched'' until a period of time is observed during which the underlying objects $S$ and $R$ are not modified. This way, we force the operation to observe as many $SLupdate$ operations as possible before allowing it to respond. The method works by repeatedly performing an $R.DRead_p()$ operation on line~\ref{DRead1}, then an $S.scan()$ operation on line~\ref{SLSscan}, and finally another $R.DRead_p()$ operation on line~\ref{DRead2}. We will often refer to this sequence of operations as the \emph{main loop} of the $SLscan$ method. Process $p$ continues to perform this sequence of operations until the same vector is returned by all of these calls. Whenever $p$ observes that the contents of $R$ and $S$ are inconsistent, $p$ helps pending $SLupdate$ operations by writing the previously-taken snapshot of $S$ to $R$ on line~\ref{SLSDWrite} before resuming its main loop. When $p$ observes that the $S.scan()$ operation and the two $R.DRead_p()$ operations return the same vector, $p$ will make sure that $R$ was not changed between its most recent pair of $R.DRead_p()$ operations by checking the boolean flag returned by the second $DRead$ on line~\ref{CHcond}; if this flag is false, then $p$ can safely return the value that was returned by its final $DRead$. Thus, process $p$ continues to perform its main loop until it observes that no process interferes during its most recently executed sequence of read operations.

Both $SLscan$ and $SLupdate$ operations work to stabilize the contents of $S$ and $R$. The idea is that the underlying snapshot object $S$ always contains the most recent state of the object, and operations write the state of $S$ that they observed most recently to $R$ (on both line~\ref{SLUDWrite} for $SLupdate$ operations and line~\ref{SLSDWrite} for $SLscan$ operations). A pending $SLupdate_p(x)$ operation by process $p$ can be linearized as soon as some concurrent $SLscan$ operation returns a vector that contains $x$ in its $p$-th entry. The key is that we choose a linearization point for this $SLupdate$ operation at which the interpreted value of $S[p]$ is $x$ and the value of $R[p]$ is $x$. We choose to linearize $SLscan$ operations on their last read of shared memory. That is, an $SLscan$ operation linearizes at its final execution of the $R.DRead$ operation on line~\ref{DRead2}.

Our choice of linearization points results in a strong linearization function because $SLscan$ operations always linearize at their final shared memory step, and when an $SLscan$ operation linearizes, it is already determined which $SLupdate$ operations are caused to linearize by this $SLscan$ operation. Hence, no operations can be retroactively inserted anywhere in the established linearization order.

For any operation $op$ in a transcript $T$ on some instance of the implementation provided in Algorithm~\ref{slss}, we define $pt(op)$ as follows:
\begin{enumerate}[labelindent=0pt,labelwidth=\widthof{\ref{def:interp2}},itemindent=1em,leftmargin=!,label=\textbf{R-\arabic*},ref={R-\arabic*}]
	\item\label{ptr1} Suppose $op$ is some $SLscan$ operation. Then $pt(op)$ is the time at which the final shared memory step is performed by $op$. That is, $pt(op) = time(op^{\ref{DRead2}})$.
	\item\label{ptr2} Suppose $op$ is some $SLupdate_p(x)$ operation for some $x \in D$. If $T$ contains an $SLscan$ operation $sc$ such that $pt(sc) \neq \infty$, $time(inv(op)) < pt(sc)$, and $sc$ returns a vector whose $p$-th entry contains $x$, then let $sc_0$ be such an $SLscan$ operation with minimum possible $pt$ value. Let $t = pt(sc_0)$ if $sc_0$ exists, or $t = \infty$ otherwise. Then $pt(op) = \min \bigl(t, time(op^{\ref{SLUDWrite}})\bigr)$. Recall that if $op^{\ref{SLUDWrite}} \not\in T$, then $time(op^{\ref{SLUDWrite}}) = \infty$.
\end{enumerate}

If $pt(op) \neq \infty$ for some operation $op$, then we say $op$ \emph{linearizes} at $pt(op)$.
\iffull
Let $\mathcal{T}$ represent the set of all possible transcripts on some instance of the object implementation from Algorithm~\ref{slss}. For every transcript $T \in \mathcal{T}$ define a sequential history $f(T)$, such that for every pair of distinct operations $op_1, op_2 \in \Gamma (T)$ by processes $p_1$ and $p_2$ respectively, with $pt(op_1) \neq \infty$ and $pt(op_2) \neq \infty$, $op_1 \xrightarrow{f(T)} op_2$ if and only if

\begin{enumerate}[labelindent=0pt,labelwidth=\widthof{\ref{def:interp2}},itemindent=1em,leftmargin=!,label=\textbf{U-\arabic*},ref={U-\arabic*}]
	\item operation $op_1$ linearizes before $op_2$ (i.e. $pt(op_1) < pt(op_2)$), or
	\item operations $op_1$ and $op_2$ have the same linearization point (i.e. $pt(op_1) = pt(op_2)$), they both have the same invocation description (i.e. they are either both $SLupdate$ operations or both $SLscan$ operations), and $p_1 < p_2$, or\label{ptsorttiebreaker}
	\item operations $op_1$ and $op_2$ have the same linearization point (i.e. $pt(op_1) = pt(op_2)$), $op_1$ is an $SLupdate$ operation, and $op_2$ is an $SLscan$ operation.
\end{enumerate}

Note that $f(T)$ does not contain any operation $op \in \Gamma (T)$ for which $pt(op) = \infty$.


For the remainder of Section~\ref{sectionSLSS}, let $T \in \mathcal{T}$ be some finite transcript on a snapshot object $O$ implemented by Algorithm~\ref{slss}. Additionally, fix a linearization point function $pt_S$ for $S$, where $S$ is the linearizable snapshot object used by Algorithm~\ref{slss}.
By definition of linearization point functions, there exists a linearization $L$ of $\Gamma(T|S)$ such that
\ifspringer
	\begin{gather}
	\begin{split}
		&\text{for every $op \in \Gamma(T|S)$ such that $pt_S(op) \neq \infty$,} \\ &\text{$op \in L$, and}\label{linptfunctions:in}
	\end{split} \\
	\begin{split}
		&\text{for every $op_1, op_2 \in L$, if $op_1 \xrightarrow{L} op_2$ then} \\ &\text{$pt_S(op_1) \leq pt_S(op_2)$.}\label{linptfunction:prec}
	\end{split}
	\end{gather}
\else
	\begin{gather}
		\text{for every $op \in \Gamma(T|S)$ such that $pt_S(op) \neq \infty$, $op \in L$, and}\label{linptfunctions:in} \\
		\text{for every $op_1, op_2 \in L$, if $op_1 \xrightarrow{L} op_2$ then $pt_S(op_1) \leq pt_S(op_2)$.}\label{linptfunction:prec}
	\end{gather}
\fi

\fi

\begin{lemma}
\label{slss:ptinmethod}
For any operation $op \in \Gamma(T)$, $pt(op) \in \bigl[time(inv(op)), time(rsp(op))\bigr]$.
\end{lemma}

\iffull
\begin{proof}
Let $op \in \Gamma(T)$ be some $SLscan_p()$ operation by $p$. Then $pt(op) = time(op^{\ref{DRead2}})$ by \ref{ptr1}. Hence, $pt(op) \in \bigl[time(inv(op)), time(rsp(op))\bigr]$.

Let $op \in \Gamma(T)$ be some $SLupdate_{p}(x_p)$ operation by $p$ for some value $x_p \in D$. By \ref{ptr2}, $pt(op)$ is explicitly defined as a time after $time(inv(op))$ and not after $time(rsp(op))$, which immediately implies the lemma.
\end{proof}
\fi

\begin{lemma}\label{lemma:scanforcesuptolin}
	Suppose $up \in \Gamma(T)$ is some $SLupdate_p(x)$ operation by $p$, for some value $x \in D$. If there exists an $SLscan()$ operation $sc \in \Gamma(T)$ with $time(inv(up)) < pt(sc)$ and $sc$ returns some vector with $x$ in its $p$-th entry, then $pt(up) \leq pt(sc)$.
\end{lemma}

\begin{proof}
This is trivially true if $pt(sc) = \infty$. Otherwise, the lemma follows from \ref{ptr2}.
\end{proof}

\begin{lemma}\label{lemma:eanowrite}
	Let $op \in \Gamma(T)$ be some complete $SLscan_p()$ operation by some process $p$. Then no $R.DWrite$ operation happens in the interval $\bigl[time(op^{\ref{DRead1}}), time(op^{\ref{DRead2}})\bigr]$.
\end{lemma}

\iffull
\begin{proof}
This follows directly from the sequential specification of ABA-detecting registers and the if-statement on line~\ref{CHcond}.
\end{proof}
\fi

\begin{lemma}\label{lemmaEQatlin}
	Let $up \in \Gamma(T)$ be some  $SLupdate_p(x)$ operation with $pt(up) \neq \infty$, for some process $p$ and some $x \in D$. Then the interpreted value of $S[p]$ induced by $pt_S$ is $x$ and $R[p] = x$ at $pt(up)$.
\end{lemma}

\iffull
\begin{proof}
There are two cases:

\begin{enumerate}[labelindent=0pt,labelwidth=\widthof{\ref{def:interp2}},itemindent=1em,leftmargin=!,label=(\roman*)]
	\item Operation $up$ linearizes at $time(up^{\ref{SLUDWrite}})$ (i.e. $pt(up) = time(up^{\ref{SLUDWrite}})$).
	Let $up_u$ be the $S.update_p(x)$ operation performed by $up$ on line~\ref{Supdate}.
	Since entries of $S$ are single-writer, and the $SLupdate$ method only contains a single $S.update$ call,
	\ifspringer
	\begin{equation}
	\begin{split}
		&\text{the interpreted value of $S[p]$ is $x$ throughout} \\ &\text{the interval $\bigl(pt_S(up_u), time(rsp(up))\bigr]$.}\label{lemmaEQ:interpinterval}
	\end{split}
	\end{equation}
	\else
	\begin{equation}
		\text{the interpreted value of $S[p]$ is $x$ throughout the interval $\bigl(pt_S(up_u), time(rsp(up))\bigr]$.}\label{lemmaEQ:interpinterval}
	\end{equation}
	\fi
	By definition of linearization point functions,
	\begin{equation}
		\text{$pt_S(up_u) \in \bigl[time(inv(up_u)), time(rsp(up_u))\bigr]$.}\label{lemmaEQ:ptsupuininterval}
	\end{equation}
	By (\ref{lemmaEQ:ptsupuininterval}) and the fact that processes perform operations sequentially, $pt_S(up_u) < time(up^{\ref{SLUDWrite}})$.
	By this and (\ref{lemmaEQ:interpinterval}),
	\ifspringer
	\begin{equation}
		\begin{split}
		&\text{the interpreted value of $S[p]$ is $x$ at} \\ &\text{$time(up^{\ref{SLUDWrite}})$.}\label{lemmaEQatlin:interp}
		\end{split}
	\end{equation}
	\else
	\begin{equation}
		\text{the interpreted value of $S[p]$ is $x$ at $time(up^{\ref{SLUDWrite}})$.}\label{lemmaEQatlin:interp}
	\end{equation}
	\fi
	Let $sc_u$ be the $S.scan()$ operation performed by $up$ on line~\ref{SLUscan}.
	Again, since processes perform operations sequentially and $pt_S(up_u)$ and $pt_S(sc_u)$ both occur between the invocations and responses of $up_u$ and $sc_u$, respectively, $pt_S(up_u) < pt_S(sc_u)$.
	By this, (\ref{lemmaEQ:interpinterval}), and Observation~\ref{obs:interpvalue}, $sc_u$ returns a vector whose $p$-th entry contains $x$.
	Then the vector written to $R$ by $up$ on line~\ref{SLUDWrite} contains $x$ in its $p$-th entry.
	Hence,
	\begin{equation}
		\text{$R[p] = x$ at $time(up^{\ref{SLUDWrite}})$.}\label{lemmaEQatlin:value}
	\end{equation}
	By (\ref{lemmaEQatlin:interp}), (\ref{lemmaEQatlin:value}), and the assumption that $pt(up) = time(up^{\ref{SLUDWrite}})$, we obtain the claim in the lemma statement.
	
	\item Operation $up$ linearizes before $time(up^{\ref{SLUDWrite}})$ (that is, $pt(up) < time(up^{\ref{SLUDWrite}})$).
	Then by \ref{ptr2} there is some $SLscan$ operation $sc$ that returns a vector whose $p$-th entry contains $x$ while $up$ is pending, and $pt(sc) = pt(up)$.
	By this and the fact that $pt(sc) = time(sc^{\ref{DRead2}})$ by \ref{ptr1},
	\begin{equation}
		\text{$time(inv(up)) < time(sc^{\ref{DRead2}})$.}\label{lemmaEQ:upinvbeforeread}
	\end{equation}
	Since $sc$ linearizes at its final execution of line~\ref{DRead2}, the vectors returned by the $R.DRead()$ operations on line~\ref{DRead1} and line~\ref{DRead2}, along with the vector returned by the $S.scan()$ operation on line~\ref{SLSscan}, must satisfy the condition on line~\ref{CHcond}.
	In particular, if $sc_s$ is the final $S.scan()$ operation performed by $sc$, then
	\begin{equation}
		\text{$sc_s$ returns a vector with $x$ in its $p$-th entry.}\label{lemmaEQ:scsreturns}
	\end{equation}	
	
	Since $L$ is a linearization of $\Gamma(T|S)$, by (\ref{lemmaEQ:scsreturns}) and the sequential specification of the snapshot type, there exists an $S.update_p(x)$ operation $up_{u_x}$ by $p$ such that
	\ifspringer
	\begin{gather}
		\text{$up_{u_x} \xrightarrow{L} sc_s$, and}\label{lemmaEQ:upleqsc} \\
		\begin{split}
		&\text{there is no $S.update_p(x')$ operation $up_u'$ by $p$} \\ &\text{such that $x \neq x'$ and $up_{u_x} \xrightarrow{L} up_u' \xrightarrow{L} sc_s$.}\label{lemmaEQ:thereisnoup}
		\end{split}
	\end{gather}
	\else
	\begin{gather}
		\text{$up_{u_x} \xrightarrow{L} sc_s$, and}\label{lemmaEQ:upleqsc} \\
		\text{there is no $S.update_p(x')$ operation $up_u'$ by $p$ such that $x \neq x'$ and $up_{u_x} \xrightarrow{L} up_u' \xrightarrow{L} sc_s$.}\label{lemmaEQ:thereisnoup}
	\end{gather}
	\fi
	Applying (\ref{linptfunction:prec}) to (\ref{lemmaEQ:upleqsc}),
	\begin{equation}
		\text{$pt_S(up_{u_x}) \leq pt_S(sc_s)$.}\label{lemEQ:ptsupuxleqptsscs}
	\end{equation}
	
	Suppose that $up_{u_x}$ is performed by an $SLupdate_p(x)$ operation other than $up$.
	There is no $S.update_p(x')$ operation $up_u'$ by $p$ such that $x' \neq x$ and $pt_S(up_{u_x}) < pt_S(up_u') < pt_S(sc_s)$, since, using the contrapositive of (\ref{linptfunction:prec}), this would contradict (\ref{lemmaEQ:thereisnoup}).
	Also, there cannot exist such an $up_u'$ by $p$ such that $pt_S(up_{u_x}) = pt_S(up_u')$, since $up_{u_x}$ and $up_u'$ must both linearize between their respective invocations and responses by the definition of linearization point functions.
	Now suppose there exists an $S.update_p(x')$ operation $up_u'$ by $p$ such that $x' \neq x$ and 
	\begin{equation}
		\text{$pt_S(sc_s) \leq pt_S(up_u') < time(inv(up))$.}\label{lemEQ:ptsscsleqptsupuprime}
	\end{equation}
	Let $up'$ be the $SLupdate_p(x')$ operation that performs $up_u'$ on line~\ref{Supdate}.
	Clearly, $up' \neq up$, since $x \neq x'$.
	Since processes perform operations sequentially, $time(rsp(up')) < time(inv(up))$.
	But by this, (\ref{lemmaEQ:upinvbeforeread}), and (\ref{lemEQ:ptsscsleqptsupuprime}), $up'$ must perform its $R.DWrite$ operation from line~\ref{SLUDWrite} in the interval $\bigl[pt_S(sc_s), time(sc^{\ref{DRead2}})\bigr)$, and hence in the interval $\bigl[time(sc^{\ref{DRead1}}), time(sc^{\ref{DRead2}})\bigr]$.
	This contradicts Lemma~\ref{lemma:eanowrite}, and therefore 
	\ifspringer
	\begin{equation}
	\small
	\begin{split}
	&\text{no $S.update_p(x')$ operation $up_u'$ by $p$ with $x' \neq x$} \\ &\text{satisfies $pt_S(up_u') \in \bigl[pt_S(up_{u_x}), time(inv(up))\bigr)$.}\label{lemEQ:nosupdatexprime}
	\end{split}
	\end{equation}
	\else
	\begin{equation}
	\begin{split}
	&\text{no $S.update_p(x')$ operation $up_u'$ by $p$ with $x' \neq x$ satisfies} \\ &\text{$pt_S(up_u') \in \bigl[pt_S(up_{u_x}), time(inv(up))\bigr)$.}\label{lemEQ:nosupdatexprime}
	\end{split}
	\end{equation}
	\fi
	Since $up$ only performs a single $S.update_p(x)$ operation, and $up$ does not perform any $S.update_p(x')$ operations for any $x' \neq x$, by (\ref{lemEQ:nosupdatexprime}) and Observation~\ref{obs:interpintervalafter},
	\ifspringer
	\begin{equation}
		\begin{split}
		&\text{the interpreted value of $S[p]$ is $x$ throughout} \\ &\text{$\bigl(pt_S(up_{u_x}), time(rsp(up))\bigr]$.}\label{lemmaEQatlin:interprangeupdate}
		\end{split}
	\end{equation}
	\else
	\begin{equation}
		\text{the interpreted value of $S[p]$ is $x$ throughout $\bigl(pt_S(up_{u_x}), time(rsp(up))\bigr]$.}\label{lemmaEQatlin:interprangeupdate}
	\end{equation}
	\fi
	If $up_{u_x}$ is performed by $p$ during $up$, then (\ref{lemmaEQatlin:interprangeupdate}) is implied by the fact that $up$ only performs a single $S.update$ operation, along with Observation~\ref{obs:interpintervalafter}.
	
	Notice that, since processes perform operations sequentially, $pt_S(sc_s) < time(sc^{\ref{DRead2}})$.
	Thus, $pt_S(sc_s) < pt(up)$.
	By this, (\ref{lemEQ:ptsupuxleqptsscs}), and (\ref{lemmaEQatlin:interprangeupdate})
	\begin{equation}
		\text{the interpreted value of $S[p]$ is $x$ at $pt(up)$.}\label{lemmaEQatlin:interpvalatpt}
	\end{equation}
	Since $sc$ linearizes at the time of its final $R.DRead()$ operation on line~\ref{DRead2}, and $sc$ returns a vector with $x$ in its $p$-th entry,
	\begin{equation}
		\text{$R[p] = x$ at $pt(sc) = pt(up)$.}\label{lemmaEQatlin:valueatpt}
	\end{equation}
	By (\ref{lemmaEQatlin:interpvalatpt}) and (\ref{lemmaEQatlin:valueatpt}) we obtain the claim in the lemma statement.
\end{enumerate}
\end{proof}
\fi

\begin{lemma}\label{updateexists}
	Let $sc \in \Gamma(T)$ be some $SLscan_q()$ operation by process $q$ such that $pt(sc) \neq \infty$, which returns $(x_1, x_2, \ldots , x_n)$. Then $x_p \neq \bot$ if and only if there exists some $SLupdate_p(x_p)$ operation $up \in \Gamma(T)$, for some $x_p \in D$, such that $pt(up) \leq pt(sc)$.
\end{lemma}

\iffull
\begin{proof}
Let $sc_s$ be the final $S.scan()$ operation performed by $sc$. Since $sc$ returns $(x_1, \ldots, x_n)$,

\begin{gather}
	\text{$sc_s$ returns $(x_1, \ldots, x_n)$, and}\label{lem:upinscanreturns} \\
	\text{$R[p] = x_p$ at $time(sc^{\ref{DRead2}}) = pt(sc)$.}\label{lem:upindreadreturns}
\end{gather}

Additionally, since $pt(sc) \neq \infty$ by the lemma assumption, and $pt(sc) = time(sc^{\ref{DRead2}})$ by \ref{ptr1}, $sc_s$ is complete in $T$.
Using this, along with the fact that $pt_S(sc_s) \in \bigl[time(inv(sc_s)), time(rsp(sc_s))\bigr]$ by the definition of linearization point functions, we obtain

\begin{equation}
	\text{$pt_S(sc_s) \neq \infty$.}\label{lem:updateexistsptscsneqinfty}
\end{equation}

By (\ref{linptfunctions:in}) and (\ref{lem:updateexistsptscsneqinfty}),

\begin{equation}
	\text{$sc_s \in L$.}\label{lem:updatescsinl}
\end{equation}

Suppose $x_p \neq \bot$ for some process $p$.
Then by (\ref{lem:upinscanreturns}), (\ref{lem:updatescsinl}), and the sequential specification of snapshot objects (along with the fact that $L$ is a linearization of $\Gamma(T|S)$), there must exist an $S.update_p(x_p)$ operation $up_u$ by $p$ such that $up_u \xrightarrow{L} sc_s$.
By this and (\ref{linptfunction:prec}),

\begin{equation}
	\text{$pt(up_u) \leq pt(sc_s)$.}\label{lem:updateupuleqscs}
\end{equation}

Let $up_u$ be invoked by the $SLupdate_p(x_p)$ operation $up$. Since $up$ is invoked before $pt(sc)$, $pt(up) \leq pt(sc)$ by Lemma~\ref{lemma:scanforcesuptolin}.

Now suppose $x_p = \bot$. To derive a contradiction, suppose that some $SLupdate_p(x)$ operation $up$ linearizes at or before $pt(sc)$, for some $x \in D$ such that $x \neq \bot$. That is,

\begin{equation}
	\text{$pt(up) \leq pt(sc)$.}\label{lem:upineq1}
\end{equation}

By Lemma~\ref{lemmaEQatlin},

\begin{gather}
	\text{the interpreted value of $S[p]$ at $pt(up)$ is $x$, and}\label{lem:upinInterp} \\
	\text{$R[p] = x$ at $pt(up)$.}\label{lem:upinRcontainsatlin}
\end{gather}

By (\ref{lem:upindreadreturns}) and our assumption that $x_p = \bot$, $R[p] = \bot$ at $pt(sc)$.
This along with (\ref{lem:upineq1}) and (\ref{lem:upinRcontainsatlin}) implies that there must exist an $R.DWrite_q(X)$ operation $dw$ by some process $q$ with $X[p] = \bot$, such that

\begin{equation}
	\text{$pt(up) < time(dw) \leq pt(sc)$.}\label{lem:upinBetween} \\
\end{equation}

By (\ref{lem:upinBetween}) and Lemma~\ref{lemma:eanowrite},

\begin{equation}
	pt(up) < time(dw) < time(sc^{\ref{DRead1}}).\label{lem:upinBeforeRead}
\end{equation}

By (\ref{lem:upinInterp}) and the definition of interpreted value, there exists an $S.update_p(x)$ operation $up_u$ by $p$ such that

\begin{equation}
	\text{$pt_S(up_u) < pt(up)$.}\label{lem:upinptupultptup}
\end{equation}

Using (\ref{lem:upinBeforeRead}), (\ref{lem:upinptupultptup}), along with the fact that $pt_S(sc_s) \in \bigl[time(inv(sc_s)), time(rsp(sc_s))\bigr]$ by the definition of linearization point functions, we obtain

\begin{equation}
	\text{$pt_S(up_u) < pt_S(sc_s)$.}\label{lem:upuhappensbeforescs}
\end{equation}

By (\ref{lem:upuhappensbeforescs}) and the contrapositive of (\ref{linptfunction:prec}), $up_u \xrightarrow{L} sc_s$.
Since there are no $S.update_q(\bot)$ operations for any process $q$ by assumption, $up_u \xrightarrow{L} sc_s$ implies that $sc_s$ returns some vector with $x' \neq \bot$ in its $p$-th entry.
This along with our assumption that $x_p = \bot$ contradicts (\ref{lem:upinscanreturns}).

\end{proof}
\fi

\begin{lemma}
\label{validity}
Let $sc$ be some complete $SLscan_q()$ operation in $T$ by process $q$ that returns $(x_1, \ldots , x_n)$, and suppose $x_p \neq \bot$ for some process $p$. Then there is some $SLupdate_p(x_p)$ operation $up_{x_p}$ such that $pt(up_{x_p}) \leq pt(sc)$, and there is no $SLupdate_p(x)$ operation $up_x$ such that $pt(up_{x_p}) < pt(up_x) \leq pt(sc)$, for any $x \in D$.
\end{lemma}

\begin{proof}
	Let $up_{x_p}$ be the last $SLupdate_p(x_p)$ operation by $p$ that linearizes at or before $pt(sc)$ (Lemma~\ref{updateexists} guarantees that such an operation exists).
	To derive a contradiction, suppose that some $SLupdate_p(x)$ operation by $p$ with $x \neq x_p$ exists, which linearizes in the interval $\bigl(pt(up_{x_p}), pt(sc)\bigr]$.
	Let $up_x$ be the latest such operation.
	So $pt(up_{x_p}) < pt(up_x) \leq pt(sc)$.
	Additionally, since $sc$ returns a vector with $x_p$ in its $p$-th entry, and $up_x$ writes $x \neq x_p$, $up_x$ does not linearize at $pt(sc)$ by \ref{ptr1} and \ref{ptr2} (note that at most a single $SLscan$ operation may linearize at any step, since each $SLscan$ linearizes at one of its own steps by \ref{ptr1}).
	Hence,
	\begin{equation}
		pt(up_{x_p}) < pt(up_x) < pt(sc).\label{lemval:ineq}
	\end{equation}
	By Lemma~\ref{lemmaEQatlin},
	\begin{gather}
		\text{the interpreted value of $S[p]$ is $x$ at $pt(up_x)$, and}\label{lemval:interp} \\
		\text{$R[p] = x$ at $pt(up_x)$.}\label{lemval:register}
	\end{gather}
	Let $sc_\ell$ be the final $S.scan()$ operation performed by $sc$.
	Since processes perform operations sequentially, and $pt_S(sc_\ell) \in \bigl[time(inv(sc_\ell)), time(rsp(sc_\ell))\bigr]$ by definition of linearization point functions,
	\begin{equation}
		time(sc^{\ref{DRead1}}) < pt_S(sc_\ell) < time(sc^{\ref{DRead2}}).\label{lemval:timeseq}
	\end{equation}
	There are two cases:
	\begin{enumerate}[labelindent=0pt,labelwidth=\widthof{\ref{def:interp2}},itemindent=1em,leftmargin=!,label=(\roman*)]
		\item Operation $up_x$ linearizes after $sc_\ell$.
			That is,
			\begin{equation}
				pt(up_x) > pt_S(sc_\ell).\label{lemval:afterscan}
			\end{equation}
			By \ref{ptr1},
			\begin{equation}
				pt(sc) = time(sc^{\ref{DRead2}}).\label{lemval:linatread}
			\end{equation}
			By (\ref{lemval:timeseq}) and (\ref{lemval:afterscan}), $time(sc^{\ref{DRead1}}) < pt(up_x)$.
			By (\ref{lemval:ineq}) and (\ref{lemval:linatread}), $pt(up_x) < time(sc^{\ref{DRead2}})$.
			Hence,
			\begin{equation}
				time(sc^{\ref{DRead1}}) < pt(up_x) < time(sc^{\ref{DRead2}}).\label{lemval:upindanger}
			\end{equation}
			By (\ref{lemval:register}) and the fact that $sc^{\ref{DRead2}}$ returns a vector with $x_p$ in its $p$-th entry, there is an $R.DWrite(X)$ operation $dw$, with $X[p] = x_p$, such that
			\begin{equation}
				pt(up_x) < time(dw) < time(sc^{\ref{DRead2}}).\label{lemval:dwrite}
			\end{equation}
			But by (\ref{lemval:upindanger}) and (\ref{lemval:dwrite}),
			\ifspringer
			\begin{equation*}
				\text{$time(dw) \in \bigl(time(sc^{\ref{DRead1}}), time(sc^{\ref{DRead2}})\bigr)$,}
			\end{equation*}
			\else
			$time(dw) \in \bigl(time(sc^{\ref{DRead1}}), time(sc^{\ref{DRead2}})\bigr)$,
			\fi
			which contradicts Lemma~\ref{lemma:eanowrite}.
		\item Operation $up_x$ linearizes not after $sc_\ell$.
			That is,
			\begin{equation}
				pt(up_x) \leq pt_S(sc_\ell).\label{lemval:beforescan}
			\end{equation}
			By (\ref{lemval:interp}) and the definition of interpreted values, there must exist an $S.update_p(x)$ operation $up_{u_x}$ by $p$ such that
			\ifspringer
			\begin{gather}
				\text{$pt_S(up_{u_x}) < pt(up_x)$, and}\label{lemval:ptupuxltptupx} \\
				\begin{split}
				&\text{there is no $S.update_p(x')$ operation $up_{u_x}'$ by $p$} \\ &\text{with $x' \neq x$ and} \\ &\text{$pt_S(up_{u_x}) < pt_S(up_{u_x}') \leq pt(up_x)$.}\label{lemval:theredoesnotexistan}
				\end{split}
			\end{gather}
			\else
			\begin{gather}
				\text{$pt_S(up_{u_x}) < pt(up_x)$, and}\label{lemval:ptupuxltptupx} \\
				\begin{split}
				&\text{there is no $S.update_p(x')$ operation $up_{u_x}'$ by $p$ with $x' \neq x$ and} \\ &\text{$pt_S(up_{u_x}) < pt_S(up_{u_x}') \leq pt(up_x)$.}\label{lemval:theredoesnotexistan}
				\end{split}
			\end{gather}
			\fi
			
			By (\ref{lemval:beforescan}) and (\ref{lemval:ptupuxltptupx}), $pt_S(up_{u_x}) < pt_S(sc_\ell)$.
			Then by contrapositive of (\ref{linptfunction:prec}), 
			\begin{equation}
				\text{$up_{u_x} \xrightarrow{L} sc_\ell$.}\label{lemval:upuxbeforescell}
			\end{equation}
			Since $sc$ returns a vector with $x_p \neq x$ in its $p$-th entry, $sc_\ell$ must also return a vector with $x_p \neq x$ in its $p$-th entry.
			By this, (\ref{lemval:upuxbeforescell}), the fact that $L$ is a linearization of $\Gamma(T|S)$, and the sequential specification of the snapshot type, there must exist an $S.update_p(x_p)$ operation $up_{u_p}$ by $p$ such that 
			\begin{equation}
				\text{$up_{u_x} \xrightarrow{L} up_{u_p} \xrightarrow{L} sc_\ell$.}\label{lemval:upuxbeforeupupbeforesc}
			\end{equation}
			By (\ref{linptfunction:prec}) and (\ref{lemval:upuxbeforeupupbeforesc}),
			\begin{gather}
				\text{$pt_S(up_{u_x}) \leq pt_S(up_{u_p})$, and}\label{lemval:ptsupuxleqptsupup} \\
				\text{$pt_S(up_{u_p}) \leq pt_S(sc_\ell)$.}\label{lemval:ptsupupleqptsscell}
			\end{gather}
			 Note that by definition of linearization point functions,
			\begin{gather}
			\small
			 \text{$pt_S(up_{u_x}) \in \bigl[time(inv(up_{u_x})), time(rsp(up_{u_x}))\bigr]$,}\label{lemval:ptsupuxininterval} \\
			\small
			 \text{$pt_S(up_{u_p}) \in \bigl[time(inv(up_{u_p})), time(rsp(up_{u_p}))\bigr]$.}\label{lemval:ptsupupininterval}
			 \end{gather}
			 By (\ref{lemval:ptsupuxininterval}), (\ref{lemval:ptsupupininterval}), and the fact that both $up_{u_x}$ and $up_{u_p}$ are performed by $p$, $pt_S(up_{u_x}) < pt_S(up_{u_p})$.
			 By this and (\ref{lemval:theredoesnotexistan}),
			 \begin{equation}
			 	\text{$pt(up_x) < pt_S(up_{u_p})$.}\label{ptupxltptupup}
			 \end{equation}
			 Due to (\ref{lemval:ineq}) and the fact that $up_{x_p}$ and $up_x$ are performed by $p$, the $S.update_p(x_p)$ operation by $up_{x_p}$ linearizes before the invocation of $up_x$.
			 This along with (\ref{ptupxltptupup}) implies that $up_{u_p}$ is performed during some $SLupdate_p(x_p)$ operation $up_{x_p}'$ by $p$ such that $up_{x_p}' \neq up_{x_p}$.
			 Then by Lemma~\ref{lemma:scanforcesuptolin}, 
			 \begin{equation}
				 \text{$pt(up_{x_p}') \leq pt(sc)$.}\label{ptupxpLSKDJFSDLKFJlksdjf}
			 \end{equation}
			 Additionally, since $up_{x_p}$ and $up_{x_p}'$ are both performed by $p$, they are performed in sequence.
			This along with (\ref{ptupxpLSKDJFSDLKFJlksdjf}) contradicts the fact that $up_{x_p}$ is the final $SLupdate_p(x_p)$ operation by $p$ that linearizes not after $pt(sc)$.


	\end{enumerate}
\end{proof}

\begin{lemma}
\label{slss:linearizable}
	Algorithm~\ref{slss} is linearizable.
\end{lemma}

\iffull
\begin{proof}
Lemma~\ref{slss:ptinmethod} ensures that $f$ preserves the happens-before order of the interpreted history $\Gamma (T)$. Lemma~\ref{validity} ensures that Algorithm~\ref{slss} satisfies the sequential specification of a snapshot object, and therefore $f(T)$ is a valid sequential history. Thus, $f(T)$ is a linearization of the interpreted history $\Gamma (T)$.
\end{proof}
\fi

\begin{lemma}
\label{slss:prefpres}
	The function $f$ is prefix-preserving.
\end{lemma}

\iffull
\begin{proof}
Consider each time $t$ of $T$, and some operation $op \in \Gamma (T)$ by $p$. Then $pt(op) = t$ if
\begin{enumerate}[labelindent=0pt,labelwidth=\widthof{\ref{def:interp2}},itemindent=1em,leftmargin=!,label=(\roman*)]
	\item operation $op$ is some $SLscan_p()$ operation and $op^{\ref{DRead2}}$ happens at $t$ (this case follows from \ref{ptr1}), or
	\item operation $op$ is some $SLupdate_{p}(x_p)$ operation for some value $x_p \in D$, and $t$ is the first point in time not before $time(inv(op))$ that satisfies $t = time(sc^{\ref{DRead2}})$ for some $SLscan()$ operation $sc$, and $sc^{\ref{DRead2}}$ returns a vector whose $p$-th entry contains $x_p$ (this case follows from both \ref{ptr1} and \ref{ptr2}), or
	\item operation $op$ is some $SLupdate_p(x_p)$ operation for some value $x_p \in D$ that did not linearize at any step prior to $t$, and $op^{\ref{SLUDWrite}}$ happens at $t$ (this case follows from \ref{ptr2}).
\end{enumerate}
Thus, at step $t$ it is entirely determined which operations linearize at $t$. That is, for any operation $op \in \Gamma(T)$, whether $pt(op) = t$ or not can be decided soley based on previous steps (i.e. steps earlier than $t$) in $T$. Additionally, once it is determined that $pt(op) = t$, $pt(op)$ does not change with any future step (i.e. steps later than $t$) of $T$.
Along with the fact that processes perform operations sequentially, Observation~\ref{slss:ptinmethod} implies that for any time $t$ of $T$, the set of operations $Op$ that linearize at $t$ contains at most one operation by each process.
Therefore, \ref{ptsorttiebreaker} imposes a total order on the operations in $Op$.
Hence, if $T$ is a prefix of $T' \in \mathcal{T}$, then $f(T)$ is a prefix of $f(T')$.
\end{proof}
\fi

\begin{theorem}
\label{stronglin}
	The snapshot object implemented by Algorithm~\ref{slss} is strongly linearizable.
\end{theorem}
\iffull
\begin{proof}
Lemma~\ref{slss:linearizable} shows that the sequential history $f(T)$ is a linearization of $\Gamma (T)$. Furthermore, Lemma~\ref{slss:prefpres} proves that $f$ is prefix-preserving. Thus, $f$ is a strong linearization function for $\mathcal{T}$.
\end{proof}
\fi

\iffull
\ifthesis
\section{Lock-Freedom and Complexity Analysis}\label{compslss}
\else
\subsection{Lock-Freedom and Complexity Analysis} \label{compslss}
\fi
To simplify our analysis of the amortized complexity of our strongly linearizable snapshot object implementation, we first present a modification to Algorithm~\ref{slss}.
The pseudocode of this implementation is provided in Algorithm~\ref{slsswithseq}.
Our modification adds a sequence number field to each entry of the snapshot object; that is, each entry of the snapshot object stores a pair $(x, s)$, where $x \in D$ is a value, and $s$ is an unbounded sequence number.
Initially, the snapshot contains the vector $\bigl((\bot, 0), \ldots, (\bot, 0)\bigr)$.
Each process also stores a local sequence number (called $seq$ in Algorithm~\ref{slsswithseq}), which is incremented every time the process updates the snapshot object on line~\ref{seq:Supdate}.
Following this, process $p$ performs an $S.update_p(x, seq)$ operation on line~\ref{seq:Supdate}, which stores the pair $(x, seq)$ in the $p$-th entry of the snapshot object.
The remainder of the modified implementation is essentially identical to Algorithm~\ref{slss}.
Note that $S.scan()$ and $R.DRead()$ operations now return vectors of pairs.
That is, if $X$ is returned by some $S.scan()$ or $R.DRead()$ operation, then $X = \bigl((x_1, s_1), \ldots, (x_n, s_n)\bigr)$, such that each $x_i \in D$ is a value, and each $s_i$ is a sequence number.
We use $vals(X)$ to denote the vector of values stored by $X$ (i.e. $vals(X) = (x_1, \ldots, x_n)$).

\begin{algorithm}[h]
  \caption{Modified Strongly Linearizable Snapshot Object}
  \label{slsswithseq}

  \SetKwProg{Fn}{Function}{:}{}
  
  \nonl\textbf{shared:}\;
  	\nonl\quad atomic snapshot object $S = \bigl((\bot, 0), \ldots, (\bot, 0)\bigr)$\;
  	\nonl\quad atomic ABA-detecting register $R =  \bigl((\bot, 0), \ldots, (\bot, 0)\bigr)$\;
  \nonl\;
  \nonl\textbf{local} (to each process):\;
    \nonl\quad Integer $seq = 0$\;
    
  \nonl\;
  \nonl\Fn{SLupdate$_p(x)$} {
    $seq \gets seq + 1$\;				\label{seq:seqinc}
    $S.update_p(x, seq)$\;				\label{seq:Supdate}
    $s \gets S.scan()$\;				\label{seq:SLUscan}
    $R.DWrite_p(s)$\;					\label{seq:SLUDWrite}
  }
    
  \nonl\;
  \nonl\Fn{SLscan$_p()$} {
    \Repeat{$(vals(s_1) = vals(\ell) = vals(s_2)) \; \mathbf{and} \; !c_2$} {
      $(s_1, c_1) \gets R.DRead_p()$\;		\label{seq:DRead1}
      $\ell \gets S.scan()$\;		\label{seq:SLSscan}
      $(s_2, c_2) \gets R.DRead_p()$\;		\label{seq:DRead2}
      \If{$!(vals(s_1) = vals(\ell) = vals(s_2))$} {	\label{seq:EQcond}
        $R.DWrite_p(\ell)$\;		\label{seq:SLSDWrite}
      }
    } \label{seq:CHcond}
    
    \Return $s_2$
  }
    
\end{algorithm}

Note that Algorithm~\ref{slsswithseq} also uses an atomic snapshot object for $S$ rather than a linearizable one.
This is acceptable for our analysis because we aim to express the amortized complexity of Algorithm~\ref{slsswithseq} in terms of the number of operation invocations on $S$ and $R$, without regard for how $S$ or $R$ are implemented.
This simplification allows us to avoid using the notion of interpreted value from Section~\ref{sec:interpvalue}.
Instead, for any transcript $T \in \mathcal{T}$, we may simply reference the \emph{value} stored by $S$ at any particular time $t$ of $T$.
Since $update$ operations on $S$ are atomic, the value of $S$ at any particular step of $T$ is well-defined.

Since Algorithm~\ref{slsswithseq} uses unbounded sequence numbers, the implementation uses unbounded space.
However, since Algorithm~\ref{slsswithseq} performs exactly the same shared memory operations as in Algorithm~\ref{slss}, it is easy to see that these implementations have the same amortized complexity.
Hence, while our analysis is performed directly on Algorithm~\ref{slsswithseq}, all of the results in this section may also be applied to Algorithm~\ref{slss}.

We first introduce some notation that will simplify our argument throughout this section. Let $X$ be a vector that is returned by some $S.scan()$ or $R.DRead()$ operation. 
For a process $p$ and an entry $X[p]$ of $X$, define $seq(X[p])$ as the second entry of the pair stored in $X[p]$ (i.e. the sequence number stored by $X[p]$). Define $seq(X) = \sum_{i = 1}^{n} seq(X[i])$. We first observe that as the underlying snapshot object $S$ is modified, $seq$ values of subsequent $S.scan()$ operations never decrease. For the sake of clarity, when we say $X = Y$, for some vectors $X$ and $Y$, we mean $vals(X) = vals(Y)$ and $seq(X) = seq(Y)$. For the remainder of this section, fix a transcript $T$ on Algorithm~\ref{slsswithseq}.

\begin{observation}\label{obsnolookback}
	Suppose that the value of $S$ is $X$ at time $t$ in $T$. If the value of $S$ is $X'$ at time $u \geq t$ of $T$, then
	\begin{enumerate}[label=(\alph*)]
		\item $seq(X'[p]) \geq seq(X[p])$, for every process $p$, and\label{obsnolookback:entry}
		\item $seq(X') \geq seq(X)$.\label{obsnolookback:whole}
	\end{enumerate}
\end{observation}

\begin{proof}
Part \ref{obsnolookback:entry} follows from the fact that each process that performs an $SLupdate$ operation increments its local $seq$ variable on line~\ref{seq:seqinc}, prior to invoking the $S.update$ operation on line~\ref{seq:Supdate}.
Part \ref{obsnolookback:whole} follows trivially from part \ref{obsnolookback:entry}.
\end{proof}

\begin{observation}\label{obsEQ}
	Let $X_1$ and $X_2$ be vectors returned by two $S.scan()$ operations in $T$. If $seq(X_1) = seq(X_2)$, then $X_1 = X_2$.
\end{observation}

\begin{proof}
Let $op_1$ and $op_2$ be two complete $S.scan()$ operations that return $X_1$ and $X_2$ respectively, such that $seq(X_1) = seq(X_2)$. Without loss of generality, suppose $time(op_1) \leq time(op_2)$. To derive a contradiction, suppose that $X_1 \neq X_2$. Thus, some $S.update$ operation by $p$ must happen in the interval $\bigl(time(op_1), time(op_2)\bigr)$, for some process $p$. Since the sequence numbers written by subsequent $S.update$ operations by $p$ always increase (by the increment on line~\ref{seq:seqinc}),
\begin{equation}
	X_1[p] < X_2[p].\label{obsEQ:entry}
\end{equation}
By Observation~\ref{obsnolookback}~\ref{obsnolookback:entry}, $X_1[q] \leq X_2[q]$ for every process $q$.
This, combined with (\ref{obsEQ:entry}), implies that $seq(X_1) < seq(X_2)$, which is a contradiction.
\end{proof}

\begin{observation}\label{obsmorerecent}
	Suppose that at time $t$, the value of $S$ is $X$, while $R$ contains a vector $X'$. Then $seq(X) \geq seq(X')$.
\end{observation}

\begin{proof}
Notice that the only $R.DWrite(X')$ statements present in Algorithm~\ref{slss} acquire $X'$ from some previous $S.scan()$ operation. Hence, $seq(X') \geq seq(X)$ by Observation~\ref{obsnolookback}~\ref{obsnolookback:whole}.
\end{proof}

\begin{lemma}\label{lemmaoldstate}
	Let $sc$ be the first $S.scan()$ operation that returns $X$ (that is, no $S.scan()$ operation $sc'$ such that $time(sc') < time(sc)$ returns $X$). Suppose there are $k$ operations
  \begin{displaymath} 
  R.DWrite(X_1), R.DWrite(X_2), \ldots, R.DWrite(X_k)
  \end{displaymath}
  that happen after $time(sc)$ in $T$, such that $seq(X_i) < seq(X)$ for all $i \in \{1, \ldots, k\}$. Then $k \leq n - 1$.
\end{lemma}

\begin{proof}
By Observation~\ref{obsnolookback}, any $S.scan()$ operation in $T$ that is invoked after $time(sc)$ returns a vector $Y$ with $seq(Y) \geq seq(X)$. Thus, 
\ifspringer
\begin{equation}
\begin{split}
&\text{any $S.scan()$ operation that returns a vector $X'$} \\ &\text{with $seq(X') < seq(X)$ must be invoked before} \\ &\text{$time(sc)$.}\label{lemmaoldstate:oldscan}
\end{split}
\end{equation}
\else
\begin{equation}
\begin{split}
&\text{any $S.scan()$ operation that returns a vector $X'$ with $seq(X') < seq(X)$ must be} \\ &\text{invoked before $time(sc)$.}\label{lemmaoldstate:oldscan}
\end{split}
\end{equation}
\fi

Note that each $R.DWrite$ statement in Algorithm~\ref{slsswithseq} is preceded by an $S.scan$ statement in the same method (i.e. the $R.DWrite$ on line~\ref{seq:SLUDWrite} in the $SLupdate$ method is preceded by the $S.scan$ on line~\ref{seq:SLUscan}, and the $R.DWrite$ on line~\ref{seq:SLSDWrite} in the $SLscan$ method is preceded by the $S.scan$ on line~\ref{seq:SLSscan}).
Furthermore,
\ifspringer
\begin{equation}
\small
\begin{split}
&\text{for any vector $Y$ and any $R.DWrite_p(Y)$ operation} \\ &\text{$dw$ by $p$, the latest $S.scan()$ operation by $p$ that is} \\ &\text{invoked before $dw$ returns $Y$.}\label{lemmaoldstate:scanbeforedwrite}
\end{split}
\end{equation}
\else
\begin{equation}
\begin{split}
&\text{for any vector $Y$ and any $R.DWrite_p(Y)$ operation $dw$ by $p$, the latest $S.scan()$ operation} \\ &\text{by $p$ that is invoked before $dw$ returns $Y$.}\label{lemmaoldstate:scanbeforedwrite}
\end{split}
\end{equation}
\fi
Together, (\ref{lemmaoldstate:oldscan}) and (\ref{lemmaoldstate:scanbeforedwrite}) imply the following: for any $R.DWrite(X')$ operation $dw$ by $p$ such that $seq(X') < seq(X)$, the latest $S.scan()$ operation invoked by $p$ before $dw$ must have been invoked before $time(sc)$.
This immediately implies that
\ifspringer
\begin{equation}
\small
\begin{split}
	&\text{each process performs at most one $R.DWrite(X')$} \\ &\text{operation that happens after $time(sc)$, with} \\ &\text{$seq(X') < seq(X)$}.\label{lemmaoldstate:katmostn}
\end{split}
\end{equation}
\else
\begin{equation}
\begin{split}
	&\text{each process performs at most one $R.DWrite(X')$ operation that happens after $time(sc)$,} \\ &\text{with $seq(X') < seq(X)$}.\label{lemmaoldstate:katmostn}
\end{split}
\end{equation}
\fi
Suppose $sc$ is performed by process $q$.
Since each process performs operations sequentially, (\ref{lemmaoldstate:oldscan}) and (\ref{lemmaoldstate:scanbeforedwrite}) together imply that no $R.DWrite_q(X')$ operation by $q$ happens after $time(sc)$.
This combined with (\ref{lemmaoldstate:katmostn}) implies the statement in the lemma.
\end{proof}


\begin{lemma}\label{lemmabetweenwrites}
Suppose that the value of $S$ is $X$ at some step $t$ in $T$.
Let $sc_Y \in T$ be an $S.scan()$ operation that returns $Y$ with $seq(Y) > seq(X)$, such that there does not exist an $S.scan()$ operation $sc_{Y'}$ that returns a vector $Y'$ with $seq(Y') > seq(X)$ and $time(sc_{Y'}) < time(sc_Y)$.
Then,
\begin{enumerate}[label=(\alph*)]
\item for any pair of $R.DWrite_p(X)$ operations $dw_1, dw_2$ by $p$ on line~\ref{seq:SLSDWrite} (i.e. during $SLscan$ operations) such that $time(dw_1) < time(dw_2) < time(sc_Y)$, there exists an $R.DWrite(X')$ operation that happens in $\bigl(time(dw_1), time(dw_2)\bigr)$, with $seq(X') < seq(X)$, and\label{lemmabetweenwrites:pair}
\item if $dw_1, \ldots, dw_k$ is a sequence of operations \\ $R.DWrite_p(X_1), \ldots, R.DWrite_p(X_k)$ performed by $p$ such that $seq(X_1) = \ldots = seq(X_k) = seq(X)$, then $k\leq 2n$.\label{lemmabetweenwrites:sequence}
\end{enumerate}
\end{lemma}

\begin{proof}[Proof of Lemma~\ref{lemmabetweenwrites}~\ref{lemmabetweenwrites:pair}]
Since $dw_2$ writes the vector $X$, by the pseudocode in Algorithm~\ref{slsswithseq} the latest $S.scan()$ operation on line~\ref{seq:SLSscan} that was invoked by $p$ prior to $time(dw_2)$ must have returned $X$.
Let $V_1$ and $V_2$ be the vectors returned by the last executions of $R.DRead_p()$ by process $p$ prior to $time(dw_2)$ on lines \ref{seq:DRead1} and \ref{seq:DRead2}, respectively.
By the condition on line~\ref{seq:EQcond}, either $vals(V_1) \neq vals(X)$ or $vals(V_2) \neq vals(X)$.
Then either $V_1 \neq X$ or $V_2 \neq X$.
In either case, since $dw_1$ writes $X$ to $R$, there must exist an $R.DWrite(X')$ operation $dw_{X'}$ (with $X' \neq X$, and either $X' = V_1$ or $X' = V_2$) that happens after $time(dw_1)$, but before the $R.DRead_p()$ operation that returns a vector distinct from $X$.
Hence,
\begin{equation}
	time(dw_{X'}) \in \bigl(time(dw_1), time(dw_2)\bigr).\label{lemmabetweenwrites:interval}
\end{equation}
Since all $R.DWrite$ operations in Algorithm~\ref{slsswithseq} write the result of some previous $S.scan()$ operation to $R$, there must be an $S.scan()$ operation $sc_{X'}$ that returns $X'$ and happens prior to $time(dw_{X'})$.
Then by (\ref{lemmabetweenwrites:interval}), $time(sc_{X'}) < time(dw_2)$, and therefore $time(sc_{X'}) < time(sc_Y)$.
By the assumption of the lemma, $seq(X') \leq seq(X)$.
Since $X' \neq X$, $seq(X') \neq seq(X)$ by the contrapositive of Observation~\ref{obsEQ}, and therefore $seq(X') < seq(X)$.
\end{proof}

\begin{proof}[Proof of Lemma~\ref{lemmabetweenwrites}~\ref{lemmabetweenwrites:sequence}]
Suppose $dw_i$ is performed by an $SLupdate$ operation $up$ by $p$, for some $i \in \{1, \ldots, k\}$.
First assume that $i > 1$.
Then
\begin{equation}
	\text{$R$ contains $X_{i-1}$ at $time(dw_{i-1})$.}\label{lemmabetweenwrites:attimedwrite}
\end{equation}
Since $SLupdate$ operations perform at most a single $R.DWrite$ operation (line~\ref{seq:SLUDWrite}),
\begin{equation}
	time(dw_{i-1}) < time(inv(up)).\label{lemmabetweenwrites:dwritebeforeup}
\end{equation}
Suppose that the value of $S$ at $time(dw_{i-1})$ is $V_1$.
Then by (\ref{lemmabetweenwrites:attimedwrite}) and Observation~\ref{obsmorerecent}, $seq(V_1) \geq seq(X_{i-1}) = seq(X)$.
If the value of $S$ at $time(inv(up))$ is $V_2$, then by (\ref{lemmabetweenwrites:dwritebeforeup}) and Observation~\ref{obsnolookback}~\ref{obsnolookback:whole} $seq(V_2) \geq seq(V_1) \geq seq(X)$.
By the increment on line~\ref{seq:seqinc}, $p$ increases the sequence number of the vector stored by $S$ when it performs the $S.update$ operation on line~\ref{seq:Supdate} during $up$.
Then the vector $X'$ returned by the $S.scan()$ operation performed by $up$ on line~\ref{SLUscan} satisfies $seq(X') > seq(V_2) \geq seq(X)$.
By the pseudocode of Algorithm~\ref{slsswithseq}, $dw_i$ is an $R.DWrite_p(X')$ operation (i.e. $X_i = X'$), which contradicts the assumption that $seq(X_i) = seq(X)$.
Therefore, $i = 1$, and
\ifspringer
\begin{equation}
	\small
	\begin{split}
	&\text{only a single member of the sequence may belong} \\ &\text{to an $SLupdate$ operation.}\label{lemmabetweenwrites:belongtoslupdate}
	\end{split}
\end{equation}
\else
\begin{equation}
	\begin{split}
	&\text{only a single member of the sequence may belong to an $SLupdate$ operation.}\label{lemmabetweenwrites:belongtoslupdate}
	\end{split}
\end{equation}
\fi

By (\ref{lemmabetweenwrites:belongtoslupdate}), $dw_2, \ldots, dw_k$ are all performed by $SLscan$ operations (on line~\ref{seq:SLSDWrite}).
Let $sc_X$ be the earliest $S.scan()$ operation in $T$ that returns $X$ (i.e. there does not exist an $S.scan()$ operation $sc_X'$ which returns $X$ such that $time(sc_X') < time(sc_X)$).
Since $seq(X_i) = seq(X)$, Observation~\ref{obsEQ} implies that $X_i = X$, for all $i \in \{1, \ldots, k\}$.
Hence, since every $R.DWrite$ operation writes a vector returned by some earlier $S.scan()$ operation, 
\begin{equation}
	\text{$time(sc_X) < time(dw_i)$ for all $i \in \{1, \ldots, k\}$.}\label{lemmabetweenwrites:afterscan}
\end{equation}
By Lemma~\ref{lemmabetweenwrites}~\ref{lemmabetweenwrites:pair}, for every pair of consecutive operations $dw_i, dw_{i+1}$ that happen before $time(sc_Y)$, there exists an $R.DWrite(X')$ operation that happens in $\bigl(time(dw_i), time(dw_{i+1})\bigr)$ with $seq(X') < seq(X)$.
By Lemma~\ref{lemmaoldstate} there are at most $n - 1$ such $R.DWrite(X')$ operations that happen after $time(sc_X)$.
Combining this with (\ref{lemmabetweenwrites:afterscan}), we obtain the following:
\ifspringer
\begin{equation}
\small
\begin{split}
	&\text{There are at most $n$ $R.DWrite_p(X_i)$ operations} \\ &\text{performed during $SLscan$ operations that happen} \\ &\text{before $time(sc_Y)$.}\label{lemmabetweenwrites:beforescy}
\end{split}
\end{equation}
\else
\begin{equation}
\begin{split}
	&\text{There are at most $n$ $R.DWrite_p(X_i)$ operations performed during $SLscan$ operations} \\ &\text{that happen before $time(sc_Y)$.}\label{lemmabetweenwrites:beforescy}
\end{split}
\end{equation}
\fi
By Lemma~\ref{lemmaoldstate},
\ifspringer
\begin{equation}
\small
\begin{split}
	&\text{there are at most $n - 1$ $R.DWrite_p(X_i)$ operations} \\ &\text{performed during $SLscan$ operations that happen} \\ &\text{after $time(sc_Y)$.}\label{lemmabetweenwrites:afterscy}
\end{split}
\end{equation}
\else
\begin{equation}
\begin{split}
	&\text{there are at most $n - 1$ $R.DWrite_p(X_i)$ operations performed during $SLscan$ operations} \\ &\text{that happen after $time(sc_Y)$.}\label{lemmabetweenwrites:afterscy}
\end{split}
\end{equation}
\fi
Combining (\ref{lemmabetweenwrites:belongtoslupdate}), (\ref{lemmabetweenwrites:beforescy}), and (\ref{lemmabetweenwrites:afterscy}) yields $k \leq 2n$.
\end{proof}

\begin{lemma}\label{lem:sequence_numbers_per_scan}
	Let $op \in \Gamma(T)$ be an operation by process $p$, and suppose $op$ performs a sequence of $S.scan()$ operations $sc_1, sc_2, \ldots, sc_k$, such that $sc_i$ returns a vector $X_i$ for all $i \in \{1, \ldots, k\}$, and $seq(X_1) = \ldots = seq(X_k)$. Then $k\leq 2n^2 + 1$.
\end{lemma}

\begin{proof}
Clearly, every $SLupdate$ operation performs at most one $S.scan()$ operation.

Suppose $op$ is an $SLscan$ operation. 
Since $seq(X_1) = \ldots = seq(X_k)$, by Observation~\ref{obsEQ} $X_1 = \ldots = X_k$.
Let $X = X_1 = \ldots = X_k$.
Consider some $sc_i$ operation, with $i \in \{1, \ldots, k-1\}$. 
Operation $op$ must restart its main loop after $sc_i$, and so the condition on line~\ref{seq:CHcond} must hold at the end of the iteration of the main loop during which $sc_i$ is performed by $p$. There are two cases:

\begin{enumerate}[label=(\roman*)]
\item The vectors of values compared on line~\ref{seq:CHcond} are unequal. Then by the condition on line~\ref{seq:CHcond}, process $p$ performs an $R.DWrite_p(X)$ on line~\ref{SLSDWrite} following $sc_i$.
\item The Boolean flag returned by the $R.DRead_p()$ operation performed by $p$ on line~\ref{seq:DRead2} following $sc_i$ is $true$ (but the vectors of values compared on line~\ref{CHcond} are equal). Then, by the sequential specification of ABA-detecting registers, some $R.DWrite(X)$ operation occurs after the $R.DRead()$ operation performed by $p$ on line~\ref{seq:DRead1} prior to $sc_i$, and before the $R.DRead()$ operation performed by $p$ on line~\ref{seq:DRead2} following $sc_i$.
\end{enumerate}

In both cases, some $R.DWrite(X)$ operation happens during the iteration of the main loop in which $sc_i$ is performed. Lemma~\ref{lemmabetweenwrites}~\ref{lemmabetweenwrites:sequence} ensures that each process $p$ performs at most $2n$ $R.DWrite_p(X)$ operations. Since there are $n$ processes, we obtain $k - 1\leq 2n^2$, and hence $k \leq 2n^2 + 1$.
\end{proof}
\fi

\begin{theorem}\label{thm:slss-complexity}
  \begin{enumerate}[label=(\alph{*})]
    \item Each $SLupdate$ performs at most one $S.update$, one $S.scan$, and one \ifspringer \else \\ \fi $R.DWrite$ operation.\label{thm:slss-complexity:A}
    \item For any transcript that contains $u$ $SLupdate$ and $s$ $SLscan$ invocations, the total number of operation invocations on $S$ and $R$ during $SLscan$ operations is $O(s+n^3u)$.\label{thm:slss-complexity:B}
  \end{enumerate}
  In particular, the implementation is lock-free provided that $S$ and $R$ are.
\end{theorem}
\iffull
\begin{proof}
  Part~\ref{thm:slss-complexity:A} follows immediately from the pseudocode of Algorithm~\ref{slss}.
  
  We now prove part~\ref{thm:slss-complexity:B}.
  For any process let $s_p$ denote the number of $SLscan$ invocations by process $p$.
  Further, for $i\in\{1,\dots,s_p\}$ let $k_{p,i}$ denote the total number of times process $p$ calls $S.scan()$ in line~\ref{seq:SLSscan} during its $i$-th $SLscan()$ operation.
  From Lemma~\ref{lem:sequence_numbers_per_scan} we obtain $\sum_{i=1}^{s_p}k_{p,i}=O(un^2+s_p)$ for each process $p$.
  Using $s=\sum_p s_p$ we obtain that the total number of $S.scan()$ calls during all $SLscan()$ operations is
  \begin{displaymath}
    \sum_{p}\sum_{i=1}^{s_p}k_{p,i}
    =
    O\left(\sum_p un^2+s_p\right)
    =
    O(s+un^3).
  \end{displaymath}
\end{proof}
\fi
As mentioned previously, we can use any lock-free or wait-free linearizable snapshot implementation for $S$.
Instead of an atomic ABA-detecting register $R$, we can use the lock-free strongly linearizable one from
\ifthesis
\else
Section
\fi
\ref{slabasection}.
Thus, Theorems~\ref{thm:main-ABA}, \ref{stronglin}, and~\ref{thm:slss-complexity} yield Theorem~\ref{thm:main-snapshot}.

\ifthesis
\else
\subsection{Remarks}
\fi

In this section we presented the first lock-free strongly linearizable implementation of a snapshot object that requires only bounded space.
The time complexity of our implementation is unfortunate; even with a small number of processes, the use of our implementation seems impractical due to its $O(n^3)$ runtime.
However, when contention is low (i.e. there are few overlapping calls to $SLupdate$ and $SLscan$), our snapshot implementation performs reasonably well.
Notice that in a low-contention scenario, $SLscan$ operations are seldom forced to repeat their main loop, and therefore their runtime is dominated by the $S.scan$ call on line~\ref{SLSscan}.
Hence, choosing a reasonably efficient linearizable snapshot implementation for $S$ results in an efficient strongly linearizable snapshot.
However, when contention is high and the number of processes is sufficiently large (and $SLupdate$ is called sufficiently often), it becomes increasingly likely that $SLscan$ operations \emph{never} terminate, as they are interrupted infinitely often by $S.update$ or $R.DWrite$ calls performed during other concurrent $SLupdate$ or $SLscan$ operations.
A natural extension of this work would aim to develop more efficient implementations.

Previous strongly linearizable implementations of counters and max-registers (such as the lock-free modification of the max-register in \cite{PossImposs}) required an unbounded number of unbounded registers.
Our bounded snapshot implementation can be used to implement a lock-free strongly linearizable counter or max-register using only a bounded number of registers.
Note that these implementations still inherently require registers to store unbounded values, since the state space of both counters and (unbounded) max-registers is infinite.

\ifthesis
\else
\section{General Construction}\label{sectionGEN}
\fi
Aspnes and Herlihy \cite{generalwaitfree} defined the large class of \emph{simple} types, and demonstrated that any type in this class has a wait-free linearizable implementation from atomic multi-reader multi-writer registers.
Simple types require that any pair of operations either commute, or one overwrites the other (see below for a formal definition).
Algorithm~\ref{gen} depicts Aspnes and Herlihy's general wait-free linearizable implementation of an arbitrary simple type $\mathscr{T}$.
Processes communicate only through an atomic snapshot object, $root$ (which can be replaced with a wait-free linearizable implementation from registers).
Suppose $\mathscr{T}$ supports the set of invocation descriptions $\mathcal{O}$.
Then for every $invoke \in \mathcal{O}$, $invoke$ is implemented by the $execute(invoke)$ method.
Aspnes and Herlihy proved that Algorithm~\ref{gen} is linearizable with respect to the sequential specification of the simulated type $\mathscr{T}$.
We prove that it is in fact strongly linearizable, and thus it remains strongly linearizable if $root$ is a strongly linearizable snapshot object.
Thus, using our lock-free snapshot implementation from 
\ifthesis
\else
Section
\fi
\ref{sectionSLSS} for $root$ yields Theorem~\ref{thm:main-general-construction}.

Throughout this
\ifthesis
\else
section,
\fi
we only consider types $\mathscr{T} = (\mathcal{S}, s_0, \mathcal{O}, \mathcal{R}, \delta)$ such that, for every $s \in \mathcal{S}$ and for every $invoke \in \mathcal{O}$, $\delta(s, invoke)$ is defined.
Two sequential histories $H$ and $H'$ are \emph{equivalent} if, for any sequential history $S$, $H \circ S$ is valid if and only if $H' \circ S$ is valid.
The invocation events $inv(op_1)$ and $inv(op_2)$ \emph{commute} if, for all sequential histories $H$ such that $H \circ op_1$ and $H \circ op_2$ are valid, $H \circ op_1 \circ op_2$ and $H \circ op_2 \circ op_1$ are valid and equivalent.
The invocation event $inv(op_2)$ \emph{overwrites} $inv(op_1)$ if, for all sequential histories $H$ such that $H \circ op_1$ and $H \circ op_2$ are valid, $H \circ op_1 \circ op_2$ is valid and equivalent to $H \circ op_2$.
As a shorthand, we say $op_1$ commutes with (resp. overwrites) an operation $op_2$ if $inv(op_1)$ commutes with (resp. overwrites) $inv(op_2)$.
An invocation description $invoke_1$ commutes with (resp. overwrites) the invocation description $invoke_2$ if, for all invocation events $inv(op_1) = (O, invoke_1, id_1)$ and $inv(op_2) = (O, invoke_2, id_2)$, $inv(op_1)$ commutes with (resp. overwrites) $inv(op_2)$.
These properties allow us to describe the class of \emph{simple} types.

\begin{definition}
Let $\mathscr{T}$ be a type that supports a set of invocation descriptions $\mathcal{O}$. Then $\mathscr{T}$ is \emph{simple} if, for every pair of invocation descriptions $invoke_1, invoke_2 \in \mathcal{O}$, either $invoke_1$ and $invoke_2$ commute, or one overwrites the other.
\end{definition}


For the rest of this section, let $\mathscr{T} = (\mathcal{S}, s_0, \mathcal{O}, \mathcal{R}, \delta)$ be some simple type, and let $O$ be an object of type $\mathscr{T}$ implemented by Algorithm~\ref{gen}; that is, each invocation description $invoke \in \mathcal{O}$ is implemented by the $execute(invoke)$ method.
Let $\mathcal{T}$ be the set of all transcripts on $O$.
For ease of notation, for every operation $ex$ on $O$ let $invoc(ex)$ denote the invocation description of $ex$; that is, if $inv(ex) = (O, invoke, id)$, then $invoc(ex) = invoke$.

Algorithm~\ref{gen} maintains a representation of a graph in a shared snapshot object called $root$.
Each entry of the $root$ variable contains a reference to an instance of type $node$, which has three fields: $invocation$, $response$, and $preceding$.
The $invocation$ and $response$ fields contain an invocation description and response, respectively.
The $preceding$ field is an array containing $n$ references to nodes.
For a node $x$, $x.preceding[i]$ contains either $\bot$ or a pointer to a node $y$.

A \emph{precedence graph} $G = (V, E)$ of a history $H$ is a directed graph whose vertices are operations, such that for any $op_1, op_2 \in V$ there is a directed path of length at least 1 from $op_1$ to $op_2$ in $G$ if and only if $op_1 \xrightarrow{H} op_2$.
Notice that the happens-before relation is the transitive closure of $G$.

The following notion of \emph{dominance} is used to break ties between mutually overwriting operations.
\begin{definition}\label{defdominate}
An invocation event $inv(op_2)$ of process $p$ \emph{dominates} $inv(op_1)$ of process $q$ if either
\begin{enumerate}[label=(\arabic*)]
	\item $inv(op_2)$ overwrites $inv(op_1)$ but not vice-versa, or
	\item $inv(op_1)$ and $inv(op_2)$ overwrite each other and $p > q$.
\end{enumerate}
\end{definition}
An invocation description $invoke_2$ of process $p$ dominates $invoke_1$ of process $q$ if, for all pairs of invocation events $inv(op_1) = (O, invoke_1, id_1)$ and $inv(op_2) = (O, invoke_2, id_2)$, $inv(op_2)$ dominates $inv(op_1)$.
A \emph{linearization graph} $lingraph(G)$ is constructed by adding directed edges to the precedence graph $G$ as follows:
First, an arbitrary topological order $op_1,\dots,op_k$ of $G$ is fixed.
Then all pairs $(i,j)$, $1\leq i<j\leq k$, are considered in lexicographical order, and if one of the two invocation descriptions in $\{op_i,op_j\}$ dominates the other, an edge is added from the dominated operation to the dominating one, provided that edge does not close a cycle.
A precise description of the construction of a linearization graph is provided in the $lingraph$ method of Algorithm~\ref{gen}.

Suppose $T \in \mathcal{T}$ is a transcript on $O$.
Let $ex \in \Gamma(T)$ be an operation on $O$ by process $p$.
Process $p$ begins $ex$ by performing a $root.scan()$ operation on line~\ref{scan}.
This $root.scan()$ operation returns a vector $view$ of references to nodes.
In line~\ref{gen:constprec} process $p$ computes a precedence graph $G$ of operations on $O$ using a straightforward graph search, starting with the nodes stored in $view$ (we will explain the $precgraph$ method in more detail later).
It then calculates a sequential history $H$ on line~\ref{lin} by topologically sorting a linearization graph $lingraph(G)$.
Now $p$ constructs a new node $x$ that stores the invocation description $invoc(ex)$, and creates the invocation event $inv(op) = (O, invoc(ex), id)$, for some integer $id$.
Process $p$ then constructs a response event $rsp(op) = (resp, id)$ such that $H \circ inv(op) \circ rsp(op)$ is valid with respect to the sequential specification of $\mathscr{T}$ (lines~\ref{createentry}-\ref{getresp}).
The existence of the response $resp$ is guaranteed by our assumption that for every $s \in \mathcal{S}$ and every $invoke \in \mathcal{O}$, $\delta(s, invoke)$ is defined.
Note that the node referenced by $view[q]$ is the most recent node written to $root$ by process $q$ prior to the $root.scan()$ operation performed by $ex$.
For every process $q \in \{1, \ldots, n\}$, $p$ stores $view[q]$ in $x.preceding[q]$ on line~\ref{setprec}.
Finally, $p$ writes the address of the constructed node to the snapshot object on line~\ref{slupdate}.
For any operation $ex \in \Gamma(T)$ such that $ex^{\ref{slupdate}} \in T$, let $node(ex)$ be the node constructed by $ex$ (i.e. $node(ex)$ is the node whose address is written to $root$ during $ex^{\ref{slupdate}}$).
Since each operation instantiates a new node instance, and each node reference remains in the shared precedence graph representation forever, the algorithm uses unbounded space.

    \begin{algorithm}       
      \caption{\small Implementation of a simple type~\cite{generalwaitfree}}
      \label{gen}
        
		\SetKwProg{Fn}{Function}{:}{}
		\SetKwFor{Struct}{struct}{:}{}
		
		\nonl\Struct{node} {
			\nonl invocation description, $invocation \in \mathcal{O}$\;
			\nonl response, $response \in \mathcal{R}$\;
			\nonl pointers to nodes, $preceding[1 \ldots n]$\;
		}
        
        \nonl\textbf{shared}\;
        \nonl\quad atomic snapshot object $root = (null, \ldots, null)$\;        
        
        \nonl\;
        \nonl\Fn{lingraph$(G)$} {
        	let $op_1, \ldots, op_k$ be a topological sort of $G$\;
        	$L \gets G$\;
        	\For{$i \in \{1, \ldots, k-1\}$} {
        		\For{$j \in \{i+1, \ldots, k\}$} {
        			\If{$op_i$ dominates $op_j$ and adding $(op_j, op_i)$ to $L$ does not complete a cycle} {
        				add $(op_j, op_i)$ to $L$\;
        			}
        			\If{$op_j$ dominates $op_i$ and adding $(op_i, op_j)$ to $L$ does not complete a cycle} {
        				add $(op_i, op_j)$ to $L$\;
        			}
        		}
        	}
        	\Return $L$
        }

        \nonl\;
        \nonl\Fn{execute$_p(invoke)$}{
        	$view \gets root.scan()$\;	\label{scan}
        	$G \gets precgraph(view)$\;	\label{gen:constprec}
        	$H \gets$ topological sort of $lingraph(G)$\;	\label{lin}
        	initialize a new $node$ $e = \{\bot, \bot, \bot\}$\;	\label{createentry}
        	$e.invocation \gets invoke$\;		\label{invoc}
        	$inv(op) \gets (O, invoke, id)$\;	\label{invevent}
        	$rsp(op) \gets (resp, id)$ such that $H \circ inv(op) \circ rsp(op)$ is valid;	\label{getresp}
        	$e.response \gets resp$\;
        	\For{$i \in \{1, \ldots, n\}$}{					\label{precloop}
        		$e.preceding[i] \gets view[i]$		\label{setprec}
        	}
        	$root.update_p($address of $e)$\;		\label{slupdate}
        	\Return $e.response$
		}        
    \end{algorithm}

We proceed by first outlining an argument for the linearizability of Algorithm~\ref{gen}.
On line~\ref{getresp}, an operation ensures that it calculates a valid response with respect to the sequential history $H$ constructed on line~\ref{lin}.
Then assuming that $H$ is valid, the algorithm is correct as long as writing out a node constructed by a pending operation does not invalidate the response of some concurrent operation.
Suppose $ex_1$ and $ex_2$ are concurrent operations.
For this example, assume no other operations are concurrent with either $ex_1$ or $ex_2$.
Hence, both $ex_1$ and $ex_2$ construct the same precedence graph (call it $G$) on line~\ref{gen:constprec}.
For simplicity, assume that $ex_1$ and $ex_2$ both compute the same topological ordering (call it $H$) of $lingraph(G)$ on line~\ref{lin} (one of the key results of Aspnes and Herlihy \cite{generalwaitfree} is that every pair of topological orderings of any linearization graph are equivalent --- we state this result more formally in Section~\ref{sec:genproofsl}).
Hence, on lines \ref{invevent} and \ref{getresp} $ex_1$ and $ex_2$ construct operations $op_1$ and $op_2$, respectively, such that $H \circ op_1$ and $H \circ op_2$ are valid.
Since the simulated type is simple, either $invoc(ex_1)$ and $invoc(ex_2)$ commute, or one overwrites the other.
If $invoc(ex_1)$ and $invoc(ex_2)$ commute, then it does not matter which node is written to $root$ first, since $H \circ op_1 \circ op_2$ and $H \circ op_2 \circ op_1$ are both valid and equivalent, by definition of commutativity.
If $invoc(ex_2)$ overwrites $invoc(ex_1)$, then again it does not matter which operation writes its node to $root$ first, since any operation that views the precedence graph after both $ex_1$ and $ex_2$ have written their nodes to $root$ adds a dominance edge from $op_1$ to $op_2$ in the linearization graph.
Hence, in any topological ordering of this linearization graph $op_1$ occurs immediately before $op_2$ (recall our assumption that no other operations are concurrent with either $ex_1$ or $ex_2$).
Since $H \circ op_1 \circ op_2$ is valid by the definition of overwriting, the responses of both $op_1$ and $op_2$ are valid.
A symmetric argument applies if $invoc(ex_1)$ overwrites $invoc(ex_2)$.

We now outline the intuition behind our strong linearization function for Algorithm~\ref{gen}; a full proof of the strong linearizability of Algorithm~\ref{gen} is provided in Section~\ref{sec:genproofsl}.
Suppose $T \in \mathcal{T}$ is a transcript, and let $ex \in \Gamma(T)$ be an operation by process $p$.
After $p$ performs the $root.scan()$ operation on line~\ref{scan} during $ex$, the response of $ex$ is entirely determined, since it is chosen based on the contents of the precedence graph constructed from $view$ on line~\ref{gen:constprec}.
If during $ex$, an operation $ex'$ writes its constructed node to $root$, and $invoc(ex')$ dominates $invoc(ex)$, then $ex$ may be linearized immediately before $ex'$, since $ex$ no longer has any effect on the responses calculated by subsequent operations.
Hence, operations may only linearize when some node is written to $root$ on line~\ref{slupdate}, as it is entirely determined which operations linearize when a particular node is written to the graph (i.e. all concurrent operations that are dominated by the writing operation).

\iffull

\ifthesis
\section{Storing a Precedence Graph}\label{sec:storeprecgraph}
\else
\subsection{Storing a Precedence Graph}\label{sec:storeprecgraph}
\fi
In this section we prove that a precedence graph may be extracted from $root$. More specifically, we show that for any vector $view$ returned by a $root.scan()$ operation, $precgraph(view)$ returns a precedence graph of some history (the implementation of $precgraph$ is provided in Algorithm~\ref{alg:precgraph}).

The $precgraph(view)$ method begins by performing a $nodegraph(view)$ operation on line~\ref{precgraph:getnodegraph}, which returns a graph whose vertices are nodes.
A $nodegraph(view)$ operation begins with a straightforward graph search starting from the nodes present in $view$.
First, a process $p$ performing a $nodegraph(view)$ operation initializes an empty graph $G = (V, E)$ and an empty queue $queue$ (lines \ref{precgraph:initgraph} and \ref{precgraph:initqueue}).
Next, $p$ adds all nodes referenced in $view$ to both $queue$ and the vertex set $V$ during the loop on line~\ref{precgraph:initloop}.
The main loop of the $nodegraph$ method begins on line~\ref{precgraph:mainloop}, and continues until all nodes have been removed from $queue$.
During the main loop, $p$ first removes a node $node(ex)$ from $queue$ on line~\ref{precgraph:dequeue}.
For each node $node(ex')$ referenced in $node(ex).preceding$, $p$ adds the edge $\bigl(node(ex'), node(ex)\bigr)$ to $E$ on line~\ref{precgraph:addedge};
if $node(ex')$ is not present in $V$, then $p$ adds $node(ex')$ to $queue$ and $V$ on lines \ref{precgraph:addtoqueue} and \ref{precgraph:vaddmain}.
After the main while-loop has terminated, the $nodegraph(view)$ operation returns the computed graph $G$.
In the $precgraph(view)$ method, after the $nodegraph(view)$ operation has responded, process $p$ computes a topological ordering $node(ex_1), \ldots, node(ex_k)$ of the vertices in $G$.
Next, process $p$ initializes an array $id[1\ldots n] = [1, \ldots, 1]$ on line~\ref{precgraph:initid}.
Process $p$ then begins a for-loop on line~\ref{precgraph:opforloop}; during this loop, $p$ computes an operation on $O$ for each node present in $G$.
That is, for each $node(ex_i)$ present in the topological ordering of $G$, $p$ constructs an operation $op_i$, with $inv(op_i) = (O, node(ex_i).invocation), opid)$, and $rsp(op_i) = (node(ex_i), opid)$, where $opid$ is an integer.
The operation identifier $opid$ is computed on line~\ref{precgraph:opid}, where it is assigned the value $(id[p]\cdot n) + (p-1)$, where $p$ is the process that performed $ex_i$.
The final statement of the loop increments $id[p]$ by 1.
In general, the operation constructed for the $j$-th node in the topological ordering of $G$ which was written by process $p$ is assigned the identifier $(j \cdot n) + (p-1)$.
This way, each operation constructed during the loop is assigned a unique identifier (note that if $ex_j$ is performed by process $p$, and $inv(ex_j)$ has the identifier $opid$, then $opid \equiv (p-1) \pmod{n}$).
After the for-loop, on line~\ref{precgraph:replaceverts} $p$ replaces each $node(ex_i) \in V$ with $op_i$ (each edge is replaced similarly on line~\ref{precgraph:replaceedges}).

	\begin{algorithm}\label{alg:precgraph}
		\caption{Extraction of a precedence graph from a vector of node references.}
		\SetKwProg{Fn}{Function}{:}{}
		
		\nonl\Fn{precgraph$(view)$} {
			$G = (V, E) \gets nodegraph(view)$\;	\label{precgraph:getnodegraph}
			let $node(ex_1), \ldots, node(ex_k)$ be a topological sort of $G$\;	\label{precgraph:lettopsort}
			let $id[1\ldots n] = [1, \ldots, 1]$\;	\label{precgraph:initid}
			\For{$i \in \{1, \ldots, k\}$} {	\label{precgraph:opforloop}
				suppose $ex_i$ is by process $p$\;
				let $opid = (id[p]\cdot n) + (p-1)$\;	\label{precgraph:opid}
				let $inv(op_i) = (O, node(ex_i).invocation, opid)$\;	\label{precgraph:invevent}
				let $rsp(op_i) = (node(ex_i).response, opid)$\;	\label{precgraph:rspevent}
				$id[p] \gets id[p] + 1$\; \label{precgraph:idinc}
			}
        	replace $node(ex_i) \in V$ with $op_i$, for all $i \in \{1, \ldots, k\}$\;\label{precgraph:replaceverts}
        	replace each $\bigl(node(ex_i), node(ex_j)\bigr) \in E$ with $(op_i, op_j)$\;\label{precgraph:replaceedges}
        	\Return $G$
		}
		
		\nonl\;
        \nonl\Fn{nodegraph$(view)$} {
        	let $G = (V, E)$ be an empty graph\;	\label{precgraph:initgraph}
			let $queue$ be an empty queue\;				\label{precgraph:initqueue}
			\For{$v \in view$} {		\label{precgraph:initloop}
				\If{$v \neq null$} {
					let $node(ex)$ be the node addressed by $v$\;
					enqueue $node(ex)$ to $queue$\;
					$V \gets V \cup \{node(ex)\}$\;	\label{precgraph:vaddinit}
				}
			}
			\While{$queue$ is not empty} {	\label{precgraph:mainloop}
        		dequeue $node(ex)$ from $queue$\;	\label{precgraph:dequeue}
        		\For{each $node(ex')$ referenced in $node(ex).preceding$} {	\label{precgraph:neighbours}
        			$E \gets E \cup \bigl\{(node(ex'), node(ex))\bigr\}$\; \label{precgraph:addedge}
        			\If{$node(ex') \not\in V_N$} {	\label{precgraph:checkinv}
        				enqueue $node(ex')$ to $queue$\; \label{precgraph:addtoqueue}
        				$V \gets V \cup \{node(ex')\}$\;	\label{precgraph:vaddmain}
        			}
        		}
        	}			\label{precgraph:endloop}
        	\Return $G$
       }
    \end{algorithm}

Let $T \in \mathcal{T}$ be a transcript, and let $\mathcal{G}(T) = (V_T, E_T)$  be a graph such that
\begin{align*}
\ifspringer \small \else\fi
	\text{$V_T$} &\small\text{$\;= \{node(ex)\, :\, ex^{\ref{slupdate}} \in T\}$, and} \\
	\small\text{$E_T$} &\small\text{$\,= \Bigl\{\bigl(node(ex_1), node(ex_2)\bigr)\, :\, node(ex_1), node(ex_2) \in V_T$} \\ &\small\text{$\qquad \wedge \; \exists p \in \{1, \ldots, n\} \text{ s.t. } node(ex_2).preceding[p] =$} \\ &\small\text{$\qquad$address of $node(ex_1)\Bigr\}$.}
\end{align*}
We will show that if $root$ contains the vector $view$ after all steps in $T$ are performed, then $nodegraph(view) = \mathcal{G}(T)$.

\begin{observation}\label{obs:edgeearlier}
	If $\bigl(node(ex_1), node(ex_2)\bigr) \in E_T$, \ifspringer \\ \else\fi then $time(ex_1^{\ref{slupdate}}) < time(ex_2^{\ref{scan}})$.
\end{observation}

\begin{proof}
	Since $\bigl(node(ex_1), node(ex_2)\bigr) \in E_T$, by the definition of $E_T$ $node(ex_2).preceding[p]$ must contain the address of $node(ex_1)$, for some process $p$.
	On line~\ref{createentry} of Algorithm~\ref{gen}, each operation initializes its own node, and only modifies the fields of that node.
	Hence, only $ex_2$ modifies the $preceding$ field of $node(ex_2)$.
	Thus, $ex_2$ must place the address of $node(ex_1)$ into $node(ex_2).preceding[p]$ when it performs line~\ref{setprec} (during the $p$-th iteration of the for-loop on line~\ref{precloop}).
	Therefore, the address of $node(ex_1)$ must be in $view[p]$, where $view$ is the vector returned by $ex_2^{\ref{scan}}$.
	Then by the sequential specification of snapshot objects, $time(ex_1^{\ref{slupdate}}) < time(ex_2^{\ref{scan}})$.
\end{proof}

\begin{observation}\label{obs:pathearlier}
	If there is a path of length at least 1 from $node(ex_1)$ to $node(ex_2)$ in $\mathcal{G}(T)$, then $time(ex_1^{\ref{slupdate}}) < time(ex_2^{\ref{scan}})$.
\end{observation}

\begin{proof}
	If the length of the path from the vertex $node(ex_1)$ to the vertex $node(ex_2)$ is 1, then the observation immediately follows from Observation~\ref{obs:edgeearlier}.
	Now
	\ifspringer
	\begin{equation}\label{obs:pathearlier:indhyp}
	\begin{split}
	&\text{assume that the observation holds for any path of} \\ &\text{length $k \geq 1$.}
	\end{split}
	\end{equation}
	\else
	\begin{equation}\label{obs:pathearlier:indhyp}
	\text{assume that the observation holds for any path of length $k \geq 1$.}
	\end{equation}
	\fi
	Suppose that there is a path of length $k+1$ from $node(ex_1)$ to $node(ex_2)$.
	Let $node(ex_\ell)$ be the second-last node on this path.
	So
	\ifspringer
	\begin{gather}
	\begin{split}
	&\text{there is a path from $node(ex_1)$ to $node(ex_\ell)$ of} \\ &\text{length $k$, and} \label{obs:pathearlier:kpath} \end{split}\\ 
	\text{$\bigl(node(ex_\ell), node(ex_2)\bigr) \in E_T$.}\label{obs:pathearlier:edge}
	\end{gather}
	\else
	\begin{gather}
	\text{there is a path from $node(ex_1)$ to $node(ex_\ell)$ of length $k$, and} \label{obs:pathearlier:kpath} \\ 
	\text{$\bigl(node(ex_\ell), node(ex_2)\bigr) \in E_T$.}\label{obs:pathearlier:edge}
	\end{gather}
	\fi
	By (\ref{obs:pathearlier:indhyp}) and (\ref{obs:pathearlier:kpath}),
	\begin{equation}
		time(ex_1^{\ref{slupdate}}) < time(ex_\ell^{\ref{scan}}).\label{obs:pathearlier:1beforeell}
	\end{equation}
	By (\ref{obs:pathearlier:edge}) and Observation~\ref{obs:edgeearlier},
	\begin{equation}
		time(ex_\ell^{\ref{slupdate}}) < time(ex_2^{\ref{scan}}).\label{obs:pathearlier:ellbefore2}
	\end{equation}
	By the pseudocode of the $execute$ method in Algorithm~\ref{gen}, $time(ex_\ell^{\ref{scan}}) < time(ex_\ell^{\ref{slupdate}})$.
	This, along with (\ref{obs:pathearlier:1beforeell}) and (\ref{obs:pathearlier:ellbefore2}), imply that $time(ex_1^{\ref{slupdate}}) < time(ex_2^{\ref{scan}})$.
	Hence, the observation follows by induction.
\end{proof}

\begin{observation}\label{obs:procseq}
	For some process $p$, let $ex_1 \circ \ldots \circ ex_k$ be the longest prefix of $\Gamma(T)|p$ with $ex_k^{\ref{slupdate}} \in T$.
	Then for any $i, j \in \{1, \ldots, k\}$ such that $i < j$, there is a path of length $j - i$ from $node(ex_i)$ to $node(ex_j)$.
\end{observation}

\begin{proof}
	It suffices to show that, for any $i \in \{1, \ldots, k-1\}$, $\bigl(node(ex_i), node(ex_{i+1})\bigr) \in E_T$.
	Since $ex_i$ and $ex_{i+1}$ are both performed by $p$, they are performed in sequence.
	Hence, $time(ex_i^{\ref{slupdate}}) < time(ex_{i+1}^{\ref{scan}})$, and since each operation performs only a single $root.update$ operation (line~\ref{slupdate}), there is no $root.update$ operation by $p$ that happens in $\bigl(time(ex_i^{\ref{slupdate}}), time(ex_{i+1}^{\ref{scan}})\bigr)$.
	Then if $view$ is the vector returned by $ex_{i+1}^{\ref{scan}}$, $view[p]$ contains the address of $node(ex_i)$.
	Therefore, $ex_{i+1}$ places the address of $node(ex_i)$ into $node(ex_{i+1}).preceding[p]$ on line~\ref{setprec} (during the $p$-th iteration of the for-loop on line~\ref{precloop}).
	Then by the definition of $E_T$, $\bigl(node(ex_i), node(ex_{i+1})\bigr) \in E_T$.
\end{proof}

\begin{observation}\label{obs:nodeseq}
	Let $ex_q, ex_p \in \Gamma(T)$ be operations by processes $q$ and $p$, respectively, with $ex_q^{\ref{slupdate}}, ex_p^{\ref{slupdate}} \in T$. If $time(ex_q^{\ref{slupdate}}) < time(ex_p^{\ref{scan}})$, then there is a path from $node(ex_q)$ to $node(ex_p)$ in $\mathcal{G}(T)$.
\end{observation}

\begin{proof}
	If $q = p$, then the observation statement follows from Observation~\ref{obs:procseq}.
	
	Now assume $q \neq p$.
	Suppose $ex_{q,1}, \ldots, ex_{q,k}$ is the sequence of $execute$ operations by $q$ in $\Gamma(T)$.
	Let $ex_{q,i}$ be the final operation in this sequence with $time(ex_{q,i}^{\ref{slupdate}}) < time(ex_p^{\ref{scan}})$ (the existence of $ex_{q,i}$ is guaranteed by the observation assumption).
	That is,
	\ifspringer
	\begin{equation}
	\begin{split}
		&\text{there is no $j \in \{i+1, \ldots, k\}$ such that} \\ &\text{$time(ex_{q,j}^{\ref{slupdate}}) < time(ex_p^{\ref{scan}})$.}\label{obs:nodeseq:lateproc}
	\end{split}
	\end{equation}
	\else
	\begin{equation}
		\text{there is no $j \in \{i+1, \ldots, k\}$ such that $time(ex_{q,j}^{\ref{slupdate}}) < time(ex_p^{\ref{scan}})$.}\label{obs:nodeseq:lateproc}
	\end{equation}
	\fi
	Let $view$ be the vector returned by $ex_p^{\ref{scan}}$.
	By (\ref{obs:nodeseq:lateproc}) and the fact that $time(ex_{q,i}^{\ref{slupdate}}) < time(ex_p^{\ref{scan}})$,
	\begin{equation}
		\text{$view[q]$ contains the address of $node(ex_{q,i})$.}\label{obs:nodeseq:address}
	\end{equation}
	By (\ref{obs:nodeseq:address}), $ex_p$ places the address of $node(ex_{q,i})$ into $node(ex_p).preceding[q]$ on line~\ref{setprec}.
	Then by definition of $E_T$,
	\begin{equation}
		\bigl(node(ex_{q,i}), node(ex_p)\bigr) \in E_T.\label{obs:nodeseq:edge}
	\end{equation}
	If $ex_q = ex_{q,i}$, then the observation statement is implied by (\ref{obs:nodeseq:edge}).
	Otherwise, if $ex_q = ex_{q,j} \neq ex_{q,i}$, then $j < i$ by (\ref{obs:nodeseq:lateproc}) along with the observation assumption that $time(ex_q^{\ref{slupdate}}) < time(ex_p^{\ref{scan}})$.
	Hence, by Observation~\ref{obs:procseq}, there is a path from $node(ex_q)$ to $node(ex_{q,i})$ in $\mathcal{G}(T)$.
	By this and (\ref{obs:nodeseq:edge}), there is a path from $node(ex_q)$ to $node(ex_p)$ in $\mathcal{G}(T)$.
\end{proof}

\begin{observation}\label{obs:rootprecgraph}
	Suppose $T \in \mathcal{T}$ is a transcript. Suppose that, after all steps in $T$ are performed in order, $root$ contains the vector $view$. Let $G = (V, E) = nodegraph(view)$. Then for every operation $ex \in \Gamma(T)$ such that $ex^{\ref{slupdate}} \in T$, 
	\begin{enumerate}[label=(\alph*)]	
		\item $node(ex) \in V$, and\label{obs:rootprecgraph:A}
		\item for every $node(ex')$ referenced in $node(ex).preceding$, $\bigl(node(ex'), node(ex)\bigr) \in E$.\label{obs:rootprecgraph:B}
	\end{enumerate}
	That is, $nodegraph(view) = \mathcal{G}(T)$.
\end{observation}

\begin{proof}
	To prove part~\ref{obs:rootprecgraph:A}, we show that for every process $p$, each node whose address is written to $root$ during $T$ is added to $V$ in the calculation of $nodegraph(view)$.
	Let $ex_1, \ldots, ex_k \in \Gamma(T)$ be the (nonempty) sequence of operations by $p$ such that $ex_i^{\ref{slupdate}} \in T$, for every $i \in \{1, \ldots, k\}$.
	(If there are no such operations by process $p$, then $V$ trivially contains all nodes written to $root$ by $p$ in $T$.)
	Note that $view[p]$ contains a reference to $node(ex_k)$.
	Hence, in the computation of $nodegraph(view)$, $node(ex_k)$ is added to $V$ on line~\ref{precgraph:vaddinit}.
	Now
	\ifspringer
	\begin{equation}
	\begin{split}
		&\text{assume that for some $j \in \{2, \ldots, k\}$,} \\ &\text{$node(ex_j) \in V$.}\label{obs:precgraph:ih}
	\end{split}
	\end{equation}
	\else
	\begin{equation}
		\text{assume that for some $j \in \{2, \ldots, k\}$, $node(ex_j) \in V$.}\label{obs:precgraph:ih}
	\end{equation}
	\fi
	We proceed by showing that $node(ex_{j-1}) \in V$.
	By Observation~\ref{obs:procseq}, there exists a path of length 1 (i.e. an edge) from $node(ex_{j-1})$ to $node(ex_j)$ in $\mathcal{G}(T)$.
	Then
	\ifspringer
	\begin{equation}
	\begin{split}
		&\text{$node(ex_j).preceding[p]$ contains a reference to} \\ &\text{$node(ex_{j-1})$.}\label{obs:precgraph:inprec}
	\end{split}
	\end{equation}
	\else
	\begin{equation}
		\text{$node(ex_j).preceding[p]$ contains a reference to $node(ex_{j-1})$.}\label{obs:precgraph:inprec}
	\end{equation}
	\fi
	Due to (\ref{obs:precgraph:ih}), in the computation of $nodegraph(view)$, $node(ex_j)$ must have been added to $V$ on either line~\ref{precgraph:vaddinit} or line~\ref{precgraph:vaddmain}.
	In either case, $node(ex_j)$ is added to $queue$ on the previous line.
	Since the while-loop on line~\ref{precgraph:mainloop} does not terminate until $queue$ is empty, $node(ex_j)$ is eventually dequeued from $queue$ on line~\ref{precgraph:dequeue}.
	By (\ref{obs:precgraph:inprec}), the for-loop on line~\ref{precgraph:neighbours} eventually reaches $node(ex_{j-1})$.
	Hence, at the if-statement on line~\ref{precgraph:checkinv}, either $node(ex_{j-1}) \in V$ already, or $node(ex_{j-1})$ is added to $V$ on line~\ref{precgraph:vaddmain}.
	Therefore, $node(ex_{j-1}) \in V$, and by induction we conclude that $node(ex_1), \ldots, node(ex_k) \in V$.
	This completes the proof of part~\ref{obs:rootprecgraph:A}.
	
	Due to \ref{obs:rootprecgraph:A}, and the fact that each node is added to $queue$ before it is added to $V$ during the calculation of $nodegraph(view)$,
	\ifspringer
	\begin{equation}
	\small
	\begin{split}
		&\text{for every $ex \in \Gamma(T)$ such that $ex^{\ref{slupdate}} \in T$, $node(ex)$ is} \\ &\text{dequeued from $queue$ on line~\ref{precgraph:dequeue} at some point} \\ &\text{during the computation of $nodegraph(view)$.}\label{obs:precgraph:dequeue}
	\end{split}
	\end{equation}
	\else
	\begin{equation}
	\begin{split}
		&\text{for every $ex \in \Gamma(T)$ such that $ex^{\ref{slupdate}} \in T$, $node(ex)$ is dequeued from $queue$ on line~\ref{precgraph:dequeue}} \\ &\text{at some point during the computation of $nodegraph(view)$.}\label{obs:precgraph:dequeue}
	\end{split}
	\end{equation}
	\fi
	When $node(ex)$ is dequeued from $queue$ on line~\ref{precgraph:dequeue}, for every $node(ex')$ that is referenced in $node(ex).preceding$, the edge $\bigl(node(ex'), node(ex)\bigr)$ is added to $E$ (if it is not already present in $E$) on line~\ref{precgraph:addedge};
	part~\ref{obs:rootprecgraph:B} follows from this along with (\ref{obs:precgraph:dequeue}).
\end{proof}

Let $T \in \mathcal{T}$ be a transcript such that $root$ contains the vector $view$ after all steps of $T$ are performed in order.
Let $oper(ex)$ denote the operation on $O$ with $inv\bigl(oper(ex)\bigr) = \bigl(O, node(ex).invocation, (j\cdot n) + (p-1)\bigr)$ and $rsp\bigl(oper(ex)\bigr) = \bigl(node(ex).response, (j\cdot n) + (p-1)\bigr)$, where $ex$ is the $j$-th operation by $p$ in $\Gamma(T)$.

\begin{observation}\label{obs:operations}
	Consider the operation $op_i$ computed from $node(ex_i)$ on lines \ref{precgraph:invevent} and \ref{precgraph:rspevent} during the computation of $precgraph(view)$, for any $i \in \{1, \ldots, k\}$.
	Then $op_i = oper(ex_i)$.
\end{observation}
\begin{proof}
Suppose $ex_i$ is the $j$-th operation performed by $p$ in $\Gamma(T)$.
Let the $nodegraph(view)$ operation on line~\ref{precgraph:getnodegraph} return the graph $G$.
Consider the topological ordering $node(ex_1), \cdots, node(ex_k)$ of $G$ constructed on line~\ref{precgraph:lettopsort}; let $node(ex_{p,1}), \ldots, node(ex_{p,\ell})$ be the subsequence of $node(ex_1), \cdots, node(ex_k)$ consisting of all and only those nodes constructed by $p$.
Observation~\ref{obs:procseq} implies that $ex_i = ex_{p,j}$, as otherwise there would be a backwards edge among the sequence $node(ex_{p,1}), \ldots, node(ex_{p,\ell})$ in $G$.
Hence, $node(ex_i)$ is the $j$-th node by process $p$ that is encountered during the for-loop on line~\ref{precgraph:opforloop} (that is, each of $node(ex_{p,1}), \ldots, node(ex_{p,j-1})$ is encountered during the for-loop prior to $node(ex_{p,j}) = node(ex_i)$).
On line~\ref{precgraph:invevent}, the invocation description of $op_i$ is set to $node(ex_i).invocation$, and on line~\ref{precgraph:rspevent}, the response of $op_i$ is set to $node(ex_i).response$.
Now consider $opid$, the operation identifier calculated for $op_i$ on line~\ref{precgraph:opid}.
Since $ex_i$ is the $j$-th node by process $p$ that is encountered during the for-loop on line~\ref{precgraph:opforloop}, $id[p]$ has been incremented $j-1$ times by the time $op_i$ is constructed.
Hence, $opid = (id[p]\cdot n) + (p-1) = (j\cdot n) + (p-1)$.
Therefore, $inv(op_i) = \bigl(O, node(ex_i).invocation, (j\cdot n) + (p-1)\bigr) = inv\bigl(oper(ex_i)\bigr)$, and $rsp(op_i) = \bigl(node(ex_i).response, (j\cdot n) + (p-1)\bigr) = rsp\bigl(oper(ex_i)\bigr)$.
\end{proof}

By Observation~\ref{obs:operations} and the replacements performed on line~\ref{precgraph:replaceverts}, if $V$ is the vertex set of the graph returned by $precgraph(view)$, then $V = \{oper(ex)\; :\; node(ex)$ is in $\mathcal{G}(T)\}$.
We will use this fact without referencing Observation~\ref{obs:operations} for the remainder of the section.
Let $H$ be a history on an object $O$ of type $\mathscr{T}$ obtained from $T|root$ by doing the following for each step $t$ of $T|root$:
\begin{enumerate}[label=(\roman*)]
	\item If $t = time(ex^{\ref{scan}})$ for some operation $ex$ by process $p$, then replace this step by the invocation event $inv\bigl(oper(ex)\bigr)$ in $H$.
	\item If $t = time(ex^{\ref{slupdate}})$ for some operation $ex$ by process $p$, then replace this step by the response event $rsp\bigl(oper(ex)\bigr)$ in $H$.
	\item If $t = time\bigl(inv(op)\bigr)$, where $op$ is any $root.scan$ or $root.update$ operation, remove this step from $H$.
\end{enumerate}

\begin{lemma}
	The graph $precgraph(view)$ is a precedence graph of the history $H$.
\end{lemma}

\begin{proof}
Suppose $oper(ex_1) \xrightarrow{H} oper(ex_2)$.
By definition of happens-before order, this implies that $rsp\bigl(oper(ex_1)\bigr)$ occurs before $inv\bigl(oper(ex_2)\bigr)$ in $H$.
Then by the construction rules for $H$, $time_{T|root}(ex_1^{\ref{slupdate}}) < time_{T|root}(ex_2^{\ref{scan}})$, and hence $time_T(ex_1^{\ref{slupdate}}) < time_T(ex_2^{\ref{scan}})$.
Therefore, there is a path from $node(ex_1)$ to $node(ex_2)$ in $\mathcal{G}(T)$ by Observation~\ref{obs:nodeseq}.
Then by the replacements performed on lines \ref{precgraph:replaceverts} and \ref{precgraph:replaceedges}, there is a path from $oper(ex_1)$ to $oper(ex_2)$ in $precgraph(view)$.

Now suppose there is a path from $oper(ex_1)$ to $oper(ex_2)$ in $precgraph(view)$.
Then there is a path from $node(ex_1)$ to $node(ex_2)$ in $\mathcal{G}(T)$.
By Observation~\ref{obs:pathearlier}, $time_T(ex_1^{\ref{slupdate}}) < time_T(ex_2^{\ref{scan}})$.
This immediately implies that
\begin{equation}
	time_{T|root}(ex_1^{\ref{slupdate}}) < time_{T|root}(ex_2^{\ref{scan}}).\label{precgraphofhistory}
\end{equation}
By the construction rules for $H$, $rsp(ex_1^{\ref{slupdate}})$ is replaced with $rsp\bigl(op_1\bigr)$ in $H$, and $rsp(ex_2^{\ref{scan}})$ is replaced with $inv\bigl(oper(ex_2)\bigr)$ in $H$.
By this and (\ref{precgraphofhistory}), $rsp\bigl(oper(ex_1)\bigr)$ occurs before $inv\bigl(oper(ex_2)\bigr)$ in $H$, so $oper(ex_1) \xrightarrow{H} oper(ex_2)$.
\end{proof}

Throughout the remainder of this section, for any transcript $T \in \mathcal{T}$ such that $root$ contains the vector $view$ after all steps in $T$ are performed in order, we refer to $precgraph(view)$ as the precedence graph \emph{induced by} $\mathcal{G}(T)$.

For operations $oper(ex_1), oper(ex_2)$ in a precedence graph $G$, we say $oper(ex_1)$ \emph{precedes} $oper(ex_2)$ in $G$ if and only if there is a path from $oper(ex_1)$ to $oper(ex_2)$ in $G$.
If there is no path between operations $oper(ex_1)$ and $oper(ex_2)$ in $G$, then we say $oper(ex_1)$ and $oper(ex_2)$ are \emph{concurrent} in $G$.
If $T \in \mathcal{T}$ is a transcript such that $G$ is the precedence graph induced by $\mathcal{G}(T)$, then following two observations are immediate from Observation~\ref{obs:pathearlier} and Observation~\ref{obs:nodeseq}, respectively:

\begin{observation}\label{obs:opprecede}
	An operation $oper(ex_1)$ precedes \ifspringer operation \else\fi $oper(ex_2)$ in $G$ if and only if $time(ex_1^{\ref{slupdate}}) < time(ex_2^{\ref{scan}})$.
\end{observation}

\begin{observation}\label{obs:opconcurrent}
	The operations $oper(ex_1), oper(ex_2)$ are concurrent in $G$ if and only if
	\begin{enumerate}[labelindent=0pt,labelwidth=\widthof{\ref{def:interp2}},itemindent=1em,leftmargin=!,label=(\roman*)]
		\item $time(ex_2^{\ref{scan}}) < time(ex_1^{\ref{slupdate}})$, and
		\item $time(ex_1^{\ref{scan}}) < time(ex_2^{\ref{slupdate}})$.
	\end{enumerate}
\end{observation}

\ifthesis
\section{Proof of Strong Linearizability}\label{sec:genproofsl}
\else
\subsection{Proof of Strong Linearizability}\label{sec:genproofsl}
\fi

Notice that we cannot use Aspnes and Herlihy's linearization function (i.e. topological orderings of linearization graphs) to prove strong linearizability, since this function is not prefix preserving.
This is because operations may be written to the ``middle'' of the linearization graph; we clarify this argument with an example.
Consider a transcript $T$ containing only the operations $ex$ and $ex'$. Suppose $invoc(ex')$ dominates $invoc(ex)$.
While $ex$ is pending, suppose $ex'$ writes its constructed node to $root$.
Let $T'$ be the prefix of $T$ that ends with the response of the $root.update$ operation by $ex'$ on line~\ref{slupdate}.
Suppose that $ex$ completes its operation in $T$.
Let $G'$ be the precedence graph induced by $\mathcal{G}(T')$, and let $G$ be the precedence graph induced by $\mathcal{G}(T)$.
Then $lingraph(G')$ only contains $oper(ex')$, while $lingraph(G)$ contains both operations along with a dominance edge from $oper(ex)$ to $oper(ex')$ .
Hence, $oper(ex')$ is the only topological ordering of $lingraph(G')$, while $oper(ex) \circ oper(ex')$ is the only topological ordering of $lingraph(G)$.
This demonstrates how operations can unavoidably be written to the middle of the linearization order.
For this reason, we must define our own linearization function for Algorithm~\ref{gen}.

Let $T \in \mathcal{T}$ be some transcript.
Let $ex_1, \ldots, ex_k$ be a sequence consisting of all operations in $\Gamma(T)$, such that for any $i, j \in \{1, \ldots, k\}$ with $i < j$, $invoc(ex_j)$ does not dominate $invoc(ex_i)$ (such an ordering is guaranteed to exist, since dominance is a strict partial order).
For every $i \in \{1, \ldots, k\}$, let $pt_T(ex_i)$ be defined inductively as follows:

\begin{enumerate}[labelindent=0pt,labelwidth=\widthof{\ref{def:interp2}},itemindent=1em,leftmargin=!,label=\textbf{I-\arabic*},ref={I-\arabic*}]
	\item If $i = 1$, then define
	\begin{equation*}
		pt_T(ex_i) = time(ex_i^{\ref{slupdate}}).
	\end{equation*}\label{def:base}
	\item If $i > 1$, then let $Dom_i$ be the set of all operations $ex_h$ such that $h < i$ and $invoc(ex_h)$ dominates $invoc(ex_i)$. Define
	\ifspringer
	\begin{equation*}
	\small
	\begin{split}
		\text{$pt_T(ex_i) = \min \Bigl($} &\text{$\bigl\{pt(ex_h)\; :\; ex_h \in Dom_i \wedge time(ex_i^{\ref{scan}}) <$} \\ &\text{$pt(ex_h)\bigr\} \cup \bigl\{time(ex_i^{\ref{slupdate}})\bigr\}\Bigr)$.}
	\end{split}
	\end{equation*}\label{def:inductive}
	\else
	\begin{equation*}
		\text{$pt_T(ex_i) = \min \Bigl($} \text{$\bigl\{pt(ex_h)\; :\; ex_h \in Dom_i \wedge time(ex_i^{\ref{scan}}) <$ $pt(ex_h)\bigr\} \cup \bigl\{time(ex_i^{\ref{slupdate}})\bigr\}\Bigr)$.}
	\end{equation*}\label{def:inductive}
	\fi
\end{enumerate}

In \ref{def:inductive}, we emphasize that $Dom_i \subseteq \{ex_1, \ldots ex_{i-1}\}$; this ensures that \ref{def:base} and \ref{def:inductive} together form a proper inductive definition.
However, note that if $ex_h$ dominates $ex_i$, then it is guaranteed that $h < i$ since $ex_1, \ldots, ex_k$ is ordered by dominance.
That is, every operation whose invocation dominates $invoc(ex_i)$ must be earlier in the sequence $ex_1, \ldots, ex_k$ than $ex_i$.
Therefore, the following property is guaranteed for any $ex \in \Gamma(T)$:

\ifspringer
\begin{equation}\label{ptdefinition}
\small
\begin{split}
pt_T(ex) = \min \Bigl(&\bigl\{time(ex^{\ref{slupdate}})\bigr\} \cup \bigl\{pt(ex')\; :\; \\ &\text{$invoc(ex')$ dominates $invoc(ex)$}\;\wedge \\ &time(ex^{\ref{scan}}) < pt(ex')\bigr\}\Bigr).
\end{split}
\end{equation}
\else
\begin{equation}\label{ptdefinition}
\begin{split}
pt_T(ex) = \min \Bigl(&\bigl\{time(ex^{\ref{slupdate}})\bigr\} \cup \\ &\bigl\{pt(ex')\; :\; \text{$invoc(ex')$ dominates $invoc(ex)$}\;\wedge\; time(ex^{\ref{scan}}) < pt(ex')\bigr\}\Bigr).
\end{split}
\end{equation}
\fi
To simplify our proofs, we decompose property (\ref{ptdefinition}) into the following two rules:
\begin{enumerate}[labelindent=0pt,labelwidth=\widthof{\ref{def:interp2}},itemindent=1em,leftmargin=!,label=\textbf{J-\arabic*},ref={J-\arabic*}]
	\item If there is an operation $ex' \in \Gamma(T)$ such that $invoc(ex')$ dominates $invoc(ex)$ and $time(ex^{\ref{scan}}) < pt_T(ex') < pt_T(ex^{\ref{slupdate}})$, then let $ex_0$ be such an operation with minimal $pt_T$ value. Then $pt_T(ex) = pt_T(ex_0)$.\label{gen:ptrdom}
	\item If no such operation $ex'$ exists, then $pt_T(ex) = time(ex^{\ref{slupdate}})$. Recall that if $ex^{\ref{slupdate}} \not\in T$, then $time(ex^{\ref{slupdate}}) = \infty$.\label{gen:ptrself}
\end{enumerate}
If $pt_T(ex) \neq \infty$ for some operation $ex$, then we say $ex$ \emph{linearizes} at $pt_T(ex)$.
Also, if $T$ is clear from context, then we shorten $pt_T(ex)$ to $pt(ex)$.

Let $f(T)$ be a sequential history consisting of all and only those operations $ex_1, ex_2 \in \Gamma(T)$ by processes $q$ and $p$, respectively, with $pt(ex_1) \neq \infty$ and $pt(ex_2) \neq \infty$, such that $ex_1 \xrightarrow{f(T)} ex_2$ if and only if
\begin{enumerate}[labelindent=0pt,labelwidth=\widthof{\ref{def:interp2}},itemindent=1em,leftmargin=!,label=\textbf{K-\arabic*},ref={K-\arabic*}]
	\item operation $ex_1$ linearizes before $ex_2$ (i.e. $pt(ex_1) < pt(ex_2)$), or\label{gen:rulelin}
	\item operations $ex_1$ and $ex_2$ linearize at the same time (i.e. $pt(ex_1) = pt(ex_2)$) and $invoc(ex_2)$ dominates $invoc(ex_1)$, or\label{gen:ruledom}
	\item operations $ex_1$ and $ex_2$ linearize at the same time (i.e. $pt(ex_1) = pt(ex_2)$), $invoc(ex_2)$ does not dominate $invoc(ex_1)$, $invoc(ex_1)$ does not dominate $invoc(ex_2)$, and $q < p$.\label{gen:rulecomm}
\end{enumerate}
Since dominance is a strict partial order \cite{generalwaitfree}, \ref{gen:ruledom} and \ref{gen:rulecomm} impose a total order on the set of operations that linearize at any step of $T$.

\begin{observation}\label{obs:genptin}
	For every $ex \in \Gamma(T)$, $pt(ex)$ is in the interval $\bigl(time(ex^{\ref{scan}}), time(ex^{\ref{slupdate}})\bigr]$.
\end{observation}

\begin{proof}
If $pt(ex)$ satisfies \ref{gen:ptrdom}, then $pt(ex)$ is in the interval $\bigl(time(ex^{\ref{scan}}), time(ex^{\ref{slupdate}})\bigr)$ explicitly. Otherwise, if $pt(ex)$ satisfies \ref{gen:ptrself}, then $pt(ex) = time(ex^{\ref{slupdate}})$.
\end{proof}

\begin{observation}\label{obs:genptltdom}
	Let $ex_1 \in \Gamma(T)$ be an operation. If there exists an $ex_2 \in \Gamma(T)$ such that $time(ex_1^{\ref{scan}}) < pt(ex_2)$ and $invoc(ex_2)$ dominates $invoc(ex_1)$, then \ifspringer \\ \else\fi $pt(ex_1) \leq pt(ex_2)$.
\end{observation}

\begin{proof}
This is trivial if $pt(ex_2) = \infty$. If $time(ex_1^{\ref{slupdate}}) \leq pt(ex_2)$, then the observation follows from Observation~\ref{obs:genptin}.

Suppose $pt(ex_2) < time(ex_1^{\ref{slupdate}})$. Let $ex_3 \in \Gamma(T)$ be an operation that satisfies the following statements: 
\begin{enumerate}[labelindent=0pt,labelwidth=\widthof{\ref{def:interp2}},itemindent=1em,leftmargin=!,label=(\roman*)]
	\item $time(ex_1^{\ref{scan}}) < pt(ex_3) < time(ex_1^{\ref{slupdate}})$,\label{obs:lemmaroman1}
	\item $invoc(ex_3)$ dominates $invoc(ex_1)$, and\label{obs:lemmaroman2}
	\item $invoc(ex_3)$ has the lowest possible $pt$ value of any operation that satisfies \ref{obs:lemmaroman1} and \ref{obs:lemmaroman2}.
\end{enumerate}
Hence, $pt(ex_3) \leq pt(ex_2)$, and $pt(ex_1) = pt(ex_3)$ by \ref{gen:ptrdom}.
\end{proof}

\begin{observation}\label{obs:genlinimpliesupdate}
	Suppose $ex_1$ is an operation such that $pt(ex_1) < time(ex_1^{\ref{slupdate}})$. Then there exists an operation $ex_2$ such that $invoc(ex_2)$ dominates $invoc(ex_1)$ and $pt(ex_1) = pt(ex_2) = time(ex_2^{\ref{slupdate}})$.
\end{observation}

\begin{proof}
Since dominance is a strict partial order, there is a maximal element in the set of operations that linearize at $pt(ex_1)$.
That is, there exists an operation $ex_2$ with $pt(ex_2) = pt(ex_1)$, such that there is no operation $ex_3$ with $pt(ex_3) = pt(ex_1)$ and $invoc(ex_3)$ dominates $invoc(ex_2)$.
Therefore, $pt(ex_2)$ does not satisfy \ref{gen:ptrdom}.
Hence, $pt(ex_2)$ must satisfy \ref{gen:ptrself}, meaning $pt(ex_2) = time(ex_2^{\ref{slupdate}})$.
\end{proof}

\begin{lemma}\label{lemma:forcedom}
	Let $T \in \mathcal{T}$ be a transcript, where $G$ is the precedence graph induced by $\mathcal{G}(T)$.
	Suppose $ex_1, ex_2 \in \Gamma(T)$ are operations such that
	\begin{gather}
		\text{$oper(ex_1)$ and $oper(ex_2)$ are concurrent in $G$,}\label{assumconcur} \\
		\text{$invoc(ex_1)$ dominates $invoc(ex_2)$, and}\label{assumdom} \\
		\text{$ex_1 \xrightarrow{f(T)} ex_2$.}\label{assumprec}
	\end{gather}
	Then there exists an operation $ex_3 \in \Gamma(T)$ such that $invoc(ex_3)$ dominates $invoc(ex_1)$ and $oper(ex_3)$ precedes $oper(ex_2)$ in $G$.
\end{lemma}

\begin{proof}
	If $pt(ex_1) = pt(ex_2)$, then $ex_2 \xrightarrow{f(T)} ex_1$ by \ref{gen:ruledom}.
	This contradicts (\ref{assumprec}), so $pt(ex_1) \neq pt(ex_2)$.
	If $pt(ex_1) > pt(ex_2)$, then $ex_2 \xrightarrow{f(T)} ex_1$ by \ref{gen:rulelin}.
	Again this contradicts (\ref{assumprec}), so
	\begin{equation}
		pt(ex_1) < pt(ex_2).\label{lemma:forcedom:ptless}
	\end{equation}
	By (\ref{assumconcur}) and Observation~\ref{obs:opconcurrent}
	\begin{gather}
		time(ex_1^{\ref{scan}}) < time(ex_2^{\ref{slupdate}})\text{, and}\label{lemma:forcedom:i1less} \\
		time(ex_2^{\ref{scan}}) < time(ex_1^{\ref{slupdate}}).\label{lemma:forcedom:i2less}
	\end{gather}
	Suppose that $time(ex_2^{\ref{scan}}) < pt(ex_1)$.
	Using this and (\ref{assumdom}), $pt(ex_2) \leq pt(ex_1)$ by Observation~\ref{obs:genptltdom}.
	This contradicts (\ref{lemma:forcedom:ptless}).
	Hence,
	\begin{equation}
		pt(ex_1) < time(ex_2^{\ref{scan}}).\label{lemma:forcedom:pti1less}
	\end{equation}
	By (\ref{lemma:forcedom:i2less}), (\ref{lemma:forcedom:pti1less}), and Observation~\ref{obs:genlinimpliesupdate}, there exists an operation $ex_3$ such that
	\begin{gather}
		\text{$invoc(ex_3)$ dominates $invoc(ex_1)$, and}\label{lemma:forcedom:invoc3dom} \\
		\text{$pt(ex_3) = pt(ex_1) = time(ex_3^{\ref{slupdate}})$.}\label{lemma:forcedom:invoc3pt}
	\end{gather}
	By (\ref{lemma:forcedom:pti1less}) and (\ref{lemma:forcedom:invoc3pt}), $time(ex_3^{\ref{slupdate}}) < time(ex_2^{\ref{scan}})$, which implies that $oper(ex_3)$ precedes $oper(ex_2)$ in $G$ by Observation~\ref{obs:opprecede}.
	This, combined with (\ref{lemma:forcedom:invoc3dom}) and (\ref{lemma:forcedom:invoc3pt}), imply the statement of the lemma.
\end{proof}

The key result of Aspnes and Herlihy \cite{generalwaitfree} is stated below (note that we have combined two results from \cite{generalwaitfree} --- specifically, Lemma~11 and Theorem~17):

\begin{lemma}\label{herl:linresult}
	For any transcript $T \in \mathcal{T}$, where $G$ is the precedence graph induced by $\mathcal{G}(T)$, any topological ordering of $lingraph(G)$ is a linearization of $\Gamma(T)$.
\end{lemma}

We proceed to show that $f$ is a linearization function by demonstrating that if $f(T) = ex_1 \circ \ldots \circ ex_k$, then $oper(ex_1) \circ \ldots \circ oper(ex_k)$ is a topological ordering of $lingraph(G)$, where $G$ is the precedence graph induced by $\mathcal{G}(T)$, for any transcript $T \in \mathcal{T}$. There are two facts that make our task nontrivial. First, $f(T)$ may contain operations whose constructed nodes are not present in $\mathcal{G}(T)$. We claim that these operations are overwritten before they are added to the shared precedence graph, and the responses that are eventually calculated for each of these operations are valid for their position in the linearization order. Second, it is not immediately apparent that the order of $oper(ex_1) \circ \ldots \circ oper(ex_k)$ satisfies the dominance order that is present in a topological ordering of $lingraph(G)$.


The following lemma is taken directly from Aspnes and Herlihy \cite{generalwaitfree}.

\begin{lemma}\label{herl:dombuffer}
	Let $T \in \mathcal{T}$ be a transcript, where $G$ is the precedence graph induced by $\mathcal{G}(T)$, and let $L_1 \circ oper(ex_1) \circ L_2$ be a topological sort of $lingraph(G)$. If there exists $oper(ex_2) \in L_2$ that is concurrent with $oper(ex_1)$ in $G$, and $oper(ex_1)$ dominates $oper(ex_2)$, then there is $oper(ex_3) \in L_2$ such that $invoc(ex_3)$ dominates $invoc(ex_1)$ and $oper(ex_3)$ precedes $oper(ex_2)$ in $G$.
\end{lemma}

\begin{lemma}\label{lemma:nodominversion}
	Let $T \in \mathcal{T}$ be a transcript, where $G$ is the precedence graph induced by $\mathcal{G}(T)$, and let $L$ be some topological ordering of $lingraph(G)$. There is no pair of operations $ex_1, ex_2 \in f(T)$ such that
	\begin{gather}
		ex_1 \xrightarrow{f(T)} ex_2,\label{linfirst} \\
		oper(ex_2) \xrightarrow{L} oper(ex_1)\text{, and}\label{firstintopsort} \\
		\text{$invoc(ex_1)$ dominates $invoc(ex_2)$.}\label{hdom}
	\end{gather}
\end{lemma}

\begin{proof}
	Suppose there are a pair of operations $ex_1, ex_2 \in f(T)$ satisfying (\ref{linfirst})-(\ref{hdom}), and let $ex_1$ be the first operation in $f(T)$ for which these properties are satisfied.
	That is,
	\ifspringer
	\begin{equation}
	\begin{split}
		&\text{there is no pair of operations $ex_1', ex_2' \in f(T)$} \\ &\text{that satisfy (\ref{linfirst})-(\ref{hdom}), with $ex_1' \xrightarrow{f(T)} ex_1$.}\label{ex1isfirst}
	\end{split}
	\end{equation}
	\else
	\begin{equation}
	\begin{split}
		\text{there is no pair of operations $ex_1', ex_2' \in f(T)$ that satisfy (\ref{linfirst})-(\ref{hdom}), with $ex_1' \xrightarrow{f(T)} ex_1$.}\label{ex1isfirst}
	\end{split}
	\end{equation}
	\fi
	If $oper(ex_2)$ precedes $oper(ex_1)$ in $G$, then $time(ex_2^{\ref{slupdate}}) < time(ex_1^{\ref{scan}})$ by Observation~\ref{obs:opprecede}.
	But by Observation~\ref{obs:genptin} this would imply that $pt(ex_2) < pt(ex_1)$, which would mean $ex_2 \xrightarrow{f(T)} ex_1$ by \ref{gen:rulelin}.
	This contradicts (\ref{linfirst}).
	Hence,
	\begin{equation}
		\text{$oper(ex_2)$ does not precede $oper(ex_1)$ in $G$.}\label{lemma:nodominversion:100}
	\end{equation}
	If $oper(ex_1)$ precedes $oper(ex_2)$ in $G$, then this contradicts (\ref{firstintopsort}), since $L$ is a topological sort of $lingraph(G)$.
	Therefore,
	\begin{equation}
		\text{$oper(ex_1)$ does not precede $oper(ex_2)$ in $G$.}\label{lemma:nodominversion:200}
	\end{equation}
	By (\ref{lemma:nodominversion:100}) and (\ref{lemma:nodominversion:200}), $oper(ex_1)$ and $oper(ex_2)$ are concurrent in $G$.
	By this, (\ref{linfirst}), (\ref{hdom}), and Lemma~\ref{lemma:forcedom} there exists an operation $ex_j$ such that
	\begin{gather}
		\text{$invoc(ex_j)$ dominates $invoc(ex_1)$,}\label{lemma:nodominversion:300} \\
		\text{$pt(ex_j) = pt(ex_1) = time(ex_j^{\ref{slupdate}})$, and}\label{lemma:nodominversion:400} \\
		\text{$oper(ex_j)$ precedes $oper(ex_2)$ in $G$.}\label{lemma:nodominversion:500}
	\end{gather}
	By (\ref{lemma:nodominversion:300}), (\ref{lemma:nodominversion:400}), and \ref{gen:ruledom},
	\begin{equation}
		ex_1 \xrightarrow{f(T)} ex_j\label{lemma:nodominversion:600}
	\end{equation}
	By (\ref{lemma:nodominversion:500}) and the fact that $L$ is a topological sort of $lingraph(G)$,
	\begin{equation}
		oper(ex_j) \xrightarrow{L} oper(ex_2).\label{lemma:nodominversion:700}
	\end{equation}
	By Observation~\ref{obs:genptin}, $time(ex_j^{\ref{scan}}) < pt(ex_j) \leq time(ex_j^{\ref{slupdate}})$ and $time(ex_1^{\ref{scan}}) < pt(ex_1) \leq time(ex_1^{\ref{slupdate}})$.
	This along with (\ref{lemma:nodominversion:400}) implies that $time(ex_j^{\ref{scan}}) < time(ex_1^{\ref{slupdate}})$ and $time(ex_1^{\ref{scan}}) < time(ex_j^{\ref{slupdate}})$.
	By this and Observation~\ref{obs:opconcurrent},
	\begin{equation}
		\text{$oper(ex_j)$ and $oper(ex_1)$ are concurrent in $G$.}\label{lemma:nodominversion:800}
	\end{equation}
	By (\ref{firstintopsort}), (\ref{lemma:nodominversion:700}), and the transitivity of happens-before order,
	\begin{equation}
		oper(ex_j) \xrightarrow{L} oper(ex_1).\label{lemma:nodominversion:900}
	\end{equation}
	By (\ref{lemma:nodominversion:300}), (\ref{lemma:nodominversion:800}), (\ref{lemma:nodominversion:900}), and Lemma~\ref{herl:dombuffer}, there exists an operation $oper(ex_i)$ such that
	\begin{gather}
		oper(ex_j) \xrightarrow{L} oper(ex_i) \xrightarrow{L} oper(ex_1), \label{ellprimedombuffer} \\
		\text{$invoc(ex_i)$ dominates $invoc(ex_j)$, and}\label{invocellprimedominates} \\
		\text{$oper(ex_i)$ precedes $oper(ex_1)$ in $G$.}\label{ellprimeprec}
	\end{gather}
	By (\ref{ellprimeprec}) and Observation~\ref{obs:opprecede}, $time(ex_i^{\ref{slupdate}}) < time(ex_1^{\ref{scan}})$.
	This along with Observation~\ref{obs:genptin} implies that $pt(ex_i) < pt(ex_1)$.
	Hence, by \ref{gen:rulelin},
	\begin{equation}
		ex_i \xrightarrow{f(T)} ex_1.\label{ellprimelinfirst}
	\end{equation}
	By (\ref{lemma:nodominversion:600}), (\ref{ellprimelinfirst}), and the transitivity of happens-before order,
	\begin{equation}
		ex_i \xrightarrow{f(T)} ex_j.\label{ellprimelinbeforeell}
	\end{equation}
	Together, (\ref{ellprimedombuffer}), (\ref{invocellprimedominates}), and (\ref{ellprimelinbeforeell}) imply that $ex_i$ and $ex_j$ satisfy (\ref{linfirst})-(\ref{hdom}). By (\ref{ellprimelinfirst}), this contradicts (\ref{ex1isfirst}). Hence, there is no pair of operations that satisfies (\ref{linfirst})-(\ref{hdom}).
\end{proof}

\begin{lemma}\label{lemma:genprefpres}
	The function $f$ is prefix-preserving.
\end{lemma}

\begin{proof}
	Let $T \in \mathcal{T}$ be a transcript.
	Consider each step $t$ of $T$, and some operation $ex \in \Gamma (T)$ by $p$. By Observation~\ref{obs:genptin} are two cases:
	\begin{enumerate}[labelindent=0pt,labelwidth=\widthof{\ref{def:interp2}},itemindent=1em,leftmargin=!,label=(\roman*)]
		\item Suppose $pt(ex) < time(ex^{\ref{slupdate}})$. Then by Observation~\ref{obs:genlinimpliesupdate}, there exists an operation $ex_\ell$ such that $pt(ex) = pt(ex_\ell) = time(ex_\ell^{\ref{slupdate}})$. Hence, $pt(ex) = t$ if and only if $time(ex_\ell^{\ref{slupdate}}) = t$.
		\item Suppose $pt(ex) = time(ex^{\ref{slupdate}})$. Then $pt(ex) = t$ if and only if $time(ex^{\ref{slupdate}}) = t$.
	\end{enumerate}
	Thus, at step $t$ it is entirely determined which operations linearize at $t$.
	That is, whether $t$ satisfies $t = pt(op)$ can be deduced solely by examining the step at time $t$ of $T$.
	By this and the fact that \ref{gen:ruledom} and \ref{gen:rulecomm} impose a total order on the set of operations that linearize at each point in time in $T$, $f$ is prefix-preserving.
\end{proof}

We now address the fact that $f(T)$ may contain operations that are not in the precedence graph $G$ induced by $\mathcal{G}(T)$.
We aim to show that if $f(T) = ex_1 \circ \ldots \circ ex_k$, then $oper(ex_1) \circ \ldots \circ oper(ex_k)$ is a topological ordering of $G$; this does not hold if $oper(ex_1) \circ \ldots \circ oper(ex_k)$ contains operations that are not present in $G$.
To resolve this issue, we define a notion of a ``completion'' of a transcript (recall that completion is only defined for histories), which we call $fill$.
For any transcript $T \in \mathcal{T}$, $fill(T)$ is constructed by allowing every incomplete operation present in $f(T)$ to finish (see below for a precise definition).
If $ex \in f(T)$ is an operation by $p$, and $ex$ is incomplete in $T$, then let the $p$-solo \emph{completion} of $ex$ in $T$ be the transcript $T_p$ such that $(T \circ T_p) \in \mathcal{T}$, $T_p$ contains only steps by $p$, and the final step of $T_p$ is $rsp(ex)$.
Note that the existence of $T_p$ is guaranteed by the fact that Algorithm~\ref{gen} is wait-free.
Additionally, since $ex \in f(T)$, $pt(ex) \neq \infty$, and therefore $ex^{\ref{scan}} \in T$ by Observation~\ref{obs:genptin}.
Since the response of $ex$ is entirely determined by the vector returned by $ex^{\ref{scan}}$ (by examination of the $execute$ method of Algorithm~\ref{gen}), the response of $ex$ is entirely determined in $T$ (that is, no future steps by any process can change the response of $ex$).
Therefore, if $(T \circ T') \in \mathcal{T}$ and $T'$ contains no steps by $p$, then $(T \circ T' \circ T_p) \in \mathcal{T}$.

For any transcript $T \in \mathcal{T}$, let $fill(T)$ be a transcript constructed from $T$ as follows:
\begin{enumerate}[labelindent=0pt,labelwidth=\widthof{\ref{def:interp2}},itemindent=1em,leftmargin=!,label=\textbf{F-\arabic*},ref={F-\arabic*}]
		\item Initially, let $T' = T$.
		\item For every operation $ex \in f(T)$ by process $p$ that is incomplete in $T$, append the $p$-solo completion of $ex$ in $T$ to $T'$. That is, if $T_p$ is the $p$-solo completion of $ex$ in $T$, then let $T' = T' \circ T_p$.\label{lemma:genlinstep2}
		\item Let $fill(T) = T'$.
\end{enumerate}
	
\begin{observation}\label{obs:filltranscript}
	For any $T \in \mathcal{T}$,
	\begin{enumerate}[label=(\alph*)]
		\item $f(T) = f\bigl(fill(T)\bigr)$, and\label{obs:filla}
		\item if $f(T) = ex_1 \circ \ldots \circ ex_k$, then the precedence graph $G'$ induced by $\mathcal{G}\bigl( fill(T)\bigr)$ contains all and only those operations in $\bigl\{oper(ex_1), \ldots, oper(ex_k)\bigr\}$.\label{obs:fillb}
	\end{enumerate}
\end{observation}

\begin{proof}
	Let $T \in \mathcal{T}$, and for ease of notation let $T' = fill(T)$.
	Since \ref{lemma:genlinstep2} only appends steps in the construction of $T'$, $T$ is a prefix of $T'$.
	By this and Lemma~\ref{lemma:genprefpres},
	\begin{equation}
		\text{$f(T)$ is a prefix of $f(T')$.}\label{genlin:prefix}
	\end{equation}
	Suppose $f(T) \neq f(T')$.
	By this and (\ref{genlin:prefix}), $f(T') = f(T) \circ ex_{\ell_1} \circ \ldots \circ ex_{\ell_k}$, for some nonempty sequence of operations $ex_{\ell_1}, \ldots, ex_{\ell_k}$ with 
	\begin{equation}
		\text{$ex_{\ell_i} \not\in f(T)$ for all $i \in \{1, \ldots, k\}$.}\label{genlin:exnotin}
	\end{equation}
	Due to (\ref{genlin:exnotin}) and the fact that \ref{lemma:genlinstep2} only adds steps of operations in $f(T)$,
	\ifspringer
	\begin{equation}
	\begin{split}
	&\text{no steps of $ex_{\ell_i}$ are added in the construction} \\ &\text{of $T'$, for any $i \in \{1, \ldots, k\}$.}\label{genlin:nosteps}
	\end{split}
	\end{equation}
	\else
	\begin{equation}
	\text{no steps of $ex_{\ell_i}$ are added in the construction of $T'$, for any $i \in \{1, \ldots, k\}$.}\label{genlin:nosteps}
	\end{equation}
	\fi
	Suppose $pt_{T'}(ex_{\ell_1}) = time_{T'}(ex_{\ell_1}^{\ref{slupdate}})$.
	Since $ex_{\ell_1} \in f(T')$, it must be the case that $pt_{T'}(ex_{\ell_1}) \neq \infty$, and therefore $time_{T'}(ex_{\ell_1}^{\ref{slupdate}}) \neq \infty$.
	By this and (\ref{genlin:nosteps}), $ex_{\ell_1}^{\ref{slupdate}} \in T$.
	Then $ex_{\ell_1} \in f(T)$, which is a contradiction.
	Therefore, $pt_{T'}(ex_{\ell_1}) < time_{T'}(ex_{\ell_1}^{\ref{slupdate}})$.
	Then by Observation~\ref{obs:genlinimpliesupdate} there exists an operation $ex_\alpha \in \Gamma(T')$ such that
	\begin{gather} 
		\small\text{$pt_{T'}(ex_{\ell_1}) = pt_{T'}(ex_\alpha) = time_{T'}(ex_\alpha^{\ref{slupdate}})$ in $T'$, and}\label{genlin:ptex} \\
		\text{$invoc(ex_\alpha)$ dominates $invoc(ex_{\ell_1})$.}\label{genlin:invocex}
	\end{gather}
	By (\ref{genlin:ptex}), (\ref{genlin:invocex}), and \ref{gen:ruledom}, 
	\begin{equation}
		\text{$ex_{\ell_1} \xrightarrow{f(T')} ex_\alpha$.}\label{genlin:onebeforeell}
	\end{equation}
	Since $pt_{T'}(ex_{\ell_1}) \neq \infty$, $time_{T'}(ex_\alpha^{\ref{slupdate}}) \neq \infty$ by (\ref{genlin:ptex}).
	This, along with (\ref{genlin:nosteps}), implies that $ex_\alpha^{\ref{slupdate}} \in T$.
	Hence, $ex_\alpha \in f(T)$.
	This, along with (\ref{genlin:prefix}) and (\ref{genlin:onebeforeell}) imply that $ex_{\ell_1} \in f(T)$, which is a contradiction.
	Therefore, we have arrived at a contradiction in all cases, and our initial supposition is false.
	That is, $f(T) = f(T')$, which concludes the proof of part~\ref{obs:filla}.
	
	Let $G'$ be the precedence graph induced by $\mathcal{G}(T')$.
	During the construction of $T'$, \ref{lemma:genlinstep2} ensures that, for every operation $ex \in f(T)$, $ex$ is complete.
	In particular,
	\begin{equation}
		\text{for every operation $ex \in f(T)$, $ex^{\ref{slupdate}} \in T'$.}\label{genlin:everyopin}
	\end{equation}
	By definition, $\mathcal{G}(T')$ contains every $node(ex)$ such that $ex^{\ref{slupdate}} \in T'$, and therefore $G'$ contains every $oper(ex)$ such that $ex^{\ref{slupdate}} \in T'$.
	This, together with (\ref{genlin:everyopin}), implies part~\ref{obs:fillb}.
\end{proof}

\begin{lemma}\label{lemma:generallinearizable}
	For any transcript $T \in \mathcal{T}$, $f(T)$ is a linearization of $\Gamma(T)$.
\end{lemma}

\begin{proof}
	Let $T' = fill(T)$, and let $G'$ be the precedence graph induced by $\mathcal{G}(T')$.
	Suppose $f(T') = ex_1 \circ \ldots \circ ex_k$.
	The only difference between the sequences $ex_1, \ldots, ex_k$ and $oper(ex_1), \ldots, oper(ex_k)$ are operation identifiers, so $f(T')$ is equivalent to the history $oper(ex_1) \circ \ldots \circ oper(ex_k)$.
	Hence, if we prove that $oper(ex_1) \circ \ldots \circ oper(ex_k)$ is a linearization of $\Gamma(T')$, this implies that $f(T')$ is also a linearization of $\Gamma(T')$.
	To accomplish this, it suffices to demonstrate that $oper(ex_1) \circ \ldots \circ oper(ex_k)$ is a topological ordering of $lingraph(G')$ by Lemma~\ref{herl:linresult}.
	More specifically, we show that the sequential history $oper(ex_1) \circ \ldots \circ oper(ex_k)$ has no ``back-edges''; that is, there are no edges $\bigl(oper(ex_j), oper(ex_i)\bigr)$ in $lingraph(G')$ such that $j > i$.

	First note that by Observation~\ref{obs:filltranscript}~\ref{obs:fillb}, $G'$ contains all and only those operations $oper(ex_i)$ such that $ex_i \in f(T)$.
	Since $f(T) = f(T')$ by Observation~\ref{obs:filltranscript}~\ref{obs:filla}, this implies that 
	\begin{equation}
		\text{$ex \in f(T')$ if and only if $oper(ex)$ is in $G'$.}\label{genlin:sameops}
	\end{equation}
	Let $L$ be a topological ordering of $lingraph(G')$.
	Let $i, j \in \{1, \ldots, k\}$ with $i < j$.
	Hence,
	\begin{equation}
		ex_i \xrightarrow{f(T')} ex_j.\label{genlin:funcorder}
	\end{equation}
	By (\ref{genlin:sameops}), $oper(ex_i), oper(ex_j)$ are in $G'$, and hence $oper(ex_i), oper(ex_j)$ are in $lingraph(G')$.
	Suppose that
	\ifspringer
	\begin{equation}
	\begin{split}
		&\text{there exists an edge from $oper(ex_j)$ to} \\ &\text{$oper(ex_i)$ in $lingraph(G')$.}\label{genlin:existedge}
	\end{split}
	\end{equation}
	\else
	\begin{equation}
		\text{there exists an edge from $oper(ex_j)$ to $oper(ex_i)$ in $lingraph(G')$.}\label{genlin:existedge}
	\end{equation}
	\fi
	Suppose $oper(ex_j)$ precedes $oper(ex_i)$ in $G'$.
	By Observation~\ref{obs:opprecede}, $time(ex_j^{\ref{slupdate}}) < time(ex_i^{\ref{scan}})$.
	This implies that $pt(ex_j) < pt(ex_i)$ by Observation~\ref{obs:genptin}.
	By \ref{gen:rulelin}, $ex_j \xrightarrow{f(T')} ex_i$, which contradicts (\ref{genlin:funcorder}).
	Therefore,
	\begin{equation}
		\text{$oper(ex_j)$ does not precede $oper(ex_i)$ in $G'$.}\label{genlin:notprecede}
	\end{equation}
	By (\ref{genlin:existedge}) and the fact that $L$ is a topological ordering of $lingraph(G')$,
	\begin{equation}
		oper(ex_j) \xrightarrow{L} oper(ex_i).\label{genlin:precinL}
	\end{equation}
	Due to (\ref{genlin:existedge}) and (\ref{genlin:notprecede}), there must be a dominance edge from $oper(ex_j)$ to $oper(ex_i)$ in $lingraph(G')$.
	That is,
	\begin{equation}
		\text{$invoc(ex_i)$ dominates $invoc(ex_j)$.}\label{genlin:dom}
	\end{equation}
	But (\ref{genlin:funcorder}), (\ref{genlin:precinL}), and (\ref{genlin:dom}) contradict Lemma~\ref{lemma:nodominversion}.
	Hence, no ``back-edges'' exist in $lingraph(G')$ among the sequential history $oper(ex_1) \circ \ldots \circ oper(ex_k)$.
	Along with (\ref{genlin:sameops}), this implies that $oper(ex_1) \circ \ldots \circ oper(ex_k)$ is a topological ordering of $lingraph(G')$.
	By Lemma~\ref{herl:linresult}, $oper(ex_1) \circ \ldots \circ oper(ex_k)$ is a linearization of $\Gamma(T')$.
	Since $oper(ex_1) \circ \ldots \circ oper(ex_k)$ and $f(T')$ are equivalent, $f(T')$ is also a linearization of $\Gamma(T')$.
	By Observation~\ref{obs:filltranscript}~\ref{obs:filla}, 
	\begin{equation}
		\text{$f(T)$ is a linearization of $\Gamma(T')$.}\label{genlin:fTislin}
	\end{equation}
	
	By construction of $T'$, every operation $ex \in f(T)$ is present and complete in $T'$ (and in $\Gamma(T')$), and every complete operation $ex \in \Gamma(T')$ is in $f(T)$ by Observation~\ref{obs:genptin}.
	That is, 
	\begin{equation}
		\text{$ex \in \Gamma(T')$ is incomplete if and only if $ex \not\in f(T)$.}\label{genlin:incifonlyif}
	\end{equation}
	Let $H_C$ be the completion of $\Gamma(T')$ obtained by removing all incomplete operations from $\Gamma(T')$.
	By (\ref{genlin:incifonlyif}), 
	\begin{equation}
		\text{$H_C$ and $f(T)$ contain precisely the same operations.}\label{genlin:hcfTsame}
	\end{equation}
	Since $H_C$ is a completion of $\Gamma(T')$, $ex_1 \xrightarrow{H_C} ex_2$ implies that $ex_1 \xrightarrow{\Gamma(T')} ex_2$.
	By this, (\ref{genlin:fTislin}), and (\ref{genlin:hcfTsame}), $f(T)$ is a linearization of $H_C$.
	Since the construction of $T'$ adds no new ``high-level'' invocations (i.e. invocations that are not present in $\Gamma(T)$) to $\Gamma(T')$, $H_C$ is also a completion of $\Gamma(T)$.
	Hence, since $f(T)$ is a linearization of $H_C$, $f(T)$ is a linearization of $\Gamma(T)$.
\end{proof}

\begin{theorem}\label{theorem:genstronglin}
	Algorithm~\ref{gen} is strongly linearizable.
\end{theorem}

\begin{proof}
	Lemma~\ref{lemma:generallinearizable} shows that $f$ is a linearization function for $\Gamma(\mathcal{T})$, and Lemma~\ref{lemma:genprefpres} shows that $f$ is prefix preserving. Hence, $f$ is a strong linearization function for $\mathcal{T}$.
\end{proof}

Theorems \ref{thm:main-snapshot} and \ref{theorem:genstronglin}, along with the fact that strong linearizability is composable, yield Theorem~\ref{thm:main-general-construction}.

\fi

\ifthesis
\else
\subsection{Remarks}
\fi

In this section, we proved that Aspnes and Herlihy's general construction for simple types \cite{generalwaitfree} is strongly linearizable.
Typically, proving strong linearizability is only marginally more difficult than proving linearizability; as long as one chooses linearization points for each operation carefully, it is easy to show that a linearization function satisfies the prefix-preservation property.
With this claim in mind, the length of this section might be perplexing, since Aspnes and Herlihy have already shown that Algorithm~\ref{gen} is linearizable \cite{generalwaitfree}.
However, their linearization function is inherently not prefix-preserving, since operations may be written to the ``middle'' of the linearization graph (as discussed at the beginning of Section~\ref{sec:genproofsl}).
Hence, in Section~\ref{sec:genproofsl} we were forced to start from scratch, defining our own linearization function that satisfies the prefix-preservation property.
A considerable amount of effort was also spent on proving that the shared snapshot object $root$ always contains a representation of a particular precedence graph.
This is an invariant that was taken for granted in \cite{generalwaitfree}; for the sake of completeness, we decided to prove it formally in Section~\ref{sec:storeprecgraph}.

We also note that Aspnes and Herlihy's construction requires unbounded memory, since operations are never removed from the shared precedence graph.
This is unfortunate, since the general construction does not benefit from the fact that our snapshot implementation requires only bounded space.
Aspnes and Herlihy claim that ``for any particular data type, it should be possible to apply type-specific optimizations to discard most of the precedence graph'' \cite{generalwaitfree}.
It would be interesting to study types for which such a simplification could be used to bound the space of Aspnes and Herlihy's construction.
Moreover, it may be possible to adjust Aspnes and Herlihy's algorithm so that it only requires bounded space for any simple type.
For instance, this might be achieved by storing states, rather than operations, inside the nodes of a precedence graph.

\ifthesis
\else
\section{Discussion}
\fi


We have provided a lock-free strongly linearizable implementation of a snapshot object using atomic multi-reader multi-writer registers as base objects. We used this implementation to demonstrate that any simple object also has a lock-free strongly linearizable implementation from atomic multi-reader multi-writer registers. The class of simple types seems to be a large subset of the types that have wait-free linearizable implementations from atomic registers. We are not aware of any classifications of non-simple types that are also known to have wait-free linearizable implementations from registers.
Additionally, we have not yet explored the characteristics of objects that enable wait-free strongly linearizable implementations from registers.

As we briefly discussed in Section~\ref{sectionGEN}, the general construction defined by Aspnes and Herlihy \cite{generalwaitfree} uses unbounded space. This is a result of the fact that each operation constructs a new node, and nodes added to the shared precedence graph are never reclaimed. The storage of unbounded precedence graphs also affects the liveness properties satisfied by the algorithm; while Algorithm~\ref{gen} is wait-free, it is not \emph{bounded} wait-free. That is, there is no constant that bounds the number of steps required by any operation. This is because each operation must calculate a linearization by topologically ordering an ever-expanding precedence graph. It would be interesting to know if this construction could be bounded, for example, by ``pruning'' the precedence graph at certain stages of the algorithm. This might be accomplished by storing \emph{states}, rather than operations, in each entry of the shared snapshot object.

Regarding strong linearizability, little is known about the power of primitives with higher consensus numbers than registers and snapshots.
As mentioned in Section~\ref{chapter:intro}, Golab, Higham, and Woelfel showed that standard universal constructions using consensus objects are strongly linearizable \cite{stronglin}.
Therefore, it is possible to develop wait-free strongly linearizable implementations of any type in systems that have access to atomic compare-and-swap (CAS) or load-linked/store-conditional (LL/SC) objects.
We would like to know if there are efficient strongly linearizable implementations of useful types from such powerful base objects.
Perhaps many existing implementations of types from CAS or LL/SC objects are already strongly linearizable; in this case, it would be interesting to identify such implementations and prove that they are strongly linearizable.

Attiya, Casta\~{n}eda, and Hendler showed that a wait-free strongly linearizable implementation of an $n$-process queue or stack can be used to solve $n$-consensus \cite{ACH2018a}.
Their proof is quite simple; suppose we have access to a wait-free strongly linearizable queue.
To solve consensus, each process first writes its value to a single-writer register, then enqueues its identifier to the queue, and finally takes a snapshot of the shared memory locations used by the queue to obtain a local copy of the object.
Following this, a process simulates a dequeue on its local copy of the queue to obtain the process identifier $p$, and then decides on the value proposed by $p$.
Since the queue is strongly linearizable, at some point every process agrees on the ``head'' of the queue, which implies that each process obtains the same identifier from its dequeue operation.
This result immediately implies that there is no strongly linearizable implementation of an $n$-process queue or stack using base objects with consensus number less than $n$.
We would like to know if a similar result holds for other types with consensus number 2, such as read-modify-write types.

Our research was initially motivated by the following conjecture: every type that has a wait-free linearizable implementation from atomic registers has a wait-free strongly linearizable implementation from atomic snapshot objects. 
\ifthesis
\else\iffull
	While we have not definitively proven or disproven this conjecture, it seems that the class of simple types is a large subset of the set of types that have wait-free linearizable implementations from registers.
	A natural extension of this work would aim to discover non-simple types that have wait-free linearizable implementations from registers; if no such types exist, then combined with Theorem~\ref{thm:main-general-construction} this would prove our conjecture.
	\fi
\fi

\newpage
\bibliography{mainbib}

\end{document}